\documentclass[11pt]{article}
\usepackage{fullpage}
\usepackage{kpfonts}
\usepackage[utf8]{inputenc}
\usepackage{graphicx}
\usepackage[utf8]{inputenc}
\usepackage[T1]{fontenc}
\usepackage{sidecap}

\usepackage[numbers,sort&compress]{natbib}
\usepackage[sectionbib]{bibunits}
\defaultbibliographystyle{plainnat}
\usepackage{enumerate}
\usepackage[shortlabels]{enumitem}
\usepackage{booktabs}
\usepackage{ragged2e}

\usepackage{array}
\newcolumntype{P}[1]{>{\centering\arraybackslash}m{#1}}
\newcolumntype{U}[1]{>{\centering\arraybackslash}p{#1}}
\newcolumntype{L}[1]{>{\centering\arraybackslash}l{#1}}
\renewcommand{\arraystretch}{1.4}

\usepackage{xr-hyper}
\usepackage[colorlinks=true, citecolor=blue,linkcolor=blue]{hyperref}
\usepackage[utf8]{inputenc}
\usepackage{amsmath, amsthm, amssymb, color,colortbl}
\usepackage{physics, cleveref}
\usepackage{graphicx}
\usepackage{float}
\usepackage{cancel}
\usepackage{appendix}
\usepackage{algorithm}
\usepackage[noend]{algpseudocode}
\usepackage{import}
\usepackage{relsize}
\usepackage[all]{nowidow}
\usepackage[caption=false]{subfig}
\usepackage{thmtools, thm-restate}
\usepackage{tikz}
\usetikzlibrary{math, shapes.misc, decorations.pathmorphing,snakes}
\usepackage{notes2bib}
\usepackage{appendix}

\def\autorefapp#1{\hyperref[#1]{Appendix~\ref{#1}}}

\newcommand{\killpunct}[1]{}

\newcommand{\NP}{\textsf{NP}{}}
\newcommand{\PH}{\textsf{PH}{}}

\renewcommand{\bra}[1]{\langle #1 |}
\renewcommand{\ket}[1]{|#1\rangle}

\renewcommand{\braket}[2]{\langle#1 |  #2\rangle}
\def\ketbra#1{ |{#1}\rangle\!\langle{#1}| }
\def\ketAbraB#1#2{|{#1}\rangle\!\langle{#2}| }

\DeclareMathOperator*{\EV}{\mathbb{E}}

\newcommand{\pwhitenoise}{p_{\text{wn}}}
\newcommand{\pideal}{p_{\text{ideal}}}
\newcommand{\pnoisy}{p_{\text{noisy}}}
\newcommand{\punif}{p_{\text{unif}}}

\newtheorem{lemma}{Lemma}

\newtheorem{definition}{Definition}
\newtheorem{corollary}{Corollary}
\newtheorem{theorem}{Theorem}
\newtheorem{conjecture}{Conjecture}

\newtheorem{innercustomthm}{}
\newenvironment{customthm}[1]
  {\renewcommand\theinnercustomthm{#1}\innercustomthm}
  {\endinnercustomthm}

\DeclareMathOperator{\poly}{poly}

\definecolor{orange}{rgb}{1,0.5,0}
\definecolor{purple}{rgb}{0.5,0,0.5}
\definecolor{dark green}{rgb}{0,0.4,0}

\begin{document}

\title{Random quantum circuits transform local noise into global white noise}

\author{Alexander M. Dalzell\thanks{Institute for Quantum Information and Matter, Caltech, Pasadena, CA 91125\newline
\hspace*{16pt}AWS Center for Quantum Computing, Pasadena, CA 91125\newline
\indent \indent \indent \smaller{This work was done prior to AD joining the AWS Center for Quantum Computing.}}
\and Nicholas Hunter-Jones 
\thanks{Stanford Institute for Theoretical Physics, Stanford, CA 94305\newline
\hspace*{16pt}Perimeter Institute for Theoretical Physics, Waterloo, ON N2L 2Y5}
\and Fernando G. S. L. Brand\~ao\thanks{Institute for Quantum Information and Matter, Caltech, Pasadena, CA 91125\newline \hspace*{16pt}AWS Center for Quantum Computing, Pasadena, CA 91125}
}

\date{}

\maketitle

\begin{abstract}
  
  We study the distribution over measurement outcomes of noisy random quantum circuits in the low-fidelity regime. We show that, for local noise that is sufficiently weak and unital, correlations (measured by the linear cross-entropy benchmark) between the output distribution $\pnoisy$ of a generic noisy circuit instance and the output distribution $\pideal$ of the corresponding noiseless instance shrink exponentially with the expected number of gate-level errors, as $F=\text{exp}(-2s\epsilon \pm O(s\epsilon^2))$, where $\epsilon$ is the probability of error per circuit location and $s$ is the number of two-qubit gates. Furthermore, if the noise is incoherent, the output distribution approaches the uniform distribution $\punif$ at precisely the same rate and can be approximated as $\pnoisy \approx F\pideal + (1-F)\punif$---that is, local errors are scrambled by the random quantum circuit and contribute only white noise (uniform output). Importantly, we upper bound the total variation error (averaged over random circuit instance) in this approximation as $O(F\epsilon \sqrt{s})$, so the ``white-noise approximation'' is meaningful when $\epsilon \sqrt{s} \ll 1$, a quadratically weaker condition than the $\epsilon s\ll 1$ requirement to maintain high fidelity. The bound applies when the circuit size satisfies $s \geq \Omega(n\log(n))$, 
  which corresponds to only {\it logarithmic depth} circuits,
  and the inverse error rate satisfies $\epsilon^{-1} \geq \tilde{\Omega}(n)$, which is needed to ensure errors are scrambled faster than $F$ decays.  The white-noise approximation is useful for salvaging the signal from a noisy quantum computation; for example, it was an underlying assumption in complexity-theoretic arguments that noisy random quantum circuits cannot be efficiently sampled classically, even when the fidelity is low. Our method is based on a map from second-moment quantities in random quantum circuits to expectation values of certain stochastic processes for which we compute upper and lower bounds. 

\end{abstract}

\section{Introduction}

There is a fundamental trade-off in quantum computation between computation size and error rate. Naturally, the longer the computation, the lower the physical error rate must be to maintain a high probability of an errorless computation. Once the error rate is beneath a constant threshold, the theory of fault tolerance and quantum error correction \cite{Shor1996FaultTolerantQC,Aharonov2008FaultTolerant} may be employed to push the probability of a \textit{logical} error arbitrarily close to zero, despite the prevalence of many physical errors during the computation; however, error correction comes at the cost of additional qubits and gates. These overheads, while acceptable in an asymptotic sense, are likely to be overwhelming in the near and intermediate term. This inspires the idea of an upcoming Noisy Intermediate-Scale Quantum (NISQ) era \cite{Preskill2018NISQ}, where hardware capabilities are good enough to perform non-trivial quantum tasks on dozens or hundreds of qubits, but quantum error correction, which might require thousands or millions of qubits, remains beyond reach. 

In this paper, we study a model of NISQ devices performing random computations and prove a precise sense in which, for typical circuit instances, local errors are quickly scrambled and can be treated as white noise. For some applications, this phenomenon makes it possible for the signal of the noiseless computation to be extracted by repetition despite a large overall chance that at least one error occurs. 

Our local error model assumes that each two-qubit gate in the quantum circuit is followed by a pair of gate-independent single-qubit unital noise channels acting on the two qubits involved in the gate. For simplicity and ease of analysis, we assume each of these noise channels is identical, but we fully expect the takeaways from our work to apply when the noise strength is allowed to vary from location to location.  
For concreteness in this introduction, we can consider the depolarizing channel with error probability $\epsilon$.
In this case, the fidelity of the noisy computation with respect to the ideal computation is expected to be roughly equal to the probability that no errors occur. We see that, for a circuit with $s$ two-qubit gates, this quantity, denoted here by $F = (1-\epsilon)^{2s}$, is close to 1 only if the quantity $2\epsilon s$---the average number of errors---satisfies $2\epsilon s \ll 1$. 

However, this high-fidelity requirement is quite restrictive in practice. Already for circuits with 50 qubits at depth 20, the error rate $\epsilon$ must be on the order of $10^{-4}$ for the whole computation to run without error at least 90\% of the time; this error rate is more than an order of magnitude smaller than what has been achievable in recent experiments on superconducting qubit systems of that size \cite{Arute2019GoogleQuantumSupremacy,USTC2021StrongQCompAdv,USTC2021Zhuchongzhi2.1}.  Indeed, in their landmark 2019 quantum computational supremacy experiment \cite{Arute2019GoogleQuantumSupremacy}, a group at Google performed random circuits on 53 qubits of depth 20, but the fidelity of the computation was $F \approx 0.002$, meaning at least one error occurs in all but a tiny fraction of the trials. Similar experiments at the University of Science and Technology of China on 56 \cite{USTC2021StrongQCompAdv} and 60 \cite{USTC2021Zhuchongzhi2.1} qubits reported even smaller fidelities of $0.0007$ \cite{USTC2021StrongQCompAdv} and $0.0004$ \cite{USTC2021Zhuchongzhi2.1}.  This would not be an issue if one could determine when a trial is errorless. (In this case, one could just repeat the experiment $1/F$ times.) However, error-detection requires overheads similar to error-correction. 

Rather, low-fidelity random circuit sampling experiments and their claim of quantum computational supremacy benefit from a key assumption \cite{Boixo2018CharacterizingNearTerm,Arute2019GoogleQuantumSupremacy}: when at least one error does occur, the output of the experiment is well approximated by \emph{white noise}, that is, the output is random and uncorrelated with the ideal (noiseless) output. When this is the case, the signal of diminished size $F$ can, at least for some applications, be extracted from the white noise using $O(1/F^2)$ trials, as we explain later. Specifically, for quantum computational supremacy, the \textit{white-noise assumption} is that the  distribution $\pnoisy$ over measurement outcomes of their noisy device is close to what we call the ``white-noise distribution'' 
\begin{equation}\label{eq:whitenoisedefinition}
    \pwhitenoise = F \pideal + (1-F)\punif\,,
\end{equation} 
with $\pideal$ the ideal distribution and $\punif$ the uniform\footnote{In Google's experiment, there was biased noise during readout (they measure $\ket{0}$ more often than $\ket{1}$) that would lead the appropriate definition of white noise to be slightly non-uniform (see Supplementary Material of \cite{Arute2019GoogleQuantumSupremacy}). We believe most of our analysis could be straightforwardly generalized to account for this kind of end-of-circuit non-unital error (although mid-circuit non-unital errors would likely complicate our method). However, the goal of our work is to study the complexity and behavior of low-fidelity random circuit experiments in an idealized sense, rather than the actual implementation of such ideas in recent superconducting experiments specifically. } distribution. In particular, for the approximation to be non-trivial, we demand that the total variation distance between $\pnoisy$ and $\pwhitenoise$ be a small fraction of $F$, that is
\begin{equation}\label{eq:whitenoisecondition}
     \hspace{72 pt}\frac{1}{2}\lVert \pwhitenoise-\pnoisy \rVert_1 \ll F \,.  \qquad \text{(white-noise assumption)}
\end{equation}
This demand is necessary because we expect that $\pnoisy$ also decays toward $\punif$ such that $\frac{1}{2}\lVert \pnoisy - \punif \rVert_1 = \Theta(F)$, and thus $\punif$ is a trivial approximation for $\pnoisy$ with error $\Theta(F)$.

Prior to their experiment, the Google group provided numerical evidence \cite{Boixo2018CharacterizingNearTerm} in favor of the white-noise assumption\footnote{Note that Ref.~\cite{Boixo2018CharacterizingNearTerm} proposed the stronger ansatz that the output quantum state is a combination of the ideal output state and the maximally mixed state, which implies (but is not necessary for) the statement $\pnoisy \approx \pwhitenoise$ about classical probability distributions over measurement outcomes.}
for randomly chosen circuits by showing that the output distribution of random circuits of depth 40 on 20 qubits (arranged in a 2D lattice) subject to a local Pauli error model approaches the uniform distribution, and that the fidelity of $\pnoisy$ with respect to $\pideal$ appears to decay exponentially, consistent with $\pnoisy \approx \pwhitenoise$. However, their analysis
did not specifically estimate the distance\footnote{Ref.~\cite{Boixo2018CharacterizingNearTerm} did not specifically formulate the assumption as in Eq.~\eqref{eq:whitenoisecondition}, where we demand that the allowed approximation error decrease with the fidelity, but we argue that the approximation is only meaningful when this is true. For example, in \autorefapp{app:complexitytheorywhitenoise} we argue that such precision is necessary to make a stronger complexity-theoretic argument for quantum computational supremacy.}  between $\pnoisy$ and $\pwhitenoise$.
The white-noise condition in Eq.~\eqref{eq:whitenoisecondition} requires that the distance between $\pnoisy$ and $\pwhitenoise$ decrease as the expected number of errors increases and $F$ decays, so quantifying the differences between the distributions is vital for determining how well the white-noise approximation is obeyed.

Here we prove rigorous bounds on the error in the white-noise approximation, averaged over circuits with randomly chosen gates. Our results fully apply in two random quantum circuit architectures: first, the 1D architecture with periodic boundary conditions, where qubits are arranged in a ring and alternating layers of nearest-neighbor gates are applied; and second, the complete-graph architecture, where each gate is chosen to act on a pair of qubits chosen uniformly at random among all $n(n-1)/2$ pairs.\footnote{Additionally, our results would fully apply to architectures in $D$ spatial dimensions for any $D$ under a conjecture from Ref.~\cite{dalzell2020anticoncentration} that these architectures anti-concentrate in $O(\log(n))$ depth. Without that conjecture, a weaker result is shown.} We show that, for Pauli noise channels, the error in the white-noise approximation is small as long as (1) $\epsilon^2 s \ll 1$, (2) $s \geq \Omega(n\log(n))$, and (3) $\epsilon \ll 1/(n\log(n)) $. We believe that condition (3) could be relaxed to read $\epsilon < c/n$ for some universal constant $c = O(1)$ (numerics suggest $c=0.3$ for the complete-graph architecture). Condition (1) is a quadratic improvement over the condition $\epsilon s \ll 1$ needed for high fidelity. For circuits with $\epsilon < 0.005$, as is the case in recent experiments \cite{Arute2019GoogleQuantumSupremacy,USTC2021StrongQCompAdv,USTC2021Zhuchongzhi2.1}, thousands of gates could potentially be implemented before condition (1) fails. Note that our technical statements hold for general (non-Pauli) error channels as well, but we find that the error in the white-noise approximation is small only for incoherent noise channels. We complement this analysis with numerical results that confirm the picture presented by our theoretical proofs for the complete-graph architecture, and demonstrate that realistic NISQ-era values of the error rate and circuit size can lead to a good white-noise approximation. 

By putting the white-noise approximation for random quantum circuits on stronger theoretical footing, our work has several applications. First, the white-noise assumption is an ingredient in formal complexity-theoretic arguments that the task accomplished on noisy devices running random quantum circuits is hard for classical computers (allowing the declaration of quantum computational supremacy) \cite{Arute2019GoogleQuantumSupremacy}. We complement our main result by showing in \autoref{app:complexitytheorywhitenoise} that classically sampling from the white-noise distribution within total variation distance $\eta F$ is, in a certain complexity-theoretic sense, equivalent up to a factor of $F$ (which is optimal) to sampling from the ideal output distribution within total variation distance $O(\eta)$. This makes low-fidelity experiments where errors are common nearly as defensible for quantum computational supremacy as high-fidelity experiments where errors are rare, at least in principle. Second, our result lends theoretical justification to the usage \cite{Arute2019GoogleQuantumSupremacy, USTC2021StrongQCompAdv,USTC2021Zhuchongzhi2.1} of the linear cross-entropy metric proposed in Ref.~\cite{Arute2019GoogleQuantumSupremacy} to benchmark noise in random circuit experiments and verify that hardware has correctly performed the quantum computational supremacy task. Indeed, as a side result, we show that, for both incoherent and coherent noise, the metric decays precisely as $e^{-2s\epsilon \pm O(s\epsilon^2)}$ when $\epsilon$ is sufficiently small; this also suggests that the linear cross entropy benchmark could be reliably used to accurately estimate the underlying local noise rate $\epsilon$ \cite{Liu2021BenchmarkingRCS}. 

Beyond random circuit experiments for quantum computational supremacy, our work suggests that other scenarios where the white-noise assumption holds may be advantageous in the NISQ era, as one can eschew error-correction and nonetheless perform a fairly long quantum computation, as long as one is willing to repeat the experiment $O(1/F^2)$ times. One example of a scenario where the assumption may hold is quantum simulation of fixed chaotic Hamiltonians, since they are also believed to be efficient at scrambling errors. 

The remainder of the paper is structured as follows: in \autoref{sec:noisemodel}, we describe our setup and in particular our model for local noise within a random quantum circuit; in \autoref{sec:whitenoiseresults}, we precisely state our results; in \autoref{sec:whitenoiseimplications}, we discuss further implications and how our results fit in with prior work; in \autoref{sec:whitenoisemethodandintuition}, we give an overview of the intuition behind our result and the method we use in our proofs, which is based on a map from random quantum circuits to certain stochastic processes, which can also be interpreted as partition functions of statistical mechanical systems. This method might be regarded as an extension of the method in Ref.~\cite{dalzell2020anticoncentration}, where we studied anti-concentration in random quantum circuits. In \autoref{sec:numerics}, we present a numerical calculation of our bound for the realistic values of the circuit parameters informed by the experiments in Refs.~\cite{Arute2019GoogleQuantumSupremacy,USTC2021StrongQCompAdv,USTC2021Zhuchongzhi2.1} (although for the complete-graph architecture, rather than 2D). We conclude the main text with an outlook in \autoref{sec:whitenoiseoutlook}. The rigorous proofs and details behind the map to stochastic processes then appear in the appendices.

\section{A model of noisy random quantum circuits}\label{sec:noisemodel}

Here we describe our model of noisy random quantum circuits. Let the circuit consist of $s$ two-qudit gates acting on $n$ qudits, each with local Hilbert space dimension $q$. We follow Ref.~\cite{dalzell2020anticoncentration} in defining a \textit{random quantum circuit architecture} as an efficient algorithm that takes the circuit specifications $(n,s)$ as input and outputs a quantum circuit diagram with $s$ two-qudit gates, that is, a length-$s$ sequence of qudit pairs (without specifying the actual gates that populate the diagram).  Our results fully apply for two specific architectures: the 1D architecture with periodic boundary conditions, and the complete-graph architecture, which were previously shown in Ref.~\cite{dalzell2020anticoncentration} to have the anti-concentration property as long as $s \geq \Omega(n\log(n))$, with a particular constant prefactor. Our results would also fully apply for standard architectures in $D$ spatial dimensions (with periodic boundary conditions) if it could be proved that they also achieve anti-concentration whenever $s \geq \Omega(n \log(n))$, as was conjectured in Ref.~\cite{dalzell2020anticoncentration}. 

Given an architecture and parameters $(n, s)$, we can generate a circuit instance by choosing the circuit diagram according to the architecture and then choosing each of the unitary gates in the diagram at random according to the Haar measure. Each instance is associated with an output probability distribution $\pideal$ over $q^n$ possible computational basis measurement outcomes $x \in [q]^n$ (where $[q] = \{0,1,\ldots,q-1\}$) that would be sampled if the circuit were implemented noiselessly.  Note that in the formal analysis we include a layer of $n$ (also Haar-random) single-qudit gates at the beginning and end of the circuit without counting these $2n$ gates toward the circuit size; these might be regarded as fixing the local basis for the input product state and the measurement of the output. 

\subsection{Local noise model}

We augment this setup by inserting single-qudit noise channels into the circuit diagram, which act on qudits involved in a multi-qudit gate immediately following the gate, as shown in the example in \autoref{fig:noisycircuitdiagram}. In our model, the single-qudit gates remain noiseless and measurements are assumed to be perfect.\footnote{In the experiments of Refs.~\cite{Arute2019GoogleQuantumSupremacy,USTC2021StrongQCompAdv,USTC2021Zhuchongzhi2.1}, single-qubit gates had significantly smaller (but still non-zero) error rates compared to two-qubit gates. However, readout error rates were significantly larger than gate error rates, something that is not incorporated into our model. Our simplified noise model aims to capture the spirit of a noisy random quantum circuit experiment and show that the white-noise phenomenon can be proved in an idealized setting. We do not aim to specifically model all of the details of the experimental setups in Refs.~\cite{Arute2019GoogleQuantumSupremacy,USTC2021StrongQCompAdv,USTC2021Zhuchongzhi2.1}.}
\begin{figure}[h]
\centering
\scalebox{0.9}{
\begin{tikzpicture}[scale=1.0,thick]
\foreach[count=\i] \y in {1,2,3,4}
{\tikzmath{\j=int(\i-4);\jx=int(10-\i);\xoff = 0.2;}
\draw (0.5,\y) -- (12.2,\y);
\draw[fill=white,rounded corners] (0.8,\y-0.35) rectangle (1.6,\y+0.35);
\node at (1.2,\y) {\footnotesize $U^{\scalebox{0.6}{(\j)}}$};
\draw[fill=white,rounded corners] (11,\y-0.35) rectangle (11.8,\y+0.35);
\node at (11.4,\y) {\footnotesize $U^{\scalebox{0.6}{(\jx)}}$};
\node at (0.18,\y) {$\ket{0}$};
\draw[fill=white,rounded corners] (12+\xoff,\y-0.3) rectangle (12.6+\xoff,\y+0.3);
\draw[-latex,semithick] (12.3+\xoff,\y-0.22) -- (12.46+\xoff,\y+0.22);
\draw[semithick] (12.1+\xoff,\y-0.16) to[out=60,in=120] (12.5+\xoff,\y-0.16);
}
\foreach \x/\y in {
3.3/4,3.3/3,
5.1/2,5.1/3,
6.9/4,6.9/3,
8.7/1,8.7/2,
10.5/2,10.5/3}
{\draw[fill=white] (\x-0.05, \y) circle (0.25);
\node at (\x-0.05,\y) {\scalebox{0.65}{$\mathcal{N}$}};}
\foreach[count=\i] \x/\y in {2/3, 3.8/2, 5.6/3, 7.4/1, 9.2/2}
{\draw[fill=white,rounded corners] (\x,\y-0.3) rectangle (\x+0.8,\y+1.3);
\node at (\x+0.42,\y+0.5) {\scalebox{0.86}{$U^{(\i)}$}};}
\end{tikzpicture}
}
\caption[Example of a noisy quantum circuit diagram on $n=4$ qudits and $s=5$ two-qudit gates]{Example of a noisy quantum circuit diagram on $n=4$ qudits with $s=5$ two-qudit gates. A pair of single-qudit noise channels $\mathcal{N}$ follow each two-qudit gate. The circuit begins and ends with a layer of noiseless single-qudit gates.}
\label{fig:noisycircuitdiagram}
\end{figure}
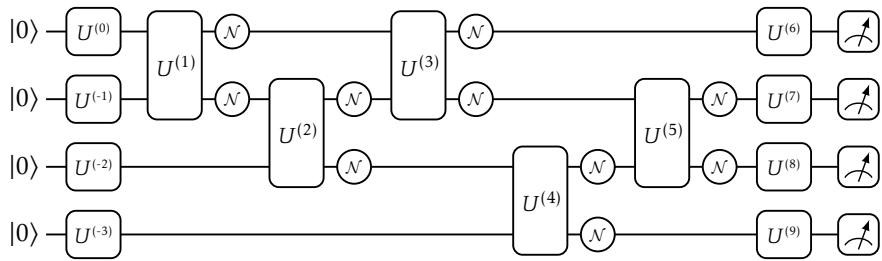
Thus, the core assumption is that the noise is local, i.e.~independent from qudit to qudit. We assume each noise channel $\mathcal{N}$ is a unital and completely positive trace-preserving map. 

For a given noise channel, there are only two parameters that matter for our analysis, the \textit{average infidelity} and the \textit{unitarity} of the channel. The average infidelity for a channel $\mathcal{N}$ is defined as
\begin{equation}\label{eq:averageinfidelity}
    r = 1-\int dV \tr\left[ V \ketbra{\psi} V^{\dagger}\mathcal{N}(V \ketbra{\psi} V^{\dagger})\right]\,,
\end{equation}
where the integral is over the Haar-measure on $q \times q$ unitary matrices $V$ and $\ketbra{\psi}$ is any pure state. The average infidelity is one measure of the overall noise strength of the channel $\mathcal{N}$. Following Refs.~\cite{Wallman2015EstimatingCoherence,Carignan-Dugas2019Unitarity}, the \textit{unitarity} is defined for unital channels as
\begin{equation}\label{eq:unitarity}
    u = \frac{q}{q-1}\left(\int dV \tr\left[ \mathcal{N}\left(V \ketbra{\psi} V^{\dagger}\right)^2\right]-\frac{1}{q}\right)\,.
\end{equation}
The unitarity is the expected purity of the output state under random choice of input state, scaled to have minimum value of 0 and maximum value of 1.

\paragraph{Examples: depolarizing, dephasing, and rotation channels}
It is helpful to consider explicitly the following three channels. First, the depolarizing channel 
\begin{align}\label{eq:depolarizing}
\mathcal{N}_{depo}(\rho) &= (1-\gamma)\rho + \gamma \frac{I}{q}  = (1-\epsilon)\rho +\frac{\epsilon}{q^2-1} \sum_{i=1}^{q^2-1}P_i \rho P_i^\dagger\,,
\end{align}
where $\gamma = \epsilon q^2/(q^2-1)$, $\{P_i\}_{i=1}^{q^2-1}$ is the set of single-qudit Pauli matrices (appropriately generalized to higher $q$), and $I$ is the $q \times q$ identity matrix. There are two ways to think of the channel: first, with probability $1-\gamma$ doing nothing and with probability $\gamma$ resetting the state to the maximally mixed state on that qudit; second, with probability $1-\epsilon$ doing nothing and with probability $\epsilon$ choosing a Pauli operator at random to apply to the qudit. 

We can also consider the dephasing channel
\begin{align}\label{eq:dephasing}
\mathcal{N}_{deph}(\rho) &= (1-\frac{q}{q-1}\epsilon)\rho +\frac{q}{q-1}\epsilon \sum_{i=0}^{q-1} \ketbra{i} \rho \ketbra{i}\,,
\end{align}
which represents doing nothing with probability $1-q\epsilon/(q-1)$ and performing a measurement in the computational basis with probability $q\epsilon/(q-1)$. 

Finally, we can consider a coherent noise channel, for example the rotation channel
\begin{align}\label{eq:rotation}
\mathcal{N}_{rot}(\rho) &= e^{-i\theta \ketbra{0}}\rho e^{i\theta \ketbra{0}}\,,
\end{align}
which applies a small unitary rotation by angle $\theta$ to the state.

The average infidelity and unitary of these channels are given in \autoref{tab:channelexamples}.
\begin{table*}[ht]
\centering
\normalsize
{\renewcommand{\arraystretch}{1.5}
 \begin{tabular}{|c|c|c|}
         \hline
         \textbf{channel}& \textbf{avg.~infidelity} $r$ & \textbf{unitarity} $u$ \\
         \hline
         depolarizing, Eq.~\eqref{eq:depolarizing} & $\frac{q}{q+1}\epsilon$  & $(1-\frac{q^2}{q^2-1}\epsilon)^2$ \\
         \hline
         dephasing, Eq.~\eqref{eq:dephasing} & $\frac{q}{q+1}\epsilon$ & $1-\frac{q^2}{q^2-1}(2\epsilon-\frac{q}{q-1}\epsilon^2)$ \\
         \hline
         rotation, Eq.~\eqref{eq:rotation} & $\frac{2(q-1)}{q(q+1)}(1-\cos(\theta))$ & $1$\\
        \hline
\end{tabular}
}
\caption{\label{tab:channelexamples} Average infidelity and unitarity for three different single-qudit noise channels, where $q$ denotes the local dimension of the qudits ($q=2$ for qubits).  }
\end{table*}

\subsection{Output distributions of the quantum circuit}

Suppose the locations of the $s$ two-qudit gates have been fixed, with gate $t$ acting on qudits $\{i_t,j_t\}$. Then a circuit instance is specified by a sequence $(U^{(-n+1)},\ldots, U^{(s+n)})$, where $U^{(t)}$ is a $q^2 \times q^2$ (two-qudit) unitary matrix if $1\leq t \leq s$ and a $q \times q$ (single-qudit) unitary matrix otherwise. Accordingly, for each $t$, let 
\begin{equation}\label{eq:U^t}
    \mathcal{U}^{(t)}(\sigma) = \left(\mathcal{I}_{[n]\setminus\{i_t,j_t\}} \otimes U_{\{i_t,j_t\}}^{(t)} \right) \sigma \left(\mathcal{I}_{[n]\setminus\{i_t,j_t\}} \otimes U_{\{i_t,j_t\}}^{
    \dagger(t)} \right)
\end{equation} 
denote the unitary channel that acts as $U^{(t)}$ on qudits $i_t$ and $j_t$ and as the identity channel (denoted by $\mathcal{I}$) on the other qudits.
To account for noise, let
\begin{equation}\label{eq:tildeU^t}
    \widetilde{\mathcal{U}}^{(t)} = 
    \begin{cases} 
    \left(\mathcal{I}_{[n]\setminus\{i_t,j_t\}}\otimes \mathcal{N}_{\{i_t\}} \otimes \mathcal{N}_{\{j_t\}}\right)\circ \mathcal{U}^{(t)} & \text{if } 1 \leq t \leq s \text{ (two-qudit)}\\
    \mathcal{U}^{(t)} & \text{otherwise (single-qudit)}
    \end{cases}
\end{equation} 
be the channel that applies noise channels after applying the unitary gate. 
Now we can define the ideal and noisy output distributions by
\begin{align}
    \pideal(x) &= \tr\left[\ketbra{x}\; \mathcal{U}^{(s+n)}\circ \cdots \circ \mathcal{U}^{(-n+1)}\left(\ketbra{I^n}\right)\right] \label{eq:pideal}\\
    \pnoisy(x) &= \tr\left[\ketbra{x}\; \widetilde{\mathcal{U}}^{(s+n)}\circ \cdots \circ \widetilde{\mathcal{U}}^{(-n+1)}\left(\ketbra{I^n}\right)\right]\,. \label{eq:pnoisy}
\end{align}

Our work compares the distribution $\pnoisy$ to the white-noise distribution $\pwhitenoise$ (defined in Eq.~\eqref{eq:whitenoisedefinition} and repeated here)
\begin{equation}
    \pwhitenoise(x) = F\pideal(x) + (1-F)q^{-n}
\end{equation}
for some choice of $F$. 
The white-noise distribution is a mixture of the ideal distribution and the uniform distribution. Note that $\pideal$, $\pnoisy$, and $\pwhitenoise$ all depend implicitly on the circuit instance $U$.  In the analysis we treat $F$ as a free parameter, and we choose it such that our bound on the distance between $\pnoisy$ and $\pwhitenoise$ is minimized. The total variation distance between two distributions $p_1$ and $p_2$ is defined as
\begin{equation}
    \text{TVD}(p_1,p_2) = \frac{1}{2}\lVert p_1 - p_2 \rVert_1 = \frac{1}{2}\sum_{x}|p_1(x)-p_2(x)| \,.
\end{equation}

\paragraph{Comment on randomness in our setup}

There are multiple types of randomness in our analysis, and in understanding our result it is important to keep track of how they interplay. First of all, the noiseless circuit instance $U$ is generated randomly by choosing each gate to be Haar random. The choice of $U$ determines an ideal \textit{pure} output state. Second of all, for each fixed choice of $U$, the noise channels may introduce randomness that makes the noisy output state \textit{mixed}. When the noise is depolarizing noise, this might be regarded as the insertion of a randomly chosen pattern of Pauli errors. Lastly, the measurement of the state in the computational basis gives rise to a random measurement outcome drawn from a certain classical probability distribution: $\pideal$ if we are considering the noiseless circuit, and $\pnoisy$ if we are considering the noisy circuit. The important thing to remember is that we are primarily concerned with thinking about \textit{fixed} instances $U$ and the interplay between the resulting probability distributions $\pideal$, $\pnoisy$ and $\pwhitenoise$ for that instance. Then, we make a statement about these distributions that holds in expectation over random choice of $U$. If desired, one could then use Markov's inequality to form bounds on the fraction of instances $U$ for which the white-noise approximation must be good. 

\paragraph{Comment on more general (universal) gate sets}
We consider random quantum circuits built from local 2-site unitary gates drawn randomly with respect to the Haar measure. As our analysis involves only second moment quantities, our results therefore directly apply to any gate set (or distribution on the 2-site unitary group) that forms an exact unitary 2-design, e.g.\ random Clifford circuits with gates drawn from the Clifford group. Furthermore, circuits constructed with gates drawn randomly from universal gate sets should give rise to similar scrambling phenomena and we expect that our results hold for such circuits, including the actual random circuit experiments performed in Refs.~\cite{Arute2019GoogleQuantumSupremacy,USTC2021StrongQCompAdv,USTC2021Zhuchongzhi2.1}. While our method is not directly generalizable to other gate sets, we anticipate that if our analysis were extendable to such gate sets, the results would only change by constant factors. 

Some evidence for this is provided by the independence of the spectral gap for universal gate sets \cite{BGgap}. This implies that the depth at which random quantum circuits scramble (and converge to approximate unitary designs) only changes by a constant factor when one considers circuits comprised of gates drawn randomly for any universal gate set \cite{BHH2016RQCtdesign}.

\section{Overview of contributions}\label{sec:whitenoiseresults}

The main result of this paper is a proof that, for typical random circuits, the output distribution $\pnoisy$ of the quantum circuit with local noise is very close to the white-noise distribution $\pwhitenoise$ if the noise is sufficiently weak. Specifically, we prove an upper bound on the expectation value of the total variation distance between the two distributions. In proving that result, we also prove a statement about the expected fidelity in noisy random quantum circuits, and another statement about the speed at which $\pnoisy$ approaches the uniform distribution. 
For all statements, the notation $\EV_U$ denotes expectation over choice of Haar-random single-qudit and two-qudit gates. 

In the rest of this section, we state our results for general noise channels, deferring the proofs to \autorefapp{app:rigorousproofs}, but first we summarize the contributions specifically applied to the depolarizing channel in \autoref{tab:resultsdepolarizing}.

\begin{table*}[ht]
\centering
\normalsize
{\renewcommand{\arraystretch}{1.75}
 \begin{tabular}{|c|c|}
         \hline
         \textbf{Fidelity decay} & $\bar{F} = e^{-2s\epsilon \; \pm \; O(s\epsilon^2)}$   \\
         \hline
         \textbf{Approach to uniform} & $\EV_U\left[\frac{1}{2}\lVert \pnoisy - \punif\rVert_1\right] \leq e^{-2s\epsilon \; + \; O(s\epsilon^2)}$ \\
         \hline
         \textbf{Distance from} $\pwhitenoise$ \textbf{for} $F = \bar{F}$ & $\EV_U\left[\frac{1}{2}\lVert \pnoisy - \pwhitenoise \rVert_1\right] \leq O(F \epsilon \sqrt{s})$ \\
        \hline
\end{tabular}
}
\caption[Summary of results in the case of depolarizing noise]{\label{tab:resultsdepolarizing} Summary of results when the noise is depolarizing (Eq.~\eqref{eq:depolarizing}) with error parameter $\epsilon$. The quantity $\bar{F}$, given in Eq.~\eqref{eq:defAveOutFidelity}, is the expectation of the linear cross entropy metric using noisy samples, normalized by its expectation using ideal samples.
These statements apply for the 1D and complete-graph architectures when the circuit size is larger than $\Omega(n\log(n))$ (corresponding to the regime where the anti-concentration property has been achieved), and assuming that the quantity $\epsilon n \log(n)$ is small enough to be neglected. We believe that this condition can be relaxed to $\epsilon < c/n$ for some constant $c$.}
\end{table*}

\paragraph{Comment on architectures} The theorem statements below are expressed only for the 1D and complete-graph architectures, which are known to anti-concentrate after circuit size $\Theta(n\log(n))$. In the appendix, we prove slightly more general statements that also hold for any architecture consisting of layers and satisfying a natural connectivity property (this includes standard architectures in $D$ spatial dimensions with periodic boundary conditions). These statements depend on the anti-concentration size $s_{AC}$ of these architectures, which is conjectured to be $\Theta(n\log(n))$ but for which the best known upper bound is $O(n^2)$ \cite{dalzell2020anticoncentration}. 

\subsection{Fidelity decay}

Define the quantity
\begin{equation}\label{eq:defAveOutFidelity}
    \bar{F} = \frac{\EV_U\Big[ \sum_x \pnoisy(x)(q^n\pideal(x)-1)\Big]}{\EV_U\Big[ \sum_x \pideal(x)(q^n\pideal(x)-1)\Big]}\,.
\end{equation} 
The quantity $\bar{F}$ may be regarded as an estimate of the fidelity of the noisy quantum device with respect to the ideal computation;
when $\pnoisy(x)$ and $\pideal(x)$ are viewed as random variables in the instance $U$, $\bar{F}$ is equal to their covariance, normalized by the variance of $\pideal$. Note also that the numerator of $\bar{F}$ is the expected score on the linear cross-entropy benchmark (as proposed in  Ref.~\cite{Arute2019GoogleQuantumSupremacy}) using samples from the noisy device, and the denominator is the expected score using samples from the ideal output distribution. Refs.~\cite{Liu2021BenchmarkingRCS, Rinott2021Statistical} studied a similar quantity, the difference being that the $\EV_U$ appears outside the fraction in their case. Additionally, note that the denominator is given by $q^n Z-1$, where $Z$ is the collision probability studied in Refs.~\cite{HarrowMehraban2018tdesign,dalzell2020anticoncentration}. The results of Ref.~\cite{dalzell2020anticoncentration} imply that the denominator becomes within a small constant factor of $(q^n-1)/(q^n+1) \approx 1$ (and can therefore be essentially ignored) after $\Theta(n\log(n))$ gates.

\begin{theorem}\label{thm:fidelitydecay}
    Consider either the complete-graph architecture or the 1D architecture with periodic boundary conditions on $n$ qudits of local Hilbert space dimension $q$ and comprised of $s$ gates. Let $r$ be the average infidelity of the local noise channels. Then there exists constants $c$ and $n_0$ such that whenever $r \leq c/n$ and $n \geq n_0$, the following holds:
\begin{align}
    \bar{F} &\geq \exp\left(-2sr(1+q^{-1})\right)e^{-O(sr^2)-O(sq^{-2n})} \label{eq:FXEBlowerbound}\\
    \bar{F} &\leq\exp\left(-2sr(1+q^{-1})\right) Q_1 \,, \label{eq:FXEBupperbound}
\end{align}
where
\begin{equation}\label{eq:Q1}
    Q_1 = \exp\left(O(sr^2) + O(rn\log(n)) + e^{O(\log(n))-\Omega(s/n)} + O(nr \log(1/(nr)))\right)\,.
\end{equation}
\end{theorem}
Note that the relationship $\epsilon= r(q+1)/q$ holds for the depolarizing channel as defined in Eq.~\eqref{eq:depolarizing}, so, ignoring the $O(q^{-2n})$ corrections,
\begin{equation}
    e^{-2s\epsilon - O(s\epsilon^2) } \leq \bar{F} \leq e^{-2s\epsilon + O(s\epsilon^2) + O(\epsilon n\log(n))+O(n\epsilon\log(1/(n\epsilon))}\,,
\end{equation}
indicating that the fidelity decreases exponentially with the expected number of Pauli errors $2s\epsilon$, as long as the noise is sufficiently weak that the other terms can be ignored. In particular, three conditions must be met to approximate $Q_1$ by 1 in Eq.~\eqref{eq:FXEBupperbound}: (1) $\epsilon^2 s \ll 1$, (2) anti-concentration has been reached, i.e.~$s \geq \Omega(n\log(n))$, and (3) $\epsilon \ll 1/(n\log(n))$. One implication of \autoref{thm:fidelitydecay} is that the same kind of decay extends to general noise channels and is observed even for coherent noise channels like the rotation channel.

\subsection{Convergence to uniform }

We show an upper bound on the expected total variation distance between the output of the noisy quantum device $\pnoisy$ and the uniform distribution. Our bound decays exponentially in the number of error locations, under certain circumstances. In particular, it decays exponentially in $(1-u)(1-q^{-2})s$ where $u$ is the unitarity of the local noise channels. 

\begin{theorem}\label{thm:convergenceToUniform}
   Consider either the complete-graph architecture or the 1D architecture with periodic boundary conditions on $n$ qudits of local Hilbert space dimension $q$ and $s$ gates. Let $u$ be the unitarity of the local noise channels (and define $v=1-u$). Then there exist constants $c$ and $n_0$ such that as long as $v \leq c/n$ and $n \geq n_0$
    \begin{equation}
        \EV_U \left[\frac{1}{2}\lVert \pnoisy - \punif \rVert_1 \right] \leq \exp(-sv(1-q^{-2})) Q_2 \,,
    \end{equation}
    where $\punif$ is the uniform distribution, and 
    \begin{equation}\label{eq:Q2}
        Q_2 = \exp\left(O(sv^2) + O(v n \log(n)) + e^{O(\log(n))-\Omega(s/n)}+O(nv \log(1/(nv))\right)\,.
    \end{equation}
\end{theorem}
Note that $Q_2$ is small under a similar three conditions as in the fidelity decay result: (1) $s(1-u)^2 \ll 1$, (2) anti-concentration has been reached, and (3) $n\log(n)(1-u) \ll 1$. 

For the depolarizing channel, $u = 1-2\epsilon(1-q^{-2})^{-1}$ up to first order in $\epsilon$, so the distance to uniform decays like $e^{-2s\epsilon}$, which is identical to the rate of fidelity decay.  On the other hand, the unitarity of the rotation channel is $u=1$, so our upper bound does not decay with $s$, even though $\bar{F}$ does decay for the rotation channel. This is expected because the rotation channel is coherent; indeed, unlike the other two examples, it sends pure states to pure states. The ideal pure state and the noisy pure state will become less and less correlated as more noise channels act, which explains why $\bar{F}$ decays, but the output distribution for the noisy pure state will not converge to uniform. 

\subsection{Distance to white-noise distribution}

We show a stronger statement that is meaningful when the noise is incoherent. Not only does the output distribution decay to uniform, it does so in a very particular way, preserving an uncorrupted signal from the ideal distribution. We show that $\pnoisy$ is close to $\pwhitenoise$ by upper bounding the expected total variation distance between the two distributions. 

\begin{theorem}\label{thm:whitenoisebound}
    Consider either the complete-graph architecture or the 1D architecture with periodic boundary conditions on $n$ qudits of local Hilbert space dimension $q$ and $s$ gates.  Let $r$ be the average infidelity and $u$ the unitarity of the local noise channels (and define $v=1-u$). Let
    \begin{equation}\label{eq:delta}
        \delta=2r(1+q^{-1})-(1-u)(1-q^{-2}) \,.
    \end{equation}
    Then, when we choose $F = \bar{F}$ as in Eq.~\eqref{eq:defAveOutFidelity}, there exist constants $c_1$, $c_2$, and $n_0$ such that as long as $v \leq c_1/n$, $r \leq c_2/n$, and $n \geq n_0$,
    \begin{align}
        \EV_U \left[\frac{1}{2}\lVert \pnoisy - \pwhitenoise \rVert_1 \right]
        \leq{}& \bar{F}\sqrt{s}\left(\sqrt{\delta}+O(v)+O(r)\right) + O\big(\bar{F}\sqrt{vn\log(n)}\big) \nonumber\\
        &\qquad +O\big(\bar{F}\sqrt{nv\log(1/nv)}\big)+ \bar{F}e^{O(\log(n))-\Omega(s/n)}\,, \label{eq:WNdistance}
    \end{align}
    whenever the right-hand side of Eq.~\eqref{eq:WNdistance} is less than $\bar{F}$.  
\end{theorem}

We make a couple of comments. First, we emphasize how small the right-hand side of Eq.~\eqref{eq:WNdistance} is. The quantity $\bar{F}$ is decaying exponentially in the number of expected errors, as shown in \autoref{thm:fidelitydecay}. We showed in \autoref{thm:convergenceToUniform} that $\pnoisy$ converges to uniform at roughly the same rate. However, the distance between $\pnoisy$ and $\pwhitenoise$ is much smaller than $\bar{F}$ if the parameters are sufficiently weak, demonstrating that the noisy and white-noise distribution are much closer to each other than either are to uniform. 

Second, let us examine the quantity $\delta$. For the depolarizing channel and the dephasing channel, the leading term in $\delta$ cancels out leaving $\delta = O(\epsilon^2)$, so the $\sqrt{\delta}$ term in Eq.~\eqref{eq:WNdistance} is on the same order as the other terms. This is a signature of incoherent noise. The coherent rotation channel, which has $u=1$ and $r = O(\theta^2)$, has $\delta = O(\theta^2)$, so $\sqrt{\delta}$ is large compared to the other terms in the expression. In this case, we would need $sr \ll 1$ for the approximation to be good, but if this is true, then $\bar{F} \approx 1$ and the white-noise approximation is trivial.

Relatedly, the parameter $\delta$ can be connected to the diamond distance $D$ of the channel $\mathcal{N}$, which is the maximum amount action by $\mathcal{N}$ can change an input state (which might be entangled with an auxiliary system) as measured by the trace norm. If $\mathcal{N}$ is applied $2s$ times, the total deviation in trace norm from the ideal output can be as large as $2sD$ in the worst case. It was shown in Ref.~\cite{Kueng2016Comparing} that $D = O(\sqrt{\delta})$, specifically
\begin{equation}
    \frac{1}{2}\sqrt{\delta}\leq D \leq \frac{q^2}{2}\sqrt{\delta}\,.
\end{equation}
It is also known that $r \leq O(D)$ and $1-u \leq O(D)$. Thus, if we ignore the final three terms in Eq.~\eqref{eq:WNdistance}, we can write our result as
\begin{equation}
    \EV_U\left[\frac{1}{2}\lVert \pnoisy - \pwhitenoise \rVert_1\right] \leq O(F D\sqrt{s})\,.
\end{equation}
This emphasizes that the fundamental result is an improved trade-off between noise and circuit size; the strength of the signal decays exponentially, but the error on the signal (after renormalization) grows quadratically slower (as $O(D\sqrt{s})$) in the case of random quantum circuits with incoherent noise than it does in the worst case (as $O(Ds)$, for arbitrary circuits and arbitrary noise channels with diamond distance $D$). 

\section{Related work and implications}\label{sec:whitenoiseimplications}

\subsection{Quantum computational supremacy}

A central motivation for our work has been recent quantum computational supremacy experiments \cite{Arute2019GoogleQuantumSupremacy, USTC2021StrongQCompAdv} that sampled from the output of noisy random quantum circuits on superconducting devices. In this context, the main claim is that no classical computer could have performed the same feat in any reasonable amount of time. While no efficient classical algorithms to simulate the quantum device performing this task are known, there is a lack of concrete theoretical evidence that no such algorithm exists.

Our work bolsters the theory behind these experiments in two ways, assuming noise in the device is sufficiently well described by our local noise model. First, our fidelity decay result validates using the linear cross-entropy metric to benchmark the overall noise rate in the device, and quantify the amount of signal from the ideal computation that survives the noise. Second, convergence to the white-noise distribution has theoretical benefits with respect to a potential proof that the random circuit sampling task accomplished by the device is actually hard for classical computers. 

\subsubsection{Linear cross-entropy benchmarking}

Quantum computational supremacy experiments are complicated by the fact that since (by definition) they cannot be replicated on a classical computer, it is non-trivial to classically verify that they actually performed the correct computational task. A partial solution to this issue has been the proposal of linear cross-entropy benchmarking, whereby a sample $x$ is generated by the device according to the noisy output distribution $\pnoisy$, and a classical supercomputer is used to compute $\pideal(x)$.\footnote{This requires exponential time but can be tractable for circuit sizes up to $n=50$ or so (in the case of a 2D architecture, the computational cost also depends on the depth of the circuit).} When $T$ samples $\{x_1,\ldots,x_T\}$ are chosen, the average
\begin{equation}
    \mathcal{F} = \frac{1}{T}\sum_{i=1}^T (q^n\pideal(x) -1)
\end{equation}
is calculated, which is an empirical measure of the circuit fidelity. We can see that the expected value of $\mathcal{F}$ is precisely $\sum_x \pnoisy(x)(q^n \pideal(x) - 1)$, which is the numerator of the quantity $\bar{F}$ defined in Eq.~\eqref{eq:defAveOutFidelity}. Meanwhile, the denominator of $\bar{F}$ becomes close to 1, so long as the output is anti-concentrated. In \autoref{thm:fidelitydecay}, we show that if the depolarizing error rate $\epsilon$ satisfies $\epsilon \ll 1/(n\log(n))$ and as long as $\epsilon^2 s \ll 1$, then there are matching upper and lower bounds on the expected value of $\mathcal{F}$, which decays with the circuit size like $e^{-2\epsilon s}$. Thus, assuming our local noise model, we prove that one can infer $\epsilon$ given $\mathcal{F}$ and $s$. The inferred value of $\epsilon$ can then be compared to the noise strength estimated when testing each circuit component individually, thus providing one method of verification that the components are behaving as expected during the experiment.

Indeed, the idea of using random circuit sampling as an alternative to randomized benchmarking was formally proposed in Ref.~\cite{Liu2021BenchmarkingRCS}, a work that has certain similarities to ours. In particular, like us, they find that the condition $1/\epsilon \geq \Omega(n)$ appears necessary for controlled decay of the fidelity. (Our result can be expressed as requiring $1/\epsilon \geq \tilde{\Omega}(n)$, where the tilde hides log factors, and we believe those log factors are not necessary for our result.) They give analytical and numerical evidence that the fidelity decays as $e^{-2 \epsilon s}$.  Additionally, like us, they use a map from random quantum circuits to identity-swap configurations to motivate their results. However, they only analytically study the fidelity decay up to first order in the error rate for a 1D architecture; that is, they compute the expected fidelity due to contributions with an error at only one location (or a correlated set of locations at the same depth).  
On the other hand, their error model is more general than ours as we do not consider correlated errors (their theoretical analysis handles Pauli errors of up to weight three); in the context of noise characterization, this is important as correlated errors are often the most difficult to diagnose. On this point, we believe correlated errors could be handled by our method with a more intricate analysis, but we leave that for future work. 
Relatedly, exponential decay of fidelity in noisy systems has been proposed \cite{slagle2021testing} as an experimentally detectable signature of quantum mechanics that distinguishes it from theories where quantum mechanics emerges from an underlying classical theory. Our work may help justify these proposals. 

Note that as the fidelity decays, more samples must be generated to form a good estimate of the mean of $\mathcal{F}$. Since $\pideal(x)$ for uniformly random $x$ has standard deviation on the order of $q^{-n}$ (assuming anti-concentration), the standard deviation of $\mathcal{F}$ is expected to decay with the number of samples like $1/\sqrt{T}$. Thus, resolving the mean of $\mathcal{F}$ with enough precision to differentiate it from 0 requires $T = \Omega(1/\mathcal{F}^2)$ samples. 

We comment that while our analysis assumes that each noise location has the same value of $\epsilon$, this is not essential to our method. We expect it could be shown that the expected value of $\mathcal{F}$ decays like $\exp(-\sum_i \epsilon_i)$ where $i$ runs over all possible noise locations. Moreover, our analysis works for any kind of local noise, not just depolarizing noise; the only relevant parameter is the average infidelity of the noise channels. This includes coherent noise; for example, the average infidelity of the coherent rotation channel given in Eq.~\eqref{eq:rotation} is less than 1 and thus leads  to exponential decay of $\mathcal{F}$. This is consistent with Ref.~\cite{Liu2021BenchmarkingRCS}, which previously showed that from the perspective of fidelity decay, every channel is equivalent to an (incoherent) Pauli noise channel.

\subsubsection{Classical hardness of sampling from the noisy output distribution}

To claim to have achieved quantum computational supremacy, the low-fidelity random circuit sampling experiments in Refs.~\cite{Arute2019GoogleQuantumSupremacy,USTC2021StrongQCompAdv} must define a concrete computational problem that their device solved, but a classical device could not also solve. Here there are a couple of options. One option is to simply rely directly on the linear cross-entropy benchmark and define the task to be generating a set of samples that scores at least
$\mathcal{F} \geq 1/\poly(n)$. A related idea is the task of Heavy Output Generation (HOG) \cite{AaronsonChen2017ComplexityTheoreticFoundations}, which is to generate outputs $x$ for which $\pideal(x)$ is large (i.e.~``heavy outputs'') significantly more often than a uniform generator. The upshot of these definitions is that in the regime where $\pideal(x)$ can be calculated classically with an exponential-time algorithm, it can be verified that the quantum device successfully performed the task. Their main drawback is that it is not clear whether running a (noisy) quantum computation is the only way to perform these tasks. Perhaps a (yet-to-be-discovered) classical algorithm can score well on the linear cross-entropy benchmark without performing an actual random circuit simulation; for example, this was the goal in Ref.~\cite{Barak2020Spoofing}.

Another option is to define the task specifically in terms of the white-noise distribution. Namely, one must produce samples from a distribution $\pnoisy$ for which $\frac{1}{2}\lVert \pnoisy - \pwhitenoise \rVert_1 \leq \eta F$ for some choice of $F$ not too small (ideally at least inverse polynomial\footnote{Inverse polynomial fidelity could be achieved while the white-noise assumption holds if, for example, the physical error rate decreases as $\Theta(1/n)$ and the circuit size grows as $\Theta(n\log(n))$ (corresponding to logarithmic depth). Deeper circuits would lead to exponentially small fidelity, although note that $\Theta(n\log(n))$ gates are sufficient for white-noise in most architectures (including 2D) assuming an anti-concentration conjecture from Ref.~\cite{dalzell2020anticoncentration}. Even if the fidelity is exponentially small, it could be argued that a (diminished) quantum speedup can survive asymptotically, but formally connecting such tasks to standard statements in complexity theory (such as the collapse of the \textit{polynomial} hierarchy) becomes more difficult.}
in $n$) and some small constant $\eta$. We refer to this task as ``white-noise random circuit sampling (RCS).'' A downside of this option is that even with unlimited computational power, an exponential number of samples from the device would be needed to definitively verify that the distribution is close to $\pwhitenoise$ in total variation distance. Our work provides a partial solution here, as we show that a local error model allows a device to accomplish the white-noise RCS task, as long as the error rate is sufficiently weak compared to the number of qubits. Thus, if the experimenters are sufficiently confident in the error model that describes their device, they can rely on our work to be confident they are performing the white-noise RCS task. 

The major upside of the white-noise RCS task is that one can give stronger evidence that it is classically hard to perform. For example, in the Supplementary Material of Ref.~\cite{Arute2019GoogleQuantumSupremacy}, it was shown that \textit{exactly} (i.e.~$\eta = 0$) sampling from $\pwhitenoise$ (a task they called ``unbiased noise $F$-approximate random circuit sampling'') in the \textit{worst case} is a hard computational task in the sense that an efficient classical algorithm for it would cause the collapse of the polynomial hierarchy ($\PH$), and further that its computational cost should be at most a factor of $F$ smaller than sampling exactly from $\pideal$.  In that spirit, we show in \autoref{thm:hardnessreduction}, in the appendix, that the more realistic task of sampling \textit{approximately} from $\pwhitenoise$ is essentially just as hard as sampling \textit{approximately} from $\pideal$, up to a linear factor of $F$ in the classical computational cost. This is important because some mild progress has been made toward establishing that approximately sampling from $\pideal$ is hard for the polynomial hierarchy, through a series of work that reduce the task of computing $\pideal(x)$ in the worst case to the task of computing $\pideal(x)$ in the average case up to some small error \cite{Bouland2019RCSComplexity,Movassagh2019QSandRQC,Bouland2021NoiseQuantumSupremacy,Kondo2021ImprovedRobustness}. Weaknesses in this result as evidence for hardness of approximate sampling were discussed in more detail in Refs.~\cite{Napp2019SEBD,Bouland2021NoiseQuantumSupremacy}, but it remains true that the white-noise-centered definition of the computational task is the likeliest route to a more robust version of quantum computational supremacy that can be grounded in well-studied complexity theoretic principles. 

\subsection{Convergence to uniform with circuit size}

It is widely understood that incoherent and uncorrected unital noise in quantum circuits should typically lead the output of a quantum circuit to lose all correlation with the ideal circuit and become nearly uniform. It is further asserted that the decay to uniform should scale with the circuit size; however, rigorous results have only shown a decay in total variation distance to uniform with the circuit depth $d$, following the form $e^{-\Omega(\epsilon d)}$. In particular, Ref.~\cite{Aharonov1996LimitationsNoisy} showed that any (even non-random) circuit with interspersed local depolarizing noise approaches uniform at least this quickly. Later, Ref.~\cite{GaoDuan2018NoisySimulation} showed the same is true for any Pauli noise model, at least for most circuits chosen from a particular random ensemble. However, in Ref.~\cite{Bouland2021NoiseQuantumSupremacy}, a stronger convergence at the rate of $e^{-\Omega(\epsilon s)}$ in random quantum circuits like ours was desired in order to show a barrier on further improvements of their worst-to-average-case reduction for computing entries of $\pideal$. To that end, they showed that exponential convergence in circuit size occurs in a toy model where each layer of unitary evolution enacts an exact global unitary $2$-design, and they conjectured the same is true in the local noise model we consider in this paper. Thus, our result in \autoref{thm:convergenceToUniform} gets close to providing the missing ingredient for their claim; for their application, we would need to extend our result to show $e^{-\Omega(\epsilon s)}$ even in the regime where $\epsilon = O(1)$, independent of $n$. Our result applies only for $\epsilon = O(1/n)$, but we believe the extension to $\epsilon = O(1)$ might also be provable with our method.

\subsection{Signal extraction in noisy experiments}

One implication of our work is that, in the parameter regime where our results apply, the signal from the noiseless random circuit experiment can be extracted by taking many samples. To illustrate this, suppose we are interested in some classical function $f(x)$ for $x\in [q]^n$ that takes values between $-1$ and $+1$. Choosing $x$ randomly from $\pideal$ induces a probability distribution over the resulting values of $f(x)$. To understand this distribution empirically (e.g.,~estimate its mean or variance), samples $x_i$ might be generated on a quantum device, but if the device is noisy, these samples will be drawn from $\pnoisy$ instead of $\pideal$. However, if $\pnoisy \approx \pwhitenoise$, then the sampled distribution over $f(x)$ will be a mixture of the ideal with weight $F$, and the distribution that arises from uniform choice of $x$ with weight $1-F$. Supposing the latter is well understood, inferences can be made about the former by repetition. For example, if $\sum_x \pideal(x) f(x) = \mu = O(1)$ and $\sum_x f(x)/q^n = 0$,\footnote{In a sense, the white-noise assumption is overkill for this application; a similar signal extraction could be performed even if $\pnoisy = F\pideal + (1-F)p_{\text{err}}$ for some non-uniform $p_{\text{err}}$ as long as drawing samples $x$ from $p_{\text{err}}$ lead to a mean for $f(x)$ that can be easily calculated in advance (when this is possible one can subtract a constant from $f$ and assume the mean is zero). However, the white-noise assumption certainly makes this process easier as it will typically be easy to calculate the mean of $f(x)$ under uniform choice of $x$. }
then the mean of $f$ under samples from $\pwhitenoise$ is $F\mu$. Meanwhile, the standard deviation of $f$ can be as large as $O(1)$, indicating that $O(1/F^2)$ samples from $\pwhitenoise$ are required to compute the mean $F\mu$ up to $O(F)$ precision. Generally, this procedure requires knowing the value of $F$. 

A concrete example of such a situation is the Quantum Approximate Optimization Algorithm (QAOA) \cite{Farhi2014QAOA}, where samples $x$ from the output of a parameterized quantum circuit are used to estimate the expectation of a classical cost function $C(x)$. The parameters can then be varied to optimize the expected value of the cost function. Our work is for Haar-random local quantum circuits, which are, in a sense, very different from QAOA circuits. For example, the marginal of typical random circuits on any constant number of qubits is very closed to maximally mixed, whereas QAOA circuits optimized for local cost functions will, by design, not have this property. Nevertheless, it is plausible that generic QAOA circuits might respond to local noise in a similar way as random quantum circuits. Indeed, in Refs.~\cite{Xue2021QAOANoise,Marshall2020LocalNoiseQAOA,Wang2021NoiseBarrenPlateau}, numerical and analytic evidence was given for the conclusion that the expectation value of the cost function and its gradient with respect to the circuit parameters decay toward zero when local noise is inserted into a QAOA circuit. This behavior would be consistent with a stronger conclusion that the output is well described by $\pwhitenoise$.

\section{Summary of method and intuition}\label{sec:whitenoisemethodandintuition}

In this section, we present a heuristic argument about why the technical statements above should hold. Then we give an overview of how we actually show it using our method, which analyzes certain Markov processes derived from the quantum circuits, extending our previous work in Ref.~\cite{dalzell2020anticoncentration}.

\subsection{Intuition behind error scrambling and error in white-noise approximation}

Our result that $\pnoisy$ is very close to $\pwhitenoise$ requires three conditions to be satisfied: (1) $\epsilon^2 s \ll 1$; (2) anti-concentration has been achieved, i.e.~$s \geq \Omega(n\log(n))$; and (3) $\epsilon n\log(n) \ll 1$.
Here, we try to motivate why these conditions should be sufficient and speculate about whether they are also necessary. In particular, we believe condition (3) can be significantly relaxed. 

For simplicity, lets restrict to qubits ($q=2$). Let $U$ denote the unitary enacted by the noiseless quantum circuit instance, so the ideal output state is the pure state $\rho_{\text{ideal}} = U\ketbra{0^n}U^{\dagger}$. If a location somewhere in the middle of the circuit experiences a Pauli error, then we could write the output state as $U_2 P U_1 \ketbra{0^n}U_1^{\dagger} P^\dagger U_2^{\dagger}$, where $P$ is a Pauli operator with support on only one qubit, and $U=U_2U_1$ is a decomposition of the unitary into gates that act before and after the error location. If we like, we can commute $P$ to act at the end of the circuit, giving $O_PU\ketbra{0^n}U^{\dagger}O_P^{\dagger}$ where $O_P = U_2 P U_2^\dagger$. Unlike $P$, the operator $O_P$ will likely have support over many qubits. Indeed, this is what we mean by scrambling; the portion of the circuit acting after the error location scrambles the local noise $P$ into more global noise $O_P$. We can handle error patterns $E$ with multiple Pauli errors similarly, by commuting each to the end one at a time and forming an associated global noise operator $O_E$.  

Next, we expand the output quantum state $\rho_{\text{noisy}}$ of the noisy circuit as a sum over all possible Pauli error patterns, weighted by the probability that each pattern occurs. Assuming the local noise is depolarizing, the probability of a pattern $E$ depends only on the number of non-identity Pauli operators in the error pattern, denoted by $|E|$.
\begin{equation}
    \rho_{\text{noisy}} = \sum_E \left(\frac{\epsilon}{3}\right)^{|E|}(1-\epsilon)^{2s-|E|} O_E \rho_{\text{ideal}} O_E^{\dagger} \,.
\end{equation}
The classical probability distribution $\pnoisy$ is then given by $\pnoisy(x) = \bra x \rho_{\text{noisy}} \ket x$ for each measurement outcome $x$. Observe that for the error pattern with $|E| = 0$ (no errors), we have $\rho_E = \rho_{\text{ideal}}$. There can be other error patterns for which $O_E\rho_{\text{ideal}}O_E^{\dagger} = \rho_{\text{ideal}}$; for example, when a lone Pauli-$Z$ error acts prior to any non-trivial gates, the state is unchanged since the initial state $\ket{0^n}$ is an eigenstate of all the Pauli-$Z$ operators. However, these error patterns are rare and for the sake of intuition we ignore this possibility. In essence, the white-noise assumption is the claim that when we take the mixture over output states for all of the error patterns, we arrive at a state $\rho_{\text{err}}$ that produces measurement outcomes that are very close to uniform. (Note that in general $\rho_{\text{err}}$ need not be close to maximally mixed to yield uniformly random measurement outcomes.) Letting $F = (1-\epsilon)^{2s}$, we may write
\begin{align}
    \rho_{\text{noisy}} &= F \rho_{\text{ideal}} + F \sum_{E: |E| > 0} \left(\frac{\epsilon/3}{1-\epsilon}\right)^{|E|} O_E \rho_{\text{ideal}} O_E^{\dagger} \\
    &= F \rho_{\text{ideal}} + (1-F) \frac{I}{2^n} + F \sum_{E: |E| > 0} \left(\frac{\epsilon/3}{1-\epsilon}\right)^{|E|} \left(O_E \rho_{\text{ideal}} O_E^{\dagger}- \frac{I}{2^n}\right)\,,
\end{align}
where $I/2^n$ denotes the maximally mixed state. This final term gives the deviations of the noisy output state $\rho_{\text{noisy}}$ from a linear combination of the ideal state and $I/2^n$. 

This allows us to state more clearly the intuition for our result. Since the circuit is randomly chosen and scrambles the local error patterns, the operators $O_E$ generally have large support and are essentially uncorrelated for different choices of error pattern $E$. Suppose we measure in the computational basis, and examine the probability of obtaining the outcome $x$. 
We can calculate the squared deviation between this value and the white-noise value under expectation over instance $U$. 
\begin{align}
    \EV_U[(\pnoisy(x) - \pwhitenoise(x))^2]
    ={}&\EV_U\left[\left(\bra x \rho_{\text{noisy}} \ket x - (F \bra x \rho_{\text{ideal}} \ket x + (1-F) 2^{-n} )\right)^2\right] \\
    ={}& F^2 \sum_{\substack{E,E'\\|E|,|E'| > 0}} \left(\frac{\epsilon/3}{1-\epsilon}\right)^{|E|+|E'|}\EV_U\left[\left(p_E(x)- 2^{-n}\right)\left(p_{E'}(x)- 2^{-n}\right) \right] \,,
\end{align}
where $p_{E}(x) = \bra{x}O_E \rho_{\text{ideal}} O_E^{\dagger} \ket{x}$.
Suppose we now make the approximation that the quantities $p_E(x)$ and $p_{E'}(x)$, when considered as functions of the random instance $U$, are independently distributed unless $E = E'$. Their mean is $2^{-n}$ and, assuming anti-concentration (condition (2)), their standard deviation is $O(2^{-n})$. Then we have
\begin{align}
      \EV_U[(\pnoisy(x) - \pwhitenoise(x))^2]
    \approx{}& F^2\sum_{E: |E| > 0} \left(\frac{\epsilon/3}{1-\epsilon}\right)^{2|E|}\EV_U \left[\left(p_E(x)- 2^{-n}\right)^2 \right]
    ={} F^2\sum_{E: |E| > 0} \left(\frac{\epsilon/3}{1-\epsilon}\right)^{2|E|}O(2^{-2n}) \\
    ={}& F^2\cdot O(2^{-2n}) \cdot\left((1+O(\epsilon^2))^{2s}-1\right) \\
    \approx{}& O(F^2 2^{-2n}\epsilon^2 s)\,
\end{align}
where the last line is true when $\epsilon^2 s \ll 1$. This implies that the deviation of each entry in the probability distribution $\pnoisy$ from the white-noise distribution is on the order of $F2^{-n}\epsilon \sqrt{s}$, and since there are $2^n$ entries, we have
\begin{equation}
    \EV_U\left[\frac{1}{2}\lVert \pwhitenoise-\pnoisy \rVert_1 \right] \approx O(F \epsilon \sqrt{s})\,.
\end{equation}
In other words, the total variation distance is much smaller than $F$ when $\epsilon^2 s \ll 1$, giving an intuitive reason for condition (1).  Moreover, without condition (2), the contribution of each term would be much larger than $O(2^{-2n})$, which illustrates why condition (2) is necessary. 

The key step in this analysis was the assumption of independence between $p_E$ and $p_{E'}$ when $E \neq E'$. This is only approximately true; indeed for a circuit that does not scramble errors, this will be a bad approximation because it might be common to have different error patterns $E$, $E'$ that produce the same (or approximately the same) effective error $O_E = O_{E'}$. However, for random quantum circuits, this outcome is unlikely for the vast majority of error pairs. Our rigorous proof, later, might be regarded as a justification of this intuition above.

Condition (3) is more subtle to motivate. In our analysis we require $\epsilon \ll 1/(n\log(n))$ so that the chance an error occurs while the circuit is still anti-concentrating (which takes $\Omega(n\log(n))$ gates) is small. This is helpful in the analysis because it allows us to essentially ignore the possibility that an error $P$ occurs near the beginning or end of the circuit, where there is insufficient time to scramble the error (either forward or backward in time). However, a finer-grained analysis might be able to handle these kinds of errors:  we believe condition (3) can be improved from $\epsilon^{-1} \gg \Omega(n\log(n)) = \tilde{\Omega}(n)$ to simply $\epsilon^{-1} \geq n/c$ for some constant $c$ that depends only on the architecture (1D vs.~complete-graph etc.). However, we do not believe that improvement beyond this point would be possible; there is a fundamental barrier that requires $\epsilon$ to scale as $O(1/n)$. 

The reason for this is essentially that if the white-noise approximation is to hold, the errors need to be scrambled at least as fast as they appear. 
The fidelity $F$ decreases like $(1-\epsilon)^{2s} = \exp(-2s\epsilon - O(s\epsilon^2))$, so each layer of $O(n)$ gates causes a decrease by a factor $\exp(-O(n\epsilon))$. Recall that we demand that the total variation distance between $\pnoisy$ and $\pwhitenoise$ be much smaller than $F$, so as $F$ decreases, this condition becomes increasingly stringent. Meanwhile, scrambling is fundamentally happening at the rate of increasing circuit depth, not size. One way to see this is simply that local Pauli errors $P$ that appear at a certain circuit location are expected to be scrambled into larger operators that grow ballistically with the depth \cite{NahumVijayHaah2018OperatorSpreading,VonKeyserlingk2018OperatorHydrodynamics}; each layer of $O(n)$ gates yields a constant amount of operator growth. Another way to see this is to consider a pair of error patterns $E$ and $E'$, where $E$ consists of a single Pauli error on qudit $j$ at layer $d$ and $E'$ consists of a single Pauli error on qudit $j$ at layer $d + \Delta$. The correlation between $p_E(x)$ and $p_{E'}(x)$, as a function of the random instance $U$, which is roughly speaking the chance that the random circuit transforms the first error into something resembling the second error, will decay exponentially with $\Delta$, the separation in depth between the two errors.\footnote{This is particularly clear if the random circuits are Clifford circuits (for which our results also apply since random Clifford gates form an exact 2-design). Clifford circuits transform the error $E$ at layer $d$ more or less uniformly at random into one of the roughly $4^\Delta$ possible Pauli operators at layer $d+\Delta$. The probability that this operator is $E'$ is exponentially small in $\Delta$.} Yet a third way to see this fact is to notice that, after a circuit has initially reached anti-concentration, convergence of the collision probability $Z=\EV_U[\sum_x \pideal(x)^2]$ to its limiting value $Z_H$ occurs like $Z = Z_H + O(Z_H)\exp(-O(s/n))$ \cite{dalzell2020anticoncentration}. Each additional layer of $O(n)$ gates only decreases the deviation of $Z$ from $Z_H$ by a constant factor. The terms $\EV_U[(p_E - 2^{-n})(p_{E'}-2^{-n})]$ for $E \neq E'$ that were ignored above are expected to obey a similar kind of decay to the value 0 for most choices of $(E,E')$, but if $F$ is decaying too fast, we are not able to neglect these terms. Each layer of $O(n)$ gates must incur at most a constant-factor decay in fidelity to not exceed the rate of scrambling; equivalently, $n\epsilon < c$ must hold for some constant $c$. 

\subsection{Noisy random quantum circuits as a stochastic process}\label{sec:overviewnoisywalk}

Our method is a manifestation of the ``stat mech method'' for random quantum circuits, developed in Refs.~\cite{Hayden2016HolographicDualityTN,NahumVijayHaah2018OperatorSpreading,VonKeyserlingk2018OperatorHydrodynamics,ZhouNahum2019EmergentStatMech} and further utilized in Refs.~\cite{Hunter-Jones2019StatMechDesign,BertiniPiroli2020ScramblingExactResults,JianYouVasseur2020MeasurementInduced,BaoChoiAltman2020TheoryPhaseTransition,Napp2019SEBD,dalzell2020anticoncentration,Liu2021BenchmarkingRCS,LiFisher2020StatMechQEC,Gullans2020QuantumCodingLowDepth,Barak2020Spoofing}, whereby averages over $k$ copies of random quantum circuits are mapped to partition functions of classical statistical mechanical systems. The mapping for $k=2$, corresponding to second-moment quantities, is particularly simple and amenable to analysis \cite{ZhouNahum2019EmergentStatMech,Hunter-Jones2019StatMechDesign,Napp2019SEBD,dalzell2020anticoncentration}. 

In Ref.~\cite{dalzell2020anticoncentration}, we analyzed the collision probability $Z = \EV_U[\sum_x \pideal(x)^2]$, a second-moment quantity, using the stat mech method, although we found it more useful to interpret the result as the expectation value of a certain stochastic process, rather than as a partition function. As we will see, this work is essentially an extension of the analysis in Ref.~\cite{dalzell2020anticoncentration} to account for the action of the single-qudit noise channels $\mathcal{N}$ that act after two-qudit gates. We explain the steps in this analysis below, and leave the formal proofs for the appendices. 

\paragraph{Expressing the total variation distance in terms of second-moment quantities} 

To apply this method, the first step is to express $\frac{1}{2}\lVert \pnoisy-\pwhitenoise \rVert_1$ in terms of
second-moment quantities. To do so, we use the general 1-norm to 2-norm bound: when $p_1$ and $p_2$ are vectors on a $q^n$-dimensional vector space, then
\begin{equation}\label{eq:1normto2norm}
    \lVert p_1-p_2 \rVert_1 \leq q^{n/2} \lVert p_1-p_2 \rVert_2 \,,
\end{equation}
where $\lVert p_1-p_2 \rVert_2 = \sqrt{\sum_x (p_1(x)-p_2(x))^2}$. Applying this identity with $p_1 = \pwhitenoise$ and $p_2 = \pnoisy$ and invoking Jensen's inequality for the concave function $\sqrt{\cdot}$, we find
\begin{equation}
    \EV_U\left[ \frac{1}{2}\lVert \pwhitenoise-\pnoisy \rVert_1\right] \leq q^{n/2} \EV_U\left[\frac{1}{2}\lVert \pwhitenoise-\pnoisy\rVert_2\right] \leq  \frac{1}{2}\sqrt{q^n\EV_U\left[\lVert \pwhitenoise-\pnoisy \rVert^2_2\right]}\,.
\end{equation}
Now we can expand
\begin{align}
   q^n\EV_U \left[\lVert \pwhitenoise - \pnoisy \rVert_2^2 \right] &=  q^n\EV_U\left[\sum_x\left( \left(F \pideal(x) + (1-F)q^{-n}\right)-\pnoisy(x)\right)^2 \right] \\
   &= (Z_2-1)-2F (Z_1-1) + F^2 (Z_0-1) \,, \label{eq:Z2Z1Z01-F^2}
\end{align}
where 
\begin{align}
    Z_0 &= q^{n}\EV_U\left[\sum_x\pideal(x)^2\right] = q^{2n}\EV_U\left[\pideal(0^n)^2\right]  \label{eq:Z0}\\
    Z_1 &= q^{n}\EV_U\left[\sum_x\pnoisy(x)\pideal(x)\right] = q^{2n}\EV_U\left[\pnoisy(0^n)\pideal(0^n)\right] \label{eq:Z1} \\
    Z_2 &= q^{n}\EV_U\left[\sum_x\pnoisy(x)^2\right] = q^{2n}\EV_U\left[\pnoisy(0^n)^2\right]  \label{eq:Z2}
\end{align}
are second-moment quantities (the second equality holds since by symmetry each term in the sum has the same value under expectation), with $Z_w$ containing $w$ copies of the noisy output and $2-w$ copies of the ideal output for each $w \in \{0,1,2\}$. Note that $Z_0 =q^n Z$ with $Z$ the collision probability studied in Refs.~\cite{dalzell2020anticoncentration,HarrowMehraban2018tdesign}. Furthermore, note that $F$ is a free parameter, and we may choose it so that it minimizes the right-hand side\footnote{Alternatively, one could choose $F$ to minimize the total variation distance bound \textit{relative} to the value of $F$, i.e.~the right-hand size of Eq.~\eqref{eq:Z2Z1Z01-F^2} divided by $F$. This minimization yields $F=(Z_2-1)/(Z_1-1)$, which is larger than $\bar{F}$. This might be the better option in some applications, but we do not choose it here because $F=(Z_2-1)/(Z_1-1)$ can be larger than 1 for some choices of noise channel $\mathcal{N}$ (in particular, coherent channels), which makes the definition of $\pwhitenoise$ meaningless.} of Eq.~\eqref{eq:Z2Z1Z01-F^2}, which occurs when
\begin{equation}
   F=\bar{F} = \frac{Z_1-1}{Z_0-1}\,,
\end{equation}
matching the definition for $\bar{F}$ in Eq.~\eqref{eq:defAveOutFidelity}.
Plugging in $F=\bar{F}$ yields
\begin{equation}\label{eq:boundUsingOptimalF}
    \EV_U\left[ \frac{1}{2}\lVert \pwhitenoise-\pnoisy \rVert_1\right] \leq \frac{1}{2} \bar{F} \sqrt{(Z_0-1)\left(\frac{(Z_0-1)(Z_2-1)}{(Z_1-1)^2}-1\right)}\,.
\end{equation}

\paragraph{Mapping second-moment quantities to stochastic processes} We bound the quantities $Z_0$, $Z_1$, and $Z_2$ by mapping them to stochastic processes.  These stochastic processes are the same as the stochastic process we studied in Ref.~\cite{dalzell2020anticoncentration}, except that the noise channels introduce slightly modified transition rules, as we now discuss. 

Second moment quantities include two copies of each random unitary gate in the circuit. The idea in Ref.~\cite{dalzell2020anticoncentration} was to perform the expectation over the two copies of each gate independently, using Haar-integration techniques. For a density matrix $\rho$ on two copies of a Hilbert space of dimension $q$, let
\begin{equation}\label{eq:Mrho}
    M[\rho] = \EV_{V}\left[V^{\otimes 2} \rho {V^{\dagger}}^{\otimes 2}\right] \,,\\
\end{equation}
where $\EV_V$ denotes expecation over choice of $V$ from the Haar measure over $q \times q$ matrices. Then, we have the following well-known formula (for which a derivation is provided in Ref.~\cite{dalzell2020anticoncentration})
\begin{equation}\label{eq:Midentityswap}
    M[\rho] =  \frac{\tr(\rho) - q^{-1}\tr(\rho S)}{q^2-1}I + \frac{\tr(\rho S) - q^{-1}\tr(\rho)}{q^2-1}S\,,
\end{equation}
where $I$ is the identity operation and $S$ is the swap operation on two copies of the single-qudit system. The equation above states that, after Haar averaging, the state of the system is simply a linear combination of identity and swap, with certain coefficients that can be readily calculated. For an $n$-qudit system acted upon by a sequence of single and two-qudit gates, this formula can be applied sequentially to each gate. After $t$ gates have been applied, the Haar-averaged state of the system can be expressed as a linear combination of $n$-fold tensor products of $I$ and $S$ (e.g.~for $n=3$, the state would be given by $c_1 I \otimes I \otimes I + c_2 I \otimes I \otimes S + c_3 I \otimes S \otimes I + \ldots +c_8 S \otimes S \otimes S$).

The important takeaway from Ref.~\cite{dalzell2020anticoncentration} was to interpret the coefficients of these $2^n$ terms as probabilities of a certain stochastic process over the set of length-$n$ bit strings $\{I,S\}^n$, which were called ``configurations.'' The stochastic process generates a sequence of $s+1$ configurations $\gamma =\smash{(\vec{\gamma}^{(0)},\ldots,\vec{\gamma}^{(s)})}$, which was called a ``trajectory,'' where the probabilistic transition from $\vec{\gamma}^{(t-1)}$ to $\vec{\gamma}^{(t)}$ depends only on the value of $\vec{\gamma}^{(t-1)}$ (Markov property).

The transition rules of the stochastic process are calculated by computing the coefficients in Eq.~\eqref{eq:Midentityswap}; here we state the result\footnote{In Ref.~\cite{dalzell2020anticoncentration}, two equivalent stochastic processes were formulated, an ``unbiased random walk'' and a ``biased random walk.'' In this paper we build from the formalism of the biased random walk.} of that calculation; more details can be found in \autorefapp{app:noiselessanalysisframework}.  First of all, the initial configuration $\vec{\gamma}^{(0)}$ is chosen at random by independently choosing each of the $n$ bits to be $I$ with probability $q/(q+1)$ and $S$ with probability $1/(q+1)$. Then,  for each time step $t$, if the $t$th gate acts on qudits $i_t$ and $j_t$, then the transition from $\vec{\gamma}^{(t-1)}$ to $\vec{\gamma}^{(t)}$ can involve a bit flip at position $i_t$, at position $j_t$, or neither (but not at both), and no bit can flip at any other position. Moreover, $\smash{\gamma_{i_t}^{(t)}}=\smash{\gamma_{j_t}^{(t)}}$ must hold, so if $\smash{\gamma_{i_t}^{(t-1)}} \neq \smash{\gamma_{j_t}^{(t-1)}}$, then one of the two bits \textit{must} be flipped. In this situation, when one bit is assigned $I$ and one is assigned $S$, the $S$ is flipped to $I$ with probability $q^2/(q^2+1)$, and the $I$ is flipped to $S$ with probability $1/(q^2+1)$. Thus, there is a bias toward making more of the assignments $I$.  The quantity $Z_0$ is given exactly by the expectation value of the quantity $q^{|\vec{\gamma}^{(s)}|}$ when trajectories $\gamma$ are generated in this fashion, where $|\vec{\nu}|$ denotes the Hamming weight of the bit string $\vec{\nu}$, that is, the number of $S$ assignments out of $n$. 
\begin{equation}\label{eq:EV0}
    Z_0 = \mathbb{E}_0\left[q^{|\vec{\gamma}^{(s)}|}\right]\,,
\end{equation}
where here $\EV_0$ denotes evolution by the stochastic process described above. 

With the stochastic process now defined, a vital observation is that the process has two fixed points, the $I^n$ configuration and the $S^n$ configuration, since whenever all the bits agree, none can be flipped. In Ref.~\cite{dalzell2020anticoncentration}, we could precisely compute the fraction of the probability mass that eventually reaches each of these fixed points if the circuit is infinitely long. Specifically, $q^n/(q^n+1)$ of the probability mass converges to $I^n$ and $1/(q^n+1)$ converges to $S^n$.\footnote{This can be straightforwardly derived by letting $Q(x)$ be the probability a configuration with $x$ $S$ assignments eventually converges to the $S^n$ fixed point and noting that it satisfies the recursion relation $Q(x) = q^2Q(x-1)/(q^2+1) + Q(x+1)/(q^2+1)$, for which the solution is $Q(x) = A q^{2x} + B$ for constants $A$ and $B$ determined by enforcing boundary conditions $Q(0)=0$ and $Q(n)=1$. The fraction of probability mass that begins at a configuration with $x$ $S$ assignments is $\binom{n}{x}q^{n-x}/(q+1)^n$, allowing the total amount of mass that reaches $S^n$ to be computed.} Then, since the $S^n$ fixed point receives a weighting of $q^n$ and the $I^n$ fixed point receives a weighting of 1 in Eq.~\eqref{eq:EV0}, we find that $Z_0 \rightarrow 2q^n/(q^n+1)$.

Noise introduces new rules into this stochastic process. Suppose the configuration immediately after the $t$th two-qudit gate is $\vec{\nu}$, and a noise channel $\mathcal{N}$ acts on qudit $i_t$. Since the noise channel is unital, if $\nu_{i_t} = I$, representing the identity operator on a two-qudit system, then the configuration is left unchanged. However, if $\nu_{i_t} = S$, then the action of the noise may cause a flip from $S$ to $I$. For the calculation of $Z_0$, there is no noise, so this happens with probability 0.
For the calculation of $Z_1$, where there is one copy of the noisy distribution and one copy of the ideal, we can again use the formula in Eq.~\eqref{eq:Midentityswap} to compute the $S \rightarrow I$ transition probability to be $rq/(q-1)$, where $r$ is the average infidelity given in Eq.~\eqref{eq:averageinfidelity}. This is explained in \autorefapp{app:effectofnoise}. 
For $Z_2$, where there are two copies of the noisy distribution, the probability of an $S \rightarrow I$ transition is calculated to be $1-u$, where $u$ is the unitarity of the noise channel given in Eq.~\eqref{eq:unitarity}. The values of $Z_1$ and $Z_2$ are thus given by
\begin{align}
    Z_1 &= \mathbb{E}_{rq/(q-1)}\left[q^{|\vec{\gamma}^{(s)}|}\right] \\
    Z_2 &= \mathbb{E}_{1-u}\left[q^{|\vec{\gamma}^{(s)}|}\right]\,,
\end{align}
where $\EV_\sigma$ denotes the stochastic process where $S\rightarrow I$ bit flips occur at each noise location with probability $\sigma$, generalizing Eq.~\eqref{eq:EV0}. 

Since noise can flip an $S$ to an $I$ but not vice versa, $I^n$ is the only fixed point of the stochastic processes for $Z_1$ and $Z_2$; the $S^n$ fixed point is only metastable: eventually, the action of noise will flip one of the $S$ bits to an $I$, and the trajectory might re-equilibrate to the $I^n$ fixed point. Our analysis consists of a careful accounting of the leakage of probability mass away from the metastable $S^n$ fixed point.

\paragraph{Analyzing the stochastic processes for a toy example} Now, we consider a toy example which captures the essence of our analysis. Suppose a circuit consists of alternating rounds of (1) a global Haar-random transformation and (2) a depolarizing noise channel on a single qudit, as depicted in \autoref{fig:noisycircuittoymodel}. Step (1) can be approximately accomplished by performing a very large number of two-qudit gates. 
\begin{figure}[h]
\centering
\scalebox{0.9}{
\begin{tikzpicture}[scale=1.0,thick]
\foreach[count=\i] \y in {1,2,3,4,5}
{\tikzmath{\j=int(\i-4);\jx=int(10-\i);\xoff = 0.2;}
\draw (1.6,\y) -- (12.2,\y);
\node at (1.28,\y) {$\ket{0}$};
\draw[fill=white,rounded corners] (12+\xoff,\y-0.3) rectangle (12.6+\xoff,\y+0.3);
\draw[-latex,semithick] (12.3+\xoff,\y-0.22) -- (12.46+\xoff,\y+0.22);
\draw[semithick] (12.1+\xoff,\y-0.16) to[out=60,in=120] (12.5+\xoff,\y-0.16);
}
\foreach \x/\y in {
3.3/4,
5.1/2,
6.9/3,
8.7/1,
10.5/5}
{\draw[fill=white] (\x, \y) circle (0.25);
\node at (\x,\y) {\scalebox{0.65}{$\mathcal{N}$}};}
\foreach[count=\i] \x in {2, 3.8, 5.6, 7.4, 9.2,11.0}
{\draw[fill=white,rounded corners] (\x,0.7) rectangle (\x+0.8,5.3);
\node at (\x+0.42,3) {\scalebox{0.86}{$U^{(\i)}$}};
}
\end{tikzpicture}
}
\caption[Simplified example of a noisy random quantum circuit that is easier to analyze]{Toy example where global Haar-random gates $U^{(t)}$ act in between a depolarizing noise channel on a single qudit. In this model we can exactly compute quantities $Z_0$, $Z_1$, and $Z_2$ because the global Haar-random gates cause the probability mass in the stochastic process to fully re-equilibrate to one of the fixed points, $I^n$ or $S^n$. 
}
\label{fig:noisycircuittoymodel}
\end{figure}
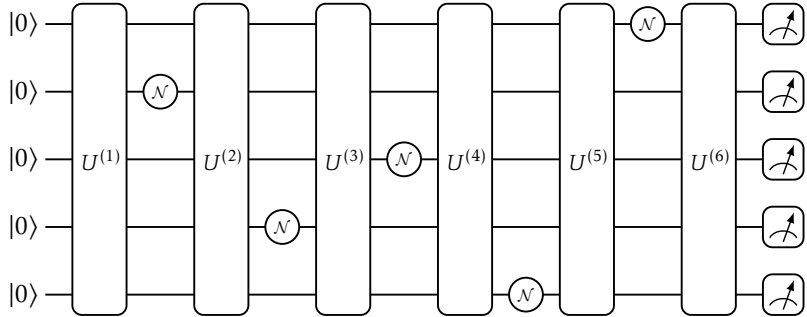
This model is similar to the toy model considered in Ref.~\cite{Bouland2021NoiseQuantumSupremacy} (the difference being that they considered single-qudit noise channels on all $n$ qudits in step (2)), which they analyzed using the Pauli string method of Refs.~\cite{DahlstenOliveiraPlenio2007TypicalEntanglement,HarrowLow2009RQC2design}.

The initial global Haar-random transformation induces perfect equilibration to the two fixed points, with $q^n/(q^n+1)$ mass reaching the $I^n$ fixed point and $1/(q^n+1)$ mass reaching the (metastable) $S^n$ fixed point. This is already sufficient to compute $Z_0-1$, which is not sensitive to the noise.
\begin{equation}
Z_0-1 = \frac{q^n-1}{q^n+1}    \,.
\end{equation}
Now suppose we want to calculate $Z_1$. Consider a piece of probabiltiy mass that is part of the $1/(q^n+1)$ fraction at the $S^n$ fixed point. The single-qudit depolarizing noise channel will flip one of the $S$ assignments to an $I$ assignment with probability $rq/(q-1) = \epsilon(1-q^{-2})^{-1}$. If this happens, there are $n-1$ $S$ assignments and 1 $I$ assignment. While it may seem that this new configuration is still close to the $S^n$ fixed point, we must remember that the random walk is biased in the $I$ direction. When we perform the next global Haar-random transformation, we get perfect re-equilibration back to the two fixed points; with probability $\frac{1-q^{-2}}{1-q^{-2n}}$ we end at the $I^n$ fixed point, and with probability $\frac{q^{-2}-q^{-2n}}{1-q^{-2n}}$ we end at the $S^n$ fixed point. These probabilities were derived in Ref.~\cite{dalzell2020anticoncentration}, and are a basic consequence of Eq.~\eqref{eq:Midentityswap}. Now, the total mass that remains at the $S^n$ fixed point is the $\frac{1}{q^n+1}(1-\frac{\epsilon}{1-q^{-2}})$ that never left and the $\frac{\epsilon}{1-q^{-2}}\frac{q^{-2}-q^{-2n}}{1-q^{-2n}}$ that left and returned, which comes out to $\frac{1}{q^n+1}(1-\frac{\epsilon}{1-q^{-2n}})$. After $2s$ single-qudit error channels have been applied, the probability mass remaining at the $S^n$ fixed point is precisely
\begin{equation}
        \substack{ \text{probability mass at } S^n \\ \text{after } 2s \text{ noise locations}} = \frac{1}{q^n+1}\left(1-\frac{\epsilon}{1-q^{-2n}} \right)^{2s} \approx \frac{1}{q^n+1}e^{-2\epsilon s}\,.
\end{equation}
This mass receives weighting of $q^n$ toward $Z_1$. Meanwhile the rest of the mass is at the $I^n$ fixed point and receives weighting of 1. This tells us that
\begin{equation}
    Z_1-1 = \frac{q^n-1}{q^n+1}\left(1-\frac{\epsilon}{1-q^{-2n}} \right)^{2s}\,.
\end{equation}
We see that in this toy model, the quantity $\bar{F} = (Z_1-1)/(Z_0-1)$ is precisely given by the fraction of probability mass originally destined for the $S^n$ fixed point that remains at the $S^n$ fixed point even after the noise locations have acted. Thus, the leakage of probability mass from $S^n$ to $I^n$ in the calculation of $Z_1$ corresponds exactly to the decay of fidelity. 

Calculating $Z_2-1$ is just as easy. Here transitions due to noise occur with probability $1-u$ where $u$ is the unitarity of the noise channel. For depolarizing noise, we have $1-u = 2\epsilon(1-q^{-2})^{-1} - O(\epsilon^2)$, so $Z_2-1$ is the same as $Z_1-1$ with the replacement $\epsilon \rightarrow 2\epsilon - O(\epsilon^2)$, giving
\begin{equation}
    Z_2-1 = \frac{q^n-1}{q^n+1}\left(1-\frac{2\epsilon}{1-q^{-2n}} +O(\epsilon^2)\right)^{2s} = \frac{q^n-1}{q^n+1}\left(1-\frac{\epsilon}{1-q^{-2n}} \right)^{4s} e^{O(s\epsilon^2)}\,.
\end{equation}

We can plug these calculations into Eq.~\eqref{eq:boundUsingOptimalF} to find that
\begin{equation}
    \EV_U\left[ \frac{1}{2}\lVert \pwhitenoise-\pnoisy \rVert_1\right] \leq \frac{1}{2} \bar{F} \sqrt{\frac{q^n-1}{q^n+1}\left(e^{O(\epsilon^2s)}-1\right)} = O(\bar{F}\epsilon \sqrt{s})\,.
\end{equation}

\paragraph{Extending the analysis to a full proof} In the proofs of our theorems, the difficulty is that the probability mass does not fully equilibrate to a fixed point before the next error location acts. Nonetheless, we manage to calculate tight bounds on $Z_1$ and $Z_2$ by keeping track of the amount of probability mass that \textit{would} re-equilibrate back to $S^n$ and $I^n$ if the rest of the gates were noiseless, which we refer to as $S$-destined and $I$-destined probability mass. We show that, as long as $\epsilon < c/n$ for some constant $c$, the $S$-destined probability mass is exponentially clustered near the $S^n$ fixed point in the sense that the probability of being $x$ bit flips away from $S^n$ conditioned on being $S$-destined decays exponentially in $x$. Thus, for a piece of $S$-destined probability mass, nearly all the bits will be assigned $S$, and the action of a noise channel reduces the amount of $S$-destined mass by a factor of roughly $1-\epsilon$. If it were the case that a constant fraction of bits were assigned $I$, then the noise would cause a flip from $S \rightarrow I$ less frequently and the fraction of the $S$-destined mass that stays $S$-destined after each noise channel would be larger than $1-\epsilon$ by an $O(\epsilon)$ amount, which would ruin the analysis.

The reason $\epsilon < c/n$ is required for the exponential clustering effect is that errors need to be rare enough for the $S$-destined mass to \textit{mostly} re-equilibrate back to $S^n$ before new errors pop up; to say it another way, the errors must get scrambled at a faster rate than they appear. If a configuration has $n-1$ $S$ assignments and 1 $I$ assignment, it will take $O(n)$ gates before the single $I$-assigned qudit participates in a gate. Thus, if errors occur at a slower rate than one per $O(n)$ gates, full re-equilibration will happen before a new error pops up most of the time. It is not clear if this condition is truly necessary for the clustering statement to hold, but we show at the very least that it is sufficient.

However, we need $\epsilon < c/n$ to hold for another (related) reason: the leakage from $S^n$ to $I^n$ must occur more slowly than the anti-concentration rate, which corresponds to the speed at which the probability mass initially equilibrates to $I^n$ and $S^n$. After all, even though the stochastic process is $I$-biased, the $I$-destined mass does not make it to the $I^n$ fixed point instantaneously. After $s$ gates, there will be some residual contribution from the not-yet-equilibrated $I$-destined mass to the calculation of quantities $Z_0-1$, $Z_1-1$, and $Z_2-1$; this contribution decays by a constant factor with every additional $O(n)$ gates. If $\epsilon =O(1/n)$, a constant fraction of the $S$-destined mass will leak away with each set of $O(n)$ gates, and if the constant prefactor on this leakage is too large, the $I$-destined mass will contribute more than the $S$-destined mass to the expectation values; as a result, the right-hand-side of Eq.~\eqref{eq:boundUsingOptimalF} will not exhibit the same kind of cancellations observed for the toy example.

In our formal analysis, we actually assume something even stronger: we require that $\epsilon \ll 1/(n\log(n))$, which essentially means that very few errors occur \textit{during} the initial anti-concentration period. However, this is done to make the analysis easier, and we do not believe this condition is necessary.

\section{Numerical estimates of error in white-noise approximation}\label{sec:numerics}

In principle, it would be possible to determine the constant factors under the big-$O$ notation in our proofs, but the result of this exercise  would likely yield extremely unfavorable numbers due to our lack of optimization throughout, and the fact that it might be possible to eliminate some of the terms in our error expression altogether with a more fine-grained analysis. The goal of this section is to provide a numerical assessment of the bound on the error in the white-noise approximation for realistic values of the circuit parameters. We find that realistic NISQ-era values of the circuit parameters \textit{can} lead to a small upper bound on the white-noise approximation error, even for circuits with several thousand gates, but we confirm that the noise rate needs to decrease like $O(1/n)$ as the system size scales up for our upper bound to be meaningful.

\subsection{Numerical method}
The numerics we present are for the complete-graph architecture. In general, the stochastic process underlying our method (described in \autoref{sec:overviewnoisywalk} and presented formally in the appendix) is a random walk over $2^n$ possible configurations of a length-$n$ bit string. However, for the complete-graph architecture there is an equivalence between all configurations with the same Hamming weight. Thus, the state space for the stochastic process is reduced to $n+1$ distinct groups of configurations (associated with Hamming weights $0,1,\ldots,n$). The quantities $Z_0$, $Z_1$, and $Z_2$, as defined in Eqs.~\eqref{eq:Z0}, \eqref{eq:Z1}, and \eqref{eq:Z2} can then be precisely computed by multiplying the (sparse) $(n+1) \times (n+1)$ transition matrices for the stochastic process. This allows us to compute the right-hand-side of Eq.~\eqref{eq:boundUsingOptimalF} for $n$ substantially large, giving a bound on  $\EV_U[\frac{1}{2}\lVert \pnoisy-\pwhitenoise \rVert_1]$. 

In our analysis below, we suppose all noise locations are subject to depolarizing noise with error probability $\epsilon$, given as in Eq.~\eqref{eq:depolarizing}. We also restrict to $q=2$ (qubits). We do not model readout errors, which are a large source of error in the actual experiments of Refs.~\cite{Arute2019GoogleQuantumSupremacy,USTC2021StrongQCompAdv,USTC2021Zhuchongzhi2.1}. We plug in specifications $(n,\epsilon,s)$ and exactly compute the quantity
\begin{equation}
    \frac{1}{2}\sqrt{(Z_0-1)\left(\frac{(Z_0-1)(Z_2-1)}{(Z_1-1)^2}-1\right)}
\end{equation}
which gives the ratio of the bound in Eq.~\eqref{eq:boundUsingOptimalF} to the fidelity $\bar{F}$. 

\subsection{Numerical bound for realistic circuit parameters}

 We first examine the bound using the circuit parameters of existing experimental setups. The Google experiment \cite{Arute2019GoogleQuantumSupremacy} ran $s=430$ gates on their $n=53$ qubit processor called \textit{Sycamore}, and their error rate per cycle, which is the analogous quantity to the total error in a two-qubit gate in our setup, was reported to be $0.9\%$. This corresponds to $\epsilon\approx0.0045$ in our model where separate noise channels act on each of the two qubits. Meanwhile, the largest experiment from USTC \cite{USTC2021Zhuchongzhi2.1} ran $s=594$ gates on their $n=60$ qubit processor called \textit{Zuchongzhi}, with a similar overall error rate per cycle.  In \autoref{fig:GoogleNumerics}, we plot the numerically calculated bound on $\frac{1}{\bar{F}}\EV_U[\frac{1}{2}\lVert \pnoisy-\pwhitenoise\rVert_1]$ as a function of circuit size for complete-graph circuits with $n=53$ and $n=60$ at $\epsilon=0.0045$. The circuit sizes $s=430$ and $s=594$ appear as large dots. 
 \begin{figure}
     \centering
     \includegraphics[width=0.7\textwidth]{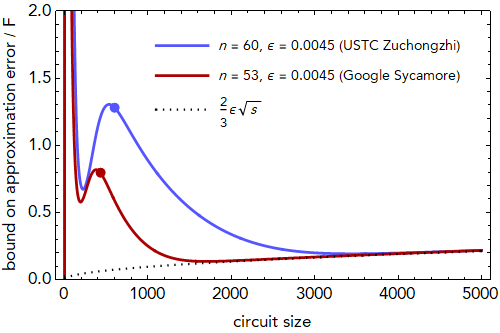}
     \caption{Plot of the numerically calculated upper bound on the expected total variation distance between $\pnoisy$ and $\pwhitenoise$ divided by $F$ for a complete-graph version of recent random quantum circuit experiments by Google (53 qubits) \cite{Arute2019GoogleQuantumSupremacy} and USTC (60 qubits) \cite{USTC2021Zhuchongzhi2.1}. The large dots represent the circuit sizes (number of two-qubit gates) implemented in those experiments. The dotted black line is the function $2\epsilon\sqrt{s}/3$ for each experiment. }
     \label{fig:GoogleNumerics}
 \end{figure}
 We find that, as expected, the bound is bad if the circuit size is too small. There is an initial spike in the bound due to the first few layers of noisy gates, which subsides quickly as those initial errors are scrambled. The behavior that follows reflects the race between fidelity decay and anti-concentration. For these values of the error rate, the fidelity decay is happening at a slower rate than anti-concentration, but it has a head start, since it takes $\Theta(n\log(n))$ gates for anti-concentration to initially be reached \cite{dalzell2020anticoncentration}; this explains why the bound is decreasing (relative to $F$) even as the circuit size passes 1000. For large $s$, both curves approach the function $2\epsilon\sqrt{s}/3$. This indicates that the constant factor underneath the $O(\epsilon \sqrt{s})$ is less than 1, at least for depolarizing noise in the complete-graph architecture. The point at which we expect the $O(\epsilon \sqrt{s})$ behavior to take over will generally be $\Theta(n\log(n)) + \Theta(n)$, where the first term corresponds to the initial anti-concentration period, and the second term corresponds to the additional time needed for anti-concentration to catch up to the fidelity. The constant prefactor under the second term will be larger when $\epsilon$ is larger and the fidelity decays more rapidly.
 
 Interestingly, the circuit size actually implemented in both of the experiments falls in a region where the bound on approximation error relative to fidelity is decreasing with circuit size, suggesting the white-noise approximation would become more meaningful if more gates were applied (at the expense of smaller fidelity). In fact, for Google's experiment, the upper bound yields a value close to 1, and for USTC, it yields a value larger than 1, indicating that, in this idealized complete-graph version of their experiments, the white-noise assumption may not hold (we would need a lower bound to know for sure). 
 
 There are a few caveats to these conclusions. First, what we plot is only an upper bound, and it is not clear whether this upper bound is tight. Second, this is for the complete-graph architecture, but the experiments of Refs.~\cite{Arute2019GoogleQuantumSupremacy,USTC2021StrongQCompAdv,USTC2021Zhuchongzhi2.1} had a 2D architecture (although one might speculate that a 2D architecture would only scramble less efficiently than the complete-graph architecture). Third, we have not modeled readout errors in the device. Fourth, we have an idealized error model of depolarizing single-qubit noise. As has been mentioned in footnotes throughout this paper, the goal of our work is not to justify the claims of quantum computational supremacy by specific noisy random quantum circuit experiments. Rather, we aim to show that the white-noise phenomenon is possible and can be proved analytically, and that this adds some justification to claims that a low-fidelity random quantum circuit experiment could in principle accomplish quantum computational supremacy. 
 
 \subsection{Threshold error rate for good white-noise bound}

 \begin{figure}
     \centering
     \subfloat[$n=53$]{{\includegraphics[width=0.47\textwidth]{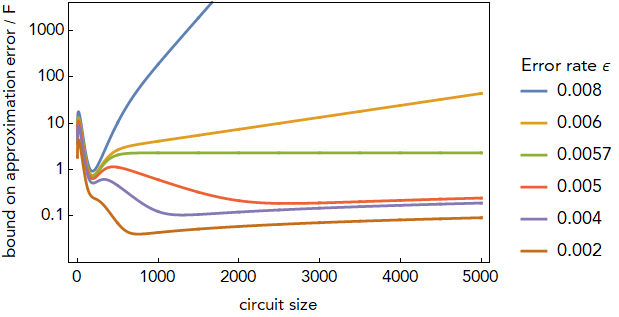}}}  \qquad
     \subfloat[$n=106$]{{\includegraphics[width=0.47\textwidth]{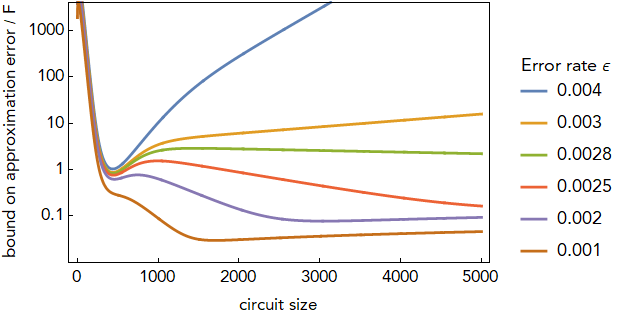}}} \\
     \subfloat[$n=159$]{{\includegraphics[width=0.47\textwidth]{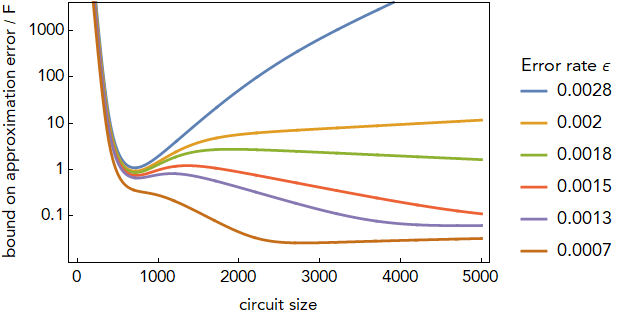}}} \qquad
     \subfloat[$n=212$]{{\includegraphics[width=0.47\textwidth]{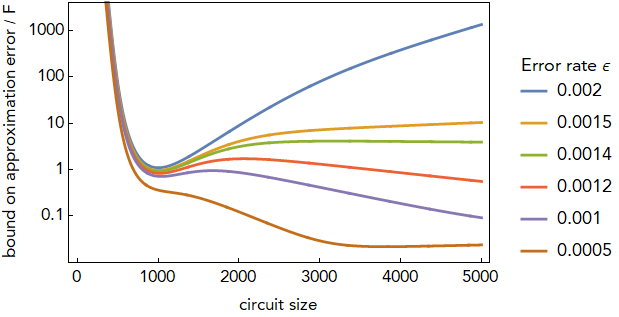}}}
     \caption{Plot of the numerically calculated upper bound on the expected total variation distance between $\pnoisy$ and $\pwhitenoise$ divided by $F$ for the complete-graph architecture at various values of $n$, $\epsilon$ and $s$. For each value of $n$, a threshold in $\epsilon$ is observed where error rates above the threshold lead to a bad approximation, while error rates below the threshold lead the approximation to become $O(F\epsilon \sqrt{s})$ once $s$ is sufficiently large. The threshold value of $\epsilon$ appears to be roughly $0.3/n$. }
     \label{fig:thresholdplots}
 \end{figure}

A key feature we observed in our theoretical analysis was the need for the error rate $\epsilon$ to decrease with $n$.  For each value of $n$, we observe a threshold error rate such that, if $\epsilon$ is beneath the threshold, our upper bound on the total variation distance follows $O(F\epsilon \sqrt{s})$ at large values of $s$, and if $\epsilon$ is above the threshold, our bound becomes (empirically) $O(Fe^{\Theta(s)})$. Without a lower bound, we cannot be sure if this is the actual behavior of the approximation error.

In \autoref{fig:thresholdplots}, we present a log plot of the numerically calculated bound on the approximation error (relative to $F$) for different values of $\epsilon$ at system sizes $n=53, 106, 159, 212$ (corresponding to integer multiples of the size of Google's 53-qubit experiment).  For $n=53$, we see that choices of $\epsilon$ beneath roughly $0.0057$ appear to approach $O(\epsilon\sqrt{s})$ scaling at large $s$, while choices of $\epsilon$ above that threshold increase exponentially with $s$. For $n=106$, $n=159$, and $n=212$, the apparent threshold decreases to roughly  $\epsilon=0.0028$, $\epsilon=0.0019$, and $\epsilon=0.0014$, respectively. This is consistent with a general threshold of roughly $\epsilon = 0.3/n$. We expect the $\epsilon = O(1/n)$ threshold to exist in other architectures as well, but with a modified constant prefactor. Architectures with a faster anti-concentration rate should have larger thresholds.

\section{Outlook}\label{sec:whitenoiseoutlook}

We have presented a comprehensive picture of how the output distribution of typical random quantum circuits behaves under a weak incoherent local noise model. As more gates are applied, the output distribution decays toward the uniform distribution in total variation distance like $e^{-2\epsilon s}$ where $\epsilon$ is the local noise strength in a Pauli error model (for non-Pauli models, this can be expressed in terms of the average infidelity $r$) and $s$ is the number of gates. Moreover, we show that the convergence to uniform happens in a very special way: the residual non-uniform component of the noisy distribution is approximately in the direction of the ideal distribution. The random quantum circuits scramble the errors that occur locally during the evolution so that they can ultimately be treated as global white noise, allowing some signal of the ideal computation to be extracted even from a noisy device. While this property had previously been conjectured---it was an underlying assumption of quantum computational supremacy experiments \cite{Arute2019GoogleQuantumSupremacy,USTC2021StrongQCompAdv}---it had not received rigorous analytical study. Basic questions like how the error in the white-noise approximation scales with $\epsilon$ and $s$ had not been investigated.

Our theorem statements are given for general, possibly coherent, noise channels. While we show that local coherent noise channels lead the output distribution to exhibit exponential decay in the linear cross-entropy benchmark for the fidelity, there is not generally also a decay toward the uniform distribution. As a result, the white-noise approximation is not good for coherent noise channels. Moreover, even for incoherent noise channels, our technical statements are only applicable if the Pauli noise strength $\epsilon$ (or for non-Pauli noise channels, the average infidelity) is beneath a threshold that shrinks with system size like $O(1/n)$ and if the circuit size is at least $\Omega(n\log(n))$. Furthermore, our bound on error in the white-noise approximation is only meaningful if $\epsilon \ll 1/(n\log(n))$. We believe the $\epsilon \ll 1/(n\log(n))$ requirement is merely a result of suboptimal analysis, but that the assumption $\epsilon < O(1/n)$ is fundamentally necessary for the approximation to be good: errors must be scrambled faster than the fidelity $F \approx e^{-2\epsilon s}$ decays. 

One implication of our result is to put low-fidelity random-circuit-based quantum computational supremacy experiments on stronger theoretical footing by showing that, as long as our local noise model is a reasonable approximation of noise in actual devices, the device produces samples from a well-understood output distribution, which can subsequently be argued is hard to classically sample. Indeed, in \autorefapp{app:complexitytheorywhitenoise}, we combine observations from previous work to show that the task of classically sampling from the white-noise distribution with fidelity $F$ up to $\eta F$ error is essentially just as hard, in a certain complexity-theoretic sense, as the task of classically sampling from the ideal distribution up to a $O(\eta)$ error. This is important because the latter task  (and variants of it in other computational models \cite{Aaronson2011BosonSampling,Bremner2016AverageCaseIQP}) has previously garnered significant theoretical scrutiny  \cite{Bouland2019RCSComplexity,Movassagh2019QSandRQC,Bouland2021NoiseQuantumSupremacy}, although it is still not known whether it is hard in a formal complexity-theoretic sense. 

These results are good news for the utility of NISQ devices more broadly. In order to perform a larger and more interesting computation, noise rates must become smaller; our work shows that, in some applications, for circuits with $s$ gates, noise rates need only decrease like $1/\sqrt{s}$, rather than $1/s$, as long as one is willing to repeat the experiment many times to extract the signal from the global white noise. A natural next question is when, besides the case of random quantum circuits, do we expect a similar white-noise phenomenon to occur? Our result shows that convergence to white-noise is a \textit{generic} property, occurring for a large fraction of randomly chosen circuits. Heuristically, this is because random quantum circuits are known to be good scramblers. However, most interesting quantum circuits are non-generic in some way. An extreme example is quantum error-correcting circuits, which are specifically designed \textit{not} to scramble errors (so that they can be corrected). The output of these circuits will not be close to the white-noise distribution. A fascinating follow-up question is whether other computations proposed for NISQ devices appear to scramble errors well enough that a similar approximation can be made. One leading candidate with relevance for many-body physics is circuits that simulate evolution by fixed chaotic Hamiltonians, since these systems are thought to scramble information efficiently. Indeed, a central motivation for studying random quantum circuits in the first place has been to model the scrambling properties of chaotic many-body systems \cite{NahumRuhmanVijayHaah2017EntanglementGrowth,NahumVijayHaah2018OperatorSpreading,VonKeyserlingk2018OperatorHydrodynamics}.

\subsection*{Acknowledgments}
We thank Adam Bouland, Bill Fefferman, Zeph Landau, Yunchao Liu, Oskar Painter, John Preskill, and Thomas Vidick for helpful feedback about this work. 
AD and FB acknowledge funding provided by the Institute for Quantum Information and Matter, an NSF Physics Frontiers Center (NSF Grant PHY-1733907). This material is also based upon work supported by the NSF Graduate Research Fellowship under Grant No.~DGE‐1745301. 
NHJ is supported in part by the Stanford Q-FARM Bloch Fellowship in Quantum Science and Engineering. NHJ would like to thank the Aspen Center for Physics for its hospitality during the completion of part of this work.
Research at Perimeter Institute is supported in part by the Government of Canada through the Department of Innovation, Science and Economic Development Canada and by the Province of Ontario through the Ministry of Colleges and Universities.

\appendix

\section{Framework for noisy circuit analysis}\label{app:noisyframework}

\subsection{Action of averaged noiseless gate on identity and swap}\label{app:noiselessanalysisframework}

The contents of this subsection contain analysis from Ref.~\cite{dalzell2020anticoncentration}, which we include again here for completeness. We also slightly modify the notation from Ref.~\cite{dalzell2020anticoncentration} so that the two-qudit identity operator $I$ is always normalized by $q^2$ and the two-qudit swap operator $S$ is always normalized by $q$, such that their traces are one. 

Since we study second-moment properties, we work with two copies of the $n$-qudit state. The initial state is $\ketbra{0^n}^{\otimes 2}$. Suppose the gate at time step $t$ acts on qudits in the set $A^{(t)} \subset [n]$ (of size either 1 or 2), and let
\begin{equation} \label{eq:M(t)}
M^{(t)}[\rho] = \EV_{U^{(t)}}\left[{U^{(t)}_{A^{(t)}}}^{\otimes 2} \rho\, {{U^{(t)}_{A^{(t)}}}^\dagger}^{\otimes 2}\right]\,,
\end{equation}
where the average is over Haar-random choice of $U^{(t)}$ and $U^{(t)}_{A^{(t)}}$ denotes the operation that acts as $U^{(t)}$ on qudits in region $A^{(t)}$ and as identity on all other qudits. 

Application of the first layer of $n$ single-qudit gates in \autoref{fig:noisycircuitdiagram} corresponds to application of $M^{(-n+1)}\circ \cdots \circ M^{(0)}$ to the initial state $\ketbra{0^n}^{\otimes 2}$. Applying the Haar integration formula in Eq.~\eqref{eq:Midentityswap} to each qubit, we find
\begin{equation}
    M^{(-n+1)}\circ \cdots \circ M^{(0)}[\ketbra{0^n}] = \frac{1}{q^n(q+1)^n}\bigotimes_{j=0}^{n-1} \left(I + S\right)_{\{j\}} = \bigotimes_{j=0}^{n-1} \left(\frac{q}{q+1}\frac{I}{q^2} + \frac{1}{q+1}\frac{S}{q}\right)_{\{j\}}\,,
\end{equation}
where the second equality expresses the formula as a linear combination of $I/q^2$ and $S/q$, both of which have trace one. The coefficients $q/(q+1)$ and $1/(q+1)$ are interpreted as probabilities that each bit of the initial configuration $\vec{\gamma}^{(0)}$ as described in \autoref{sec:overviewnoisywalk} is $I$ or $S$, respectively. 

Since the averaged state is a linear combination of tensor products of $I$ and $S$ already after the first layer, we need only compute the action of an averaged two-qudit gate on $I \otimes I$, $I \otimes S$, $S\otimes I$, and $S \otimes S$, properly normalized. Suppose gate $t$ acts on qudits $\{i_t,j_t\}$. Then $M^{(t)}$ acts trivially on all qudits outside of $\{i_t,j_t\}$ and its action on $\{i_t,j_t\}$ is computed using the Haar integration formula in Eq.~\eqref{eq:Midentityswap} (note that since the gates are $q^2 \times q^2$ matrices, we replace $q$ by $q^2$, $I$ by $I \otimes I$, and $S$ by $S \otimes S$), yielding
\begin{align}
    M^{(t)}\left[\frac{I}{q^2} \otimes \frac{I}{q^2}\right] &= \frac{I}{q^2} \otimes \frac{I}{q^2} \label{eq:MII} \\
    M^{(t)}\left[\frac{S}{q} \otimes \frac{S}{q}\right] &= \frac{S}{q} \otimes \frac{S}{q} \label{eq:MSS}\\
    M^{(t)}\left[\frac{I}{q^2} \otimes \frac{S}{q}\right] = M^{(t)}\left[\frac{S}{q} \otimes \frac{I}{q^2}\right] &=  \frac{q^2}{q^2+1}\frac{I}{q^2} \otimes \frac{I}{q^2} + \frac{1}{q^2+1}\frac{S}{q} \otimes \frac{S}{q} \label{eq:MIS}
\end{align}
The above equations correspond to the transition rules for the noiseless stochastic process mentioned in \autoref{sec:overviewnoisywalk}: if both bits are $I$ or both are $S$, then there is no change, but if one is $I$ and one is $S$, they are both set to $I$ with probability $q^2/(q^2+1)$ and both set to $S$ with probability $1/(q^2+1)$. 

This illustrates that sequential application of $M^{(t)}$ on the state will map linear combinations of tensor products of $I/q^2$ and $S/q$ to other linear combinations of tensor products of $I/q^2$ and $S/q$. The coefficients of these linear combinations transform linearly. When written in terms of the trace-one operators $I/q^2$ and $S/q$, this linear transformation will be stochastic, i.e.~the sum of the coefficients of the linear combination over tensor products will be conserved (note that the sum of coefficients in Eqs.~\eqref{eq:MII}, \eqref{eq:MSS}, and \eqref{eq:MIS} is one). 
Now, let us associate the configuration $\vec{\nu} \in \{I,S\}^n$ by the tensor product $\bigotimes_{j=0}^{n-1} \frac{\nu_j}{\tr(\nu_j)}$, which is a basis state for the vector space acted upon by $M^{(t)}$. For configurations $\vec{\nu}, \vec{\gamma} \in \{I,S\}^n$, denote the matrix elements of this (stochastic) transformation by  $M^{(t)}_{\vec{\nu}\vec{\gamma}}$, that is
\begin{equation}
    M^{(t)}\left[\bigotimes_{j=0}^{n-1} \frac{\gamma_j}{\tr(\gamma_j)}\right] = \sum_{\vec{\nu} \in \{I,S\}^n}M_{\vec{\nu}\vec{\gamma}}^{(t)}\bigotimes_{j=0}^{n-1} \frac{\nu_j}{\tr(\nu_j)} \,.
\end{equation}
The matrix elements are given explicitly by
\begin{equation}
    M^{(t)}_{\vec{\nu}\vec{\gamma}}=\begin{cases}
    1 & \text{if } \gamma_{i_t} = \gamma_{j_t} \text{ and } \vec{\gamma} = \vec{\nu} \\
    \frac{q^2}{q^2+1} & \text{if } \gamma_{i_t} \neq \gamma_{j_t} \text{ and } \nu_{i_t} = \nu_{j_t}=I \text{ and } \gamma_c = \nu_c \;\forall c \in [n]\setminus\{i_t,j_t\} \\
     \frac{1}{q^2+1} & \text{if } \gamma_{i_t} \neq \gamma_{j_t} \text{ and } \nu_{i_t} = \nu_{j_t}=S \text{ and } \gamma_c = \nu_c \;\forall c \in [n]\setminus\{i_t,j_t\} \\
    0 & \text{otherwise}
    \end{cases}
\end{equation}
Now, note that
\begin{align}
    \tr\left[\ketbra{0}^{\otimes 2} \frac{I}{q^2}\right] &= \frac{1}{q^2} \\
    \tr\left[\ketbra{0}^{\otimes 2} \frac{S}{q}\right] &= \frac{1}{q}
\end{align}
so, for $\vec{\nu} \in \{I,S\}^n$,
\begin{align}
    \tr\left[\ketbra{0^n}^{\otimes 2} \bigotimes_{j=0}^{n-1}\frac{\nu_j}{\tr(\nu_j)} \right]= \frac{q^{|\vec{\nu}|}}{q^{2n}}\,,
\end{align}
where $|\vec{\nu}|$ denotes the Hamming weight of the bit string $\vec{\nu}$, that is, the number of $S$ assignments. Working now from the definition of $Z_0$ in Eq.~\eqref{eq:Z0} and $\pideal$ in Eq.~\eqref{eq:pideal}, we have the matrix equation
\begin{equation}\label{eq:Z0asmatrixproduct}
    Z_0 = q^{2n}\tr\left[\ketbra{0^n} M^{(s)} \circ \cdots M^{(-n+1)} (\ketbra{0^n})\right] = \sum_{\gamma \in \{I,S\}^{n\times(s+1)}} \frac{q^{n-|\vec{\gamma}^{(0)}|}}{(q+1)^n}\left(\prod_{t=1}^s M^{(t)}_{\vec{\gamma}^{(t)}\vec{\gamma}^{(t-1)}}\right) q^{|\vec{\gamma}^{(s)}|}\,.
\end{equation}
The $q^{n-|\vec{\gamma}^{(0)}|}/(q+1)^n$ factor is the probability of starting in $\vec{\gamma}^{(0)}$. Thus, this can be re-expressed as
\begin{equation}
    Z_0 = \mathbb{E}_0\left[q^{|\vec{\gamma}^{(s)}|}\right]\,,
\end{equation}
where $\mathbb{E}_0$ denotes expectation over the stochastic process that generates the trajectory $\gamma = (\gamma^{(0)},\ldots,\gamma^{(s)})$, as described above, and as concluded in Eq.~\eqref{eq:EV0} of \autoref{sec:overviewnoisywalk}.  In Ref.~\cite{dalzell2020anticoncentration}, this stochastic process was termed the ``biased random walk.''

\subsection{Action of averaged noise channel on identity and swap}\label{app:effectofnoise}

Since every single-qudit noise channel is followed by a Haar-random (either single-qudit or two-qudit) gate in the circuit diagram, we are free to add a single-qudit Haar-random gate immediately after every noise channel without changing the overall circuit ensemble (the Haar measure is invariant under multiplication by any unitary). Denote this single-qudit Haar-random matrix by $V$. 
There will be a difference in the analysis between the calculation of $Z_0$, $Z_1$ and $Z_2$, where $Z_w$ contains $w$ copies of the noisy output as defined in Eqs.~\eqref{eq:Z0}, \eqref{eq:Z1}, \eqref{eq:Z2}. Define
\begin{align}
    \mathcal{N}_0 &= \mathcal{I} \otimes \mathcal{I} \\
    \mathcal{N}_1 &= \mathcal{I} \otimes \mathcal{N} \\
    \mathcal{N}_2 &= \mathcal{N} \otimes \mathcal{N}
\end{align}
with $\mathcal{I}$ denoting the single-qudit identity channel. Let $\rho$ be a state on two copies of a single-qudit Hilbert space. Then for $w \in \{0,1,2\}$, let
\begin{align}
    N_w[\rho] &= \EV_V\left[V^{\otimes 2} \; \mathcal{N}_w(\rho) V^{\dagger\otimes 2}\right]
\end{align}
be the Haar-averaged noise channel. 

We will only need to compute the action of $N_w$ on input states $\rho = I/q^2$ (here $I$ is the two-qudit identity operator) or $\rho = S/q$ since, as shown above, the random gates turn the initial state $\ketbra{0^n}$ into a linear combination of tensor products of $I/q^2$ or $S/q$ on each qudit. Note that since $\mathcal{N}$ is assumed to be unital, we have
\begin{align}
    N_w\left[\frac{I}{q^2}\right] &= \frac{I}{q^2}
\end{align}
for all $w \in \{0,1,2\}$. 
However, computing the action on $S/q$ is not as simple. Let
\begin{align}
    Y_w &= \tr\left(S \mathcal{N}_w(S)\right)\,.
\end{align}
(Note that $Y_0 = q^2$ since $\mathcal{N}_0$ is the identity channel.) Then, use Eq.~\eqref{eq:Midentityswap} and the fact that $\mathcal{N}$ is trace-preserving to show
\begin{align}\label{eq:NwintermsofY}
    N_w\left[\frac{S}{q}\right] &= \frac{q^2-Y_w}{q^2-1} \frac{I}{q^2} + \frac{Y_w-1}{q^2-1}\frac{S}{q} \,.
\end{align}
Now we relate the quantities $Y_1$ and $Y_2$ to the average infidelity and the unitarity, respectively. Recall that $\tr(AB) = \tr(S(A \otimes B))$. 
Using this trick and Eq.~\eqref{eq:Midentityswap}, the average infidelity from Eq.~\eqref{eq:averageinfidelity}, can be evaluated as follows:
\begin{align}
    r &= 1-\int dV \tr\left[ V \ketbra{\psi} V^{\dagger}\mathcal{N}(V \ketbra{\psi} V^{\dagger})\right] \\
    &= 1-\int dV \tr\left[S\left( V \ketbra{\psi} V^{\dagger}\otimes \mathcal{N}(V \ketbra{\psi} V^{\dagger})\right)\right] \\
    &= 1-\int dV \tr\left[S(\mathcal{I} \otimes \mathcal{N})\left(\left( V \ketbra{\psi} V^{\dagger}\right)^{\otimes 2}\right)\right] \\
    &= 1-\tr\left[S\mathcal{N}_1 \left(\frac{I+S}{q(q+1)}\right)\right]\\
    &= 1-\frac{1-q^{-1}Y_1}{q+1} = \frac{q-q^{-1}Y_1}{q+1}\,.
\end{align}
The unitarity from Eq.~\eqref{eq:unitarity}, can be evaluated in a similar way.
\begin{align}
    u&=\frac{q}{q-1}\left(\int dV \tr\left[ \mathcal{N}\left(V \ketbra{\psi} V^{\dagger}\right)^2\right]-\frac{1}{q}\right) \\
    &=\frac{q}{q-1}\int dV \tr\left[ S\left(\mathcal{N}\left(V \ketbra{\psi} V^{\dagger}\right)\right)^{\otimes 2}\right]-\frac{1}{q-1} \\
    &= \frac{q}{q-1}\int dV \tr\left[ S(\mathcal{N} \otimes \mathcal{N})\left(\left(V \ketbra{\psi} V^{\dagger}\right)^{\otimes 2}\right)\right]-\frac{1}{q-1} \\
    &=\frac{q}{q-1} \tr\left[ S \mathcal{N}_2\left(\frac{I+S}{q(q+1)}\right)\right]-\frac{1}{q-1} \\
    &=\frac{q +Y_2}{(q-1)(q+1)}-\frac{1}{q-1} \\
    &= \frac{Y_2 -1}{q^2-1} \,.
\end{align}
Plugging these relations back into Eq.~\eqref{eq:NwintermsofY} gives us
\begin{align}
    N_0\left[ \frac{S}{q} \right] &= \frac{S}{q} \\
    N_1\left[ \frac{S}{q} \right] &= \frac{qr}{q-1} \frac{I}{q^2} + \left(1-\frac{qr}{q-1}\right) \frac{S}{q}  \\
    N_2\left[ \frac{S}{q} \right] &= (1-u) \frac{I}{q^2} + u \frac{S}{q} \,.
\end{align}
For weak noise channels, $r$ is close to 0 and $u$ is close to 1. In this case we see that the noise causes some small amount of leakage from the $S$ state to the $I$ state, but no leakage from the $I$ state to the $S$ state, introducing an asymmetry into the problem that did not exist in the noiseless analysis.

For $t=1,\ldots, s$, let $N_w^{(t)} = \mathcal{I}_{[n]\setminus\{i_t\}} \otimes N_{w,\{i_t\}}$ be the channel that acts with the averaged noise channel on site $i_t$ and identity elsewhere, and let $N_w^{\prime\,(t)} = \mathcal{I}_{[n]\setminus\{j_t\}} \otimes N_{w,\{j_t\}}$ be the same for site $j_t$. For $t\leq 0$ and $t > s$, let $N_w^{(t)}$ be the identity channel. If $\rho$ is a linear combination of tensor products of $I/q^2$ and $S/q$, $N_w^{(t)}(\rho)$ and $N_w^{\prime\,(t)}(\rho)$ will be as well, with coefficients that transform linearly (and stochastically). For configurations $\vec{\gamma}, \vec{\nu} \in \{I,S\}^n$, let $N_{w,\vec{\nu}\vec{\gamma}}^{(t)}$ denote the matrix elements of this transformation, that is 
\begin{equation}
    N_w^{(t)}\left[\bigotimes_{j=0}^{n-1} \frac{\gamma_j}{\tr(\gamma_j)}\right] = \sum_{\vec{\nu} \in \{I,S\}^n}N_{w,\vec{\nu}\vec{\gamma}}^{(t)}\bigotimes_{j=0}^{n-1} \frac{\nu_j}{\tr(\nu_j)} \,,
\end{equation}
where for $1 \leq t \leq s$,
\begin{align}
    N_{0,\vec{\nu}\vec{\gamma}}^{(t)} &= 
    \begin{cases}
        1 & \text{if } \vec{\gamma} = \vec{\nu} \\
        0 & \text{otherwise}
    \end{cases} \label{eq:N0matrixelts} \\
    N_{1,\vec{\nu}\vec{\gamma}}^{(t)} &= 
    \begin{cases}
        1 & \text{if } \gamma_{i_t} = \nu_{i_t} = I \text{ and } \vec{\gamma} = \vec{\nu}\\
        1-\frac{qr}{q-1} & \text{if } \gamma_{i_t} = S \text{ and } \nu_{i_t} = S \text{ and } \vec{\gamma} = \vec{\nu}\\
        \frac{qr}{q-1} & \text{if } \gamma_{i_t} = S \text{ and } \nu_{i_t} = I \text{ and } \gamma_a = \nu_a \forall a \neq i_t\\
        0 & \text{otherwise}
    \end{cases} \label{eq:N1matrixelts} \\
    N_{2,\vec{\nu}\vec{\gamma}}^{(t)} &= 
    \begin{cases}
        1 & \text{if } \gamma_{i_t} = \nu_{i_t} = I \text{ and } \vec{\gamma} = \vec{\nu}\\
        u & \text{if } \gamma_{i_t} = S \text{ and } \nu_{i_t} = S \text{ and } \vec{\gamma} = \vec{\nu}\\
        1-u & \text{if } \gamma_{i_t} = S \text{ and } \nu_{i_t} = I \text{ and } \gamma_a = \nu_a\; \forall a \neq i_t\\
        0 & \text{otherwise}\,,
    \end{cases}\label{eq:N2matrixelts}
\end{align}
and $N_w^{\prime\,(t)}$ are given by the same equations, with $j_t$ replacing $i_t$. 

\subsection{Mapping noisy circuits to stochastic processes}

Define
\begin{align}
    \mathcal{U}_0^{(t)} &= \mathcal{U}^{(t)} \otimes \mathcal{U}^{(t)} \\
    \mathcal{U}_1^{(t)} &= \widetilde{\mathcal{U}}^{(t)} \otimes \mathcal{U}^{(t)}\\
    \mathcal{U}_2^{(t)} &= \widetilde{\mathcal{U}}^{(t)} \otimes \widetilde{\mathcal{U}}^{(t)}\,,
\end{align}
where $\mathcal{U}^{(t)}$ and $\widetilde{\mathcal{U}}^{(t)}$ are given in Eqs.~\eqref{eq:U^t} and \eqref{eq:tildeU^t}. Then we may write, for $w \in \{0,1,2\}$
\begin{align}
    Z_w = q^{2n} \EV_U \left[\tr\left[\ketbra{0^n}^{\otimes 2}\;\mathcal{U}_w^{(n+s)} \circ \cdots \circ \mathcal{U}_w^{(-n+1)}\left(\ketbra{0^n}^{\otimes 2}\right) \right]\right]\,.
\end{align}
Since each $U^{(t)}$ is chosen independently, we are free to perform the expectation value individually over each $\mathcal{U}_w^{(t)}$ channel. The noiseless channel $\mathcal{U}^{(t)}_0 = \mathcal{U}^{(t)\otimes 2}$ averages to $M^{(t)}$, where $M^{(t)}$ is given in Eq.~\eqref{eq:M(t)}.
The action of the noise may also be averaged, since, as discussed in \autorefapp{app:effectofnoise}, we may pull out a single-qudit Haar random gate to act after each noise location. Thus, the noiseless single qudit gates at the end of the circuit may be dropped as they are being absorbed into the noise.
Let
\begin{equation}
    M_w^{(t)} = N_w^{\prime\,(t)} \circ N_w^{(t)} \circ M^{(t)}
\end{equation} 
so that 
\begin{align}
    Z_w = q^{2n} \tr\left[\ketbra{0^n}^{\otimes 2}\;M_w^{(s)} \circ \cdots \circ M_w^{(-n+1)}\left(\ketbra{0^n}^{\otimes 2}\right) \right]\,.
\end{align}
Following the noiseless analysis of \autorefapp{app:noiselessanalysisframework}, we may now write $Z_w$ as a product of matrices
\begin{align}\label{eq:Zwmatrixproduct}
    Z_w &= \sum_{\gamma \in \{I,S\}^{n \times (3s+1)}}\frac{q^{n-|\vec{\gamma}^{(0)}|}}{(q+1)^n} \left(\prod_{t=1}^{s}  N^{\prime\,(t)}_{w,\vec{\gamma}^{(t)}\vec{\gamma}^{(t-1/3)}}
    N^{(t)}_{w,\vec{\gamma}^{(t-1/3)}\vec{\gamma}^{(t-2/3)}}
    M^{(t)}_{\vec{\gamma}^{(t-2/3)}\vec{\gamma}^{(t-1)}}\right) q^{|\vec{\gamma}^{(s)}|}
\end{align}
generalizing Eq.~\eqref{eq:Z0asmatrixproduct}. In the notation of \autoref{sec:overviewnoisywalk}, for $w=1$ this can be expressed as $Z_1=\mathbb{E}_{rq/(q+1)}[q^{|\vec{\gamma}^{(s)}|}]$ where the expectation is over the stochastic process that generates a trajectory with $3s+1$ configurations (at time values $t=0,1/3,2/3,1,\ldots,s$). For $w=2$, it reads $Z_2=\mathbb{E}_{1-u}[q^{|\vec{\gamma}^{(s)}|}]$.

The expressions for $Z_w$ as weighted sums over trajectories can alternatively be interpreted as partition functions of an Ising-like stat mech model where each $\gamma_a^{(t)}$ is an Ising variable $\{+1,-1\}$. There are interactions between adjacent Ising variables whenever a gate or noise location acts between them; the associated interaction strengths can be calculated from the matrix elements listed above. 

\subsection{Bra-ket notation for the stochastic process}
We now write the above insights in a notation that offers slightly more flexibility, which we will utilize in our proofs. The reader need only read this section to verify the proofs that appear later.  Consider a $2^n$-dimensional vector space, where orthonormal basis states are labeled by configurations $\ket{\vec{\nu}}$ for each $\vec{\nu} \in \{I,S\}^n$. 
Define the vectors
\begin{align}
    \ket{\mathbf{1}} &= \sum_{\vec{\nu} \in \{I,S\}^n} \ket{\vec{\nu}} \\
    \ket{\mathbf{q}} &= \sum_{\vec{\nu} \in \{I,S\}^n}  q^{|\vec{\nu}|}\ket{\vec{\nu}}\\
    \ket{\Lambda} &= \frac{1}{(q+1)^n}\sum_{\vec{\nu} \in \{I,S\}^n}  q^{n-|\vec{\nu}|}\ket{\vec{\nu}}\,.
\end{align}
Then we may define $2^n \times 2^n$ transition matrices $P^{(t)}$, which enact the $t$th step of the noiseless stochastic process, as well as matrices $Q_{\sigma}^{(t)}$ and $Q'^{(t)}_{\sigma}$ which enact the $S \rightarrow I$ transition with probability $\sigma$ on qudits $i_t$ and $j_t$, respectively. Explicitly we let
\begin{align}
    P^{(t)} &= \mathcal{I}_{[n]\setminus \{i_t,j_t\}} \otimes P_{\{i_t,j_t\}} \\
    Q_\sigma^{(t)} &= \mathcal{I}_{[n]\setminus \{i_t\}} \otimes \Big(\ketbra{I} + (1-\sigma) \ketbra{S} + \sigma \ketAbraB{I}{S}\Big)_{\{i_t\}} \\
    Q'^{(t)}_{\sigma}  &= \mathcal{I}_{[n]\setminus \{j_t\}} \otimes \Big(\ketbra{I} + (1-\sigma) \ketbra{S} + \sigma \ketAbraB{I}{S}\Big)_{\{j_t\}}\,,
\end{align}
where the subscripts on the right-hand side denote which bits are acted upon by which operators, and
\begin{align}
      D ={} &\ketbra{II} + \ketbra{SS} \label{eq:Dmatrix}\\
      T ={} &\frac{q^2}{q^2+1} \ketAbraB{II}{IS} +\frac{q^2}{q^2+1} \ketAbraB{II}{SI} + \frac{1}{q^2+1} \ketAbraB{SS}{SI}+\frac{1}{q^2+1} \ketAbraB{SS}{IS}\label{eq:Tmatrix}\\
     P ={} &D+ T   \label{eq:Pmatrix}\,.
\end{align}
Note that $P$ is a stochastic $4 \times 4$ matrix. 
Then, define
\begin{equation}\label{eq:Zsigma}
    \mathcal{Z}_\sigma = \bra{\mathbf{q}}\left(\prod_{t=1}^s Q'^{(t)}_{\sigma}Q^{(t)}_{\sigma} P^{(t)}\right)\ket{\Lambda}\,,
\end{equation}
If the circuit diagram is generated randomly, as is the case for the complete-graph architecture, then $\mathcal{Z}_\sigma$ is defined instead as the mean of the above expression over choice of circuit diagram.  For the specific case of the complete-graph architecture (where the pair of qudits acted upon by each gate is chosen independently from all other gates), the average of $\mathcal{Z}_\sigma$ over different circuit diagrams can be accomplished by averaging the matrix $Q'^{(t)}_{\sigma}Q^{(t)}_{\sigma} P^{(t)}$ over all choices of $\{i_t,j_t\}$. This is the convention we follow when analyzing the complete-graph architecture.  

The $\ket{\Lambda}$ in the equation above represents the distribution over the initial configuration $\vec{\gamma}^{(0)}$, and the $\bra{\mathbf{q}}$ represents the weighting given to the final configuration $\vec{\gamma}^{(s)}$. Thus, the equation for $Z_w$ in Eq.~\eqref{eq:Zwmatrixproduct} implies that
\begin{align}
    Z_0 &= \mathcal{Z}_0 \\
    Z_1 &= \mathcal{Z}_{r q/(q-1)} \\
    Z_2 &= \mathcal{Z}_{1-u}\,.
\end{align}

\section{Detailed proofs}\label{app:rigorousproofs}

The statements of our main theorems in the appendix are slightly more general than in the main text: we consider a general class of architectures that are both ``layered'' and ``regularly connected,'' which we define below. The theorem statements are in terms of the anti-concentration size $s_{AC}$ of the architecture, which is defined \cite{dalzell2020anticoncentration} to be the minimum circuit size $s$ such that $Z_0 \leq 4q^n/(q^n+1)$. The 1D architecture and complete-graph architecture are the only architectures known to have $s_{AC} = \Theta(n\log(n))$, so for clarity, we previously restricted our statements to those architectures. 

First, in \autorefapp{app:defsandmainlemmas}, we present definitions and our main lemmas, which are themselves dependent on more minor lemmas. Then, in \autorefapp{app:pfsmaintheorems}, we prove a slightly generalized version of our theorems from the main text, based on the main lemmas. Afterward, in \autorefapp{app:machinery}, we develop some more machinery and state the minor lemmas, deferring their proofs to \autorefapp{app:deferredproofslemmas}.

\subsection{Definitions and main lemmas}\label{app:defsandmainlemmas}
Our proofs apply to architectures that are layered and $h$-regularly connected for some constant $h=O(1)$. The regularly connected property was defined in Ref.~\cite{dalzell2020anticoncentration}, where it was conjectured to imply anti-concentration after $\Theta(n\log(n))$ gates, and we repeat its definition here.

First, define an architecture as in Ref.~\cite{dalzell2020anticoncentration} to be an efficient (possibly randomized) algorithm that takes as input circuit parameters $(n,s)$ and outputs a length-$s$ sequence of size-2 subsets $ (A^{(1)},\ldots,A^{(s)})$, where $A^{(t)} \subset [n]$ and $|A^{(t)}| = 2$ for each $t$. The subsets $A^{(t)}$ correspond to the pair of qudits acted upon by a gate at time step $t$. 
\begin{definition}[Regularly connected \cite{dalzell2020anticoncentration}]\label{def:regularlyconnected}
    We say an random quantum circuit architecture is $h$-\textit{regularly connected} if for any $n$, any $t$, any subsequence $A=(A^{(1)},\ldots,A^{(t)})$ and any proper subset $R \subset [n]$ of qudit indices, there is at least a $1/2$ probability that, conditioned on the first $t$ gates in the gate sequence being $A$, there exists some index $t'$ for which $t< t' \leq t+hn$, $A^{(t')} \cap R \neq \emptyset$, and $A^{(t')} \not\subset R$.
\end{definition}
If $h=O(1)$, we often simply call the architecture regularly connected, without specifying $h$. This property is a precise way of saying that the circuit does not break into multiple distinct parts that rarely interact with each other (a feature that would prevent scrambling): for any bipartition, there is usually a gate that couples one qubit from each half at least once every $O(n)$ time steps. Nearly all natural architectures are regularly connected (a notable exception being the hypercube architecture \cite{dalzell2020anticoncentration}). 

Next, we define layered, which simply means that the gates can always be neatly arranged into layers of $n/2$ non-overlapping gates. 
\begin{definition}\label{def:layered}
    An architecture is layered if any sequence of gates $(A^{(1)},\ldots,A^{(s)})$ it generates with non-zero probability has the property that for any integer $d\geq 0$, and any pair of gates in the same ``layer''
    \begin{equation}
        t_1,t_2 \in \{dn/2+1, dn/2+2,\ldots, (d+1)n/2\} 
    \end{equation}
    with $t_1 \neq t_2$, we have $A^{(t_1)} \cap A^{(t_2)} = \emptyset$. Thus, all $n$ qudits are acted upon by exactly one gate out of every $n/2$ gates. 
\end{definition}
For layered architectures we can speak clearly about the depth $d = 2s/n$. The anti-concentration depth is then defined as $d_{AC} = 2s_{AC}/n$.  We will generally require $s$ be a multiple of $n/2$ so that there are an integer number of layers.  Regular lattice architectures in $D$ spatial dimensions are typically layered, although adhering strictly to the definition would require applying periodic boundary conditions. We do not expect this condition is actually necessary for our results, but it is analytically convenient. The only place we need it is in \autoref{lem:boundSdestinedmassLAYERS}. 

Our theorems are corollaries of the following lemmas. Recall the definition of $\mathcal{Z}_\sigma$ from Eq.~\eqref{eq:Zsigma}. Note that in these proofs, all constants are dependent on $q$ as well as $h$ (the regularly connected parameter), but independent of $n$ and the noise parameters. 

%
%
\begin{lemma}\label{lem:mainlemmaZbound}
If the random quantum circuit architecture is $h$-regularly connected and layered with anti-concentration depth $d_{AC}$, then there exist constants $c_0$, $c_1$, $c_2$, $c_3$, $c_4$, $c_5$, and $n_0'$ that depend on $h$ and $q$ but not on $n$ or $\sigma$, such that as long as $\sigma \leq c_5/n$ and $n \geq n_0'$, for any value of the circuit depth $d$,
\begin{align}
     \frac{q^n-1}{q^n+1}\left(1-f_\sigma\right)^{d}\leq{} &\mathcal{Z}_\sigma-1
    \leq  \frac{q^n-1}{q^n+1}\left(1-f_\sigma \right)^{d} e^{K_\sigma} \,,
\end{align}
where
\begin{align}
    f_\sigma &=  \frac{1-(1-\sigma(1-q^{-2}))^n}{1-q^{-2n}} \label{eq:fsigma}\\
    K_\sigma &= c_0 n d \sigma^2 + c_1 n \sigma d_{AC}+c_2e^{-c_3(d-d_{AC})+2\sigma d n} + c_4 n \sigma \log(1/(n\sigma))\,. \label{eq:Ksigma}
\end{align}
\end{lemma}
\begin{proof}
The lower bound is an immediate consequence of two lemmas that appear later, \autoref{lem:Zsigmaupperlowerbound} and \autoref{lem:boundSdestinedmassLAYERS}. The upper bound is also an immediate consequence, with the constant $c_1$ absorbing an $O(n\sigma)$ term since $d_{AC} = 2s_{AC}/n \geq \Omega(\log(n))$ by the results of Ref.~\cite{dalzell2020anticoncentration}.
\end{proof}

We show the analogous statement for the complete-graph architecture. 

\begin{lemma}\label{lem:mainlemmaZboundCG}
If the random quantum circuit architecture is the complete-graph architecture, then there exist constants $c'_0$, $c'_1$, $c'_2$, $c'_3$, $c'_4$, $c'_5$, and $n_0$ that depend on $q$ but not on $n$ or $\sigma$, such that as long as $\sigma \leq c'_5/n$ and $n \geq n_0$, for any value of the circuit size $s$,
\begin{align}
     \frac{q^n-1}{q^n+1}\left(1-f'_\sigma\right)^{s}\leq{} &\mathcal{Z}_\sigma-1
    \leq  \frac{q^n-1}{q^n+1}\left(1-f'_\sigma \right)^{s} e^{K'_\sigma} \,,
\end{align}
where
\begin{align}
    f'_\sigma &=  \frac{1-(1-\sigma(1-q^{-2}))^2}{1-q^{-2n}} \label{eq:fprimesigma}\\
    K'_\sigma &= c'_0 s \sigma^2 + c'_1 \sigma s_{AC}+c'_2e^{-c'_3(s-s_{AC})/n + 4 \sigma s} + c'_4 n \sigma \log(1/(n\sigma)) \label{eq:Kprimesigma}\,,
\end{align}
and $s_{AC} = \Theta(n\log(n))$ is the anti-concentration size for the complete-graph architecture. 
\end{lemma}
\begin{proof}
    The proof is the same as \autoref{lem:mainlemmaZbound} except using \autoref{lem:boundSdestinedmassCG} in place of \autoref{lem:boundSdestinedmassLAYERS}.
\end{proof}
Note that in the regime $\sigma \leq O(1/n)$, we can bound $1-\sigma(1-q^{-2}) \geq e^{-\sigma(1-q^{-2})}e^{-O(\sigma^2)}$ and the following holds
\begin{align}\label{eq:fsigmaapprox}
    e^{-n\sigma(1-q^{-2})} e^{-O(n\sigma^2)-O(q^{-2n})} &\leq 1-f_{\sigma} \leq e^{-n\sigma(1-q^{-2})} \\
    e^{-2\sigma(1-q^{-2})} e^{-O(\sigma^2)-O(q^{-2n})} &\leq 1-f'_{\sigma} \leq e^{-2\sigma(1-q^{-2})} \label{eq:fprimesigmaapprox}\,.
\end{align}
The upper bound in Eqs.~\eqref{eq:fsigmaapprox} and \eqref{eq:fprimesigmaapprox} actually holds generally for all $\sigma$.

\subsection{Proofs of main theorems from main lemmas}\label{app:pfsmaintheorems}

\subsubsection{Proof of \autoref{thm:fidelitydecay}: fidelity decay}

\begin{customthm}{\autoref{thm:fidelitydecay}}[generalized and restated]

Consider either the complete-graph architecture or a regularly connected, layered random quantum circuit architecture with $n$ qudits of local Hilbert space dimension $q$ and $s$ gates, where the anti-concentration size is given by $s_{AC}$. Let $r$ be the average infidelity of the local noise channels. Then there exists constants $c$ and $n_0$ such that whenever $r \leq c/n$ and $n \geq n_0$, the following holds:
\begin{align}
    \bar{F} &\geq \exp\left(-2sr(1+q^{-1})\right)e^{-O(sr^2)-O(sq^{-2n})} \label{eq:FXEBlowerboundAPP}\\
    \bar{F} &\leq\exp\left(-2sr(1+q^{-1})\right) Q_1 \,, \label{eq:FXEBupperboundAPP}
\end{align}
where $\bar{F}$ is given in Eq.~\eqref{eq:defAveOutFidelity}, and
\begin{equation}\label{eq:Q1app}
    Q_1 = \exp\left(O(sr^2) + O(s_{AC}r) + e^{O(s_{AC}/n)}e^{-\Omega(s/n)} + O(nr \log(1/(nr)))\right)\,.
\end{equation}
\end{customthm}
\begin{proof}
The quantity $\bar{F}$ is precisely $(Z_1-1)/(Z_0-1) = (\mathcal{Z}_{\sigma}-1)/(\mathcal{Z}_{0}-1)$ with $\sigma=r q/(q-1)$. The statements are then direct consequences of \autoref{lem:mainlemmaZbound} for layered architectures and \autoref{lem:mainlemmaZboundCG} for the complete-graph architecture, combined with the observation in Eqs.~\eqref{eq:fsigmaapprox} and \eqref{eq:fprimesigmaapprox}. Note also that $nd=2s$. 
\end{proof}

\subsubsection{Proof of \autoref{thm:convergenceToUniform}: convergence to the uniform distribution}

\begin{customthm}{\autoref{thm:convergenceToUniform}}[generalized and restated]
Consider either the complete-graph architecture or a regularly connected, layered random quantum circuit architecture with $n$ qudits of local Hilbert space dimension $q$ and $s$ gates, where the anti-concentration size is given by $s_{AC}$. Let $u$ be the unitarity of the local noise channels (and define $v=1-u$). Then there exist constants $c$ and $n_0$ such that as long as $v \leq c/n$ and $n \geq n_0$
    \begin{equation}
        \EV_U \left[\frac{1}{2}\lVert \pnoisy - \punif \rVert_1 \right] \leq \exp(-sv(1-q^{-2})) Q_2 \,,
    \end{equation}
    where $\punif$ is the uniform distribution and 
    \begin{equation}\label{eq:Q2app}
        Q_2 = \exp\left(O(sv^2) + O(s_{AC}v) + e^{O(s_{AC}/n)}e^{-\Omega(s/n)}+O(nv \log(1/(nv))\right)\,.
    \end{equation}
\end{customthm}
\begin{proof}
    We can use the 1-norm to 2-norm inequality in Eq.~\eqref{eq:1normto2norm}, along with Jensen's inequality for the concave $\sqrt{\cdot}$ function to say
    \begin{align}
        \EV_U \left[\frac{1}{2}\lVert \pnoisy - \punif \rVert_1 \right] &\leq \frac{1}{2}\sqrt{q^n\EV_U\left[\sum_x\left(\pnoisy(x) - q^{-n}\right)^2\right]} \\
        &= \frac{1}{2}\sqrt{q^{2n}\EV_U\left[\pnoisy(0^n)^2\right]-1} = \frac{1}{2}\sqrt{Z_2-1}\\
        &= \frac{1}{2}\sqrt{\mathcal{Z}_v-1} 
    \end{align}
    Then, the theorem follows from the upper bound in \autoref{lem:mainlemmaZbound} for layered architectures and \autoref{lem:mainlemmaZboundCG} for the complete-graph architecture, with $\sigma = v$, combined with the observation in Eqs.~\eqref{eq:fsigmaapprox} and \eqref{eq:fprimesigmaapprox}. Note also that $nd=2s$.
\end{proof}

\subsubsection{Proof of \autoref{thm:whitenoisebound}: approximation by white noise}

\begin{customthm}{\autoref{thm:whitenoisebound}}[generalized and restated]
 Consider either the complete-graph architecture or a regularly connected, layered random quantum circuit architecture with $n$ qudits of local Hilbert space dimension $q$ and $s$ gates, where the anti-concentration size is given by $s_{AC}$.  Let $r$ be the average infidelity and $u$ the unitarity of the local noise channels (and define $v=1-u$). Let
    \begin{equation}\label{eq:deltaAPP}
        \delta=2r(1+q^{-1})-(1-u)(1-q^{-2}) \,.
    \end{equation}
    Then, when we choose $F = \bar{F}$ as in Eq.~\eqref{eq:defAveOutFidelity}, there exist constants $c_1$, $c_2$, and $n_0$ such that as long as $v \leq c_1/n$, $r \leq c_2/n$, and $n \geq n_0$,
    \begin{align}
        \EV_U \left[\frac{1}{2}\lVert \pnoisy - \pwhitenoise \rVert_1 \right]
        \leq{}& \bar{F}\sqrt{s}\left(\sqrt{\delta}+O(v)+O(r)\right) + O(\bar{F}\sqrt{s_{AC}v}) \nonumber\\
        &\qquad +O(\bar{F}\sqrt{nv\log(1/nv)})+ \bar{F}e^{O(s_{AC}/n)-\Omega(s/n)}\,, \label{eq:WNdistanceAPP}
    \end{align}
    whenever the right-hand side of Eq.~\eqref{eq:WNdistanceAPP} is less than $\bar{F}$.  
\end{customthm}
\begin{proof}
    Following \autoref{sec:overviewnoisywalk}, we first use the 1-norm to 2-norm bound and Jensen's inequality, and then we optimize the value of $F$. The bound on the distance between $\pwhitenoise$ and $\pnoisy$ is minimized when we choose $F = \bar{F} = (Z_1-1)/(Z_0-1)$.
    When this value is chosen, the bound can be expressed as
    \begin{equation}\label{eq:nwndistanceIntermediate}
        \EV_U \left[\frac{1}{2}\lVert \pnoisy - \pwhitenoise \rVert_1 \right] \leq \frac{1}{2}\bar{F}\sqrt{\frac{(Z_2-1)(Z_0-1)^2}{(Z_1-1)^2}-(Z_0-1)}
    \end{equation}
    Note that after the anti-concentration size has been surpassed, the quantity $Z_0-1$ rapidly approaches $\frac{q^n-1}{q^n+1}\approx 1$  from above. To evaluate $Z_0$, $Z_1$ and $Z_2$ we use the correspondence $Z_0 = \mathcal{Z}_0$, $Z_1 = \mathcal{Z}_{r q /(q-1)}$ and $Z_2 = \mathcal{Z}_v$.  The bounds from \autoref{lem:mainlemmaZbound} for layered architectures and \autoref{lem:mainlemmaZboundCG} for the complete-graph architecture then allow us to upper bound $(Z_2-1)(Z_0-1)^2/(Z_1-1)^2$, arriving at
    \begin{align}
        \frac{(Z_2-1)(Z_0-1)^2}{(Z_1-1)^2}
        \leq{}& \frac{q^n-1}{q^n+1}e^{2s\left(2r(1+q^{-1})-v(1-q^{-2})\right)}e^{O(sr^2)+O(sq^{-2n})+e^{O(s_{AC}/n)}e^{-\Omega(s/n)}}Q_2  \\
        =& \frac{q^n-1}{q^n+1} e^{2s\delta}e^{O(sr^2+sq^{-2n}+sv^2+s_{AC}v-nv \log (nv))+ e^{O(s_{AC}/n)}e^{-\Omega(s/n)}} \,,
    \end{align}
    where $Q_2$ is given in Eq.~\eqref{eq:Q2app}, and $\delta$ is given in Eq.~\eqref{eq:deltaAPP}.
    Now, working back from Eq.~\eqref{eq:nwndistanceIntermediate}, and noting that $e^x-1 < 2x$ for all $x \leq 1$, we have
    \begin{align}
        \EV_U \left[\frac{1}{2}\lVert \pnoisy - \pwhitenoise \rVert_1 \right] 
        \leq{}& \frac{\bar{F}}{2}\sqrt{4s\delta+O(sr^2+sq^{-2n}+sv^2+s_{AC}v-nv \log (nv))+ e^{O(s_{AC}/n)}e^{-\Omega(s/n)}}\\
        ={}&\bar{F}\sqrt{s}\left(\sqrt{\delta}+O(v)+O(r)\right) + O(\bar{F}\sqrt{s_{AC}v}) \nonumber\\
        &\qquad +O(\bar{F}\sqrt{nv\log(1/nv)})+ \bar{F}e^{O(s_{AC}/n)-\Omega(s/n)}
    \end{align}
    when the quantity under the square root is less than 1 (and using $\sqrt{A+B} \leq \sqrt{A}+\sqrt{B}$).
\end{proof}

\subsection{Machinery for proof}\label{app:machinery}

We now develop some more notation, and we precisely state some of our lemmas. We defer the proofs of these lemmas to \autorefapp{app:deferredproofslemmas}. As we state them, we attempt to give some commentary about the meaning and purpose of the different objects that we define and the related lemmas. 

\begin{figure}[h]
\definecolor{bdr}{HTML}{AF2D0A}
\definecolor{bdb}{HTML}{142194}
\centering
\scalebox{0.9}{
\begin{tikzpicture}[scale=1.0,thick]
\tikzmath{\xsep = 5; \ysep = 2.5; \boxwidth = 1; \boxheight = 1;}

\foreach \i in {0,1,2}
{
%
\draw (\i*\xsep,0) -- (\i*\xsep+2*\boxwidth,0);
\draw (\i*\xsep,\boxheight) -- (\i*\xsep+2*\boxwidth,\boxheight);
\draw (\i*\xsep,2*\boxheight) -- (\i*\xsep+2*\boxwidth,2*\boxheight);
%
\draw (\i*\xsep,0) -- (\i*\xsep,2*\boxheight);
\draw (\i*\xsep+\boxwidth,0) -- (\i*\xsep+\boxwidth,2*\boxheight);
\draw (\i*\xsep+2*\boxwidth,0) -- (\i*\xsep+2*\boxwidth,2*\boxheight);
%
\draw (\i*\xsep,-\ysep)--(\i*\xsep,-\ysep+\boxheight)--(\i*\xsep+2*\boxwidth,-\ysep+\boxheight)--(\i*\xsep+2*\boxwidth,-\ysep)--(\i*\xsep,-\ysep);
\draw (\i*\xsep+\boxwidth,-\ysep)--(\i*\xsep+\boxwidth,-\ysep+\boxheight);
%
\draw[->] (\i*\xsep+0.6*\boxwidth,-0.25*\boxheight)--(\i*\xsep+0.6*\boxwidth,-\ysep+1.25*\boxheight);
\draw[->] (\i*\xsep+1.6*\boxwidth,-0.25*\boxheight)--(\i*\xsep+1.6*\boxwidth,-\ysep+1.25*\boxheight);
\node at (\i*\xsep+0.3*\boxwidth,-0.5*\ysep+0.5\boxheight) {$\Delta_{i_t}$};
\node at (\i*\xsep+1.3*\boxwidth,-0.5*\ysep+0.5\boxheight) {$\Delta_{j_t}$};
%
\node at (\i*\xsep+0.5*\boxwidth,2.35*\boxheight) {\large{$i_t$}};
\node at (\i*\xsep+1.5*\boxwidth,2.35*\boxheight) {\large{$j_t$}};
}
%
\draw [->] (2.5*\boxwidth,\boxheight)--(\xsep-0.5*\boxwidth,\boxheight);
\node at (\boxwidth+0.5*\xsep,1.4*\boxheight) {\Large{$R_0^{(t)}$}};
\draw [->] (\xsep+2.5*\boxwidth,0.4*\boxheight)--(2*\xsep-0.5*\boxwidth,0.4*\boxheight);
\node at (\boxwidth+1.5*\xsep,0.75*\boxheight) {$Q'^{(t)}_\sigma Q^{(t)}_\sigma$};
\draw [->] (\xsep+2.5*\boxwidth,1.4*\boxheight)--(2*\xsep-0.5*\boxwidth,1.4*\boxheight);
\node at (\boxwidth+1.5*\xsep,1.75*\boxheight) {$\mathcal{I}$};
%
\node at (-0.5*\boxwidth,0.5*\boxheight) {\Large{$Y$}};
\node at (-0.5*\boxwidth,1.5*\boxheight) {\Large{$X$}};
\node at (-0.5*\boxwidth,-\ysep+0.5*\boxwidth) {\Large{$W$}};
\node at (-2.0*\boxwidth,0.5*\boxheight) {\large{noisy}};
\node at (-2.0*\boxwidth,1.5*\boxheight) {\large{noiseless}};
\node at (-2.0*\boxwidth,-\ysep+0.5*\boxheight) {\large{difference}};
%
\node at (0.5*\boxwidth,0.5*\boxheight) {\Large{\textcolor{bdb}{$I$}}};
\node at (0.5*\boxwidth,1.5*\boxheight) {\Large{\textcolor{bdb}{$I$}}};
\node at (0.5*\boxwidth,-\ysep+0.5*\boxheight) {\Large{\textcolor{bdr}{$S$}}};
\node at (1.5*\boxwidth,0.5*\boxheight) {\Large{\textcolor{bdr}{$S$}}};
\node at (1.5*\boxwidth,1.5*\boxheight) {\Large{\textcolor{bdr}{$S$}}};
\node at (1.5*\boxwidth,-\ysep+0.5*\boxheight) {\Large{\textcolor{bdr}{$S$}}};
\node at (\xsep+0.5*\boxwidth,0.5*\boxheight) {\Large{\textcolor{bdr}{$S$}}};
\node at (\xsep+0.5*\boxwidth,1.5*\boxheight) {\Large{\textcolor{bdr}{$S$}}};
\node at (\xsep+0.5*\boxwidth,-\ysep+0.5*\boxheight) {\Large{\textcolor{bdr}{$S$}}};
\node at (\xsep+1.5*\boxwidth,0.5*\boxheight) {\Large{\textcolor{bdr}{$S$}}};
\node at (\xsep+1.5*\boxwidth,1.5*\boxheight) {\Large{\textcolor{bdr}{$S$}}};
\node at (\xsep+1.5*\boxwidth,-\ysep+0.5*\boxheight) {\Large{\textcolor{bdr}{$S$}}};
\node at (2*\xsep+0.5*\boxwidth,0.5*\boxheight) {\Large{\textcolor{bdb}{$I$}}};
\node at (2*\xsep+0.5*\boxwidth,1.5*\boxheight) {\Large{\textcolor{bdr}{$S$}}};
\node at (2*\xsep+0.5*\boxwidth,-\ysep+0.5*\boxheight) {\Large{\textcolor{bdb}{$I$}}};
\node at (2*\xsep+1.5*\boxwidth,0.5*\boxheight) {\Large{\textcolor{bdr}{$S$}}};
\node at (2*\xsep+1.5*\boxwidth,1.5*\boxheight) {\Large{\textcolor{bdr}{$S$}}};
\node at (2*\xsep+1.5*\boxwidth,-\ysep+0.5*\boxheight) {\Large{\textcolor{bdr}{$S$}}};
\end{tikzpicture}
}
\caption{Illustration of dynamics of coupled noiseless and noisy stochastic process. The gate at time step $t$ acts on sites $\{i_t,j_t\}$; the transition from time step $t-1$ to time step $t$ can modify the assignment only at these locations. In the example above, at time step $t-1$ (left), both the $X$ and $Y$ systems are assigned $I$ at position $i_t$ and $S$ at position $j_t$. Since $i_t$ and $j_t$ are assigned different values, the transformation $R_0^{(t)}$ forces a bit flip at one of the positions, but the same bit is flipped for the $X$ and $Y$ systems. In this example, the $I$ is flipped to $S$. Then the configuration at time step $t$ (right) is formed by applying noise operators $Q'^{(t)}_\sigma Q^{(t)}_\sigma$ only to the $Y$ copy, which results in a bit flip from $S$ to $I$ independently on each location with probability $\sigma$. In the example above, only the $i_t$ assignment is flipped. The system $W$ captures the difference between the $X$ and $Y$ copies; it is assigned $S$ wherever they agree and $I$ wherever they disagree. This formalism allows us to isolate the impact of the noise on a trajectory of the stochastic process compared to what ``would have'' happened had there been no noise.}
\label{fig:XYWfigure}
\end{figure}
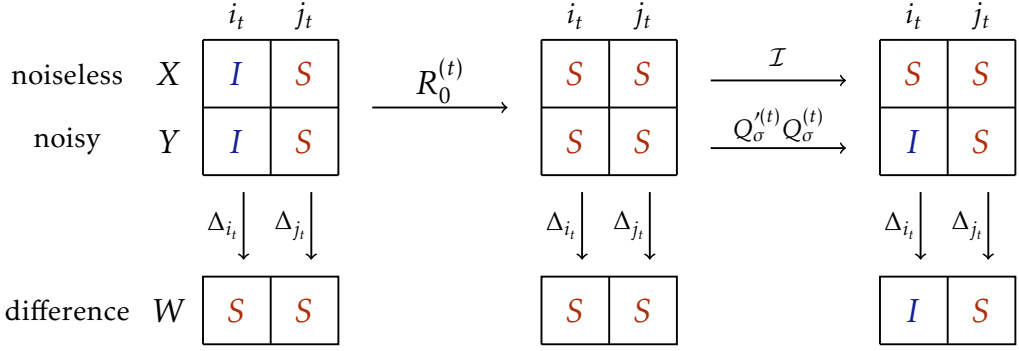

\subsubsection{Coupling a noiseless and noisy copy of the dynamics}

We have a fairly good understanding of the noiseless stochastic process from Ref.~\cite{dalzell2020anticoncentration}. Our strategy here is to examine how introducing noise perturbs that process. To that end, we consider \textit{two} copies of the random walk, where one is noiseless and one is noisy, but where they are correlated so that we can isolate the impact of the noise.

Recall that we have reduced the calculation of $\mathcal{Z}_\sigma$ to the expectation value of a random variable (the configuration) that evolves according to the stochastic transition matrix $P^{(t)}$ (representing the noiseless gate) followed by transition matrices $Q^{(t)}_\sigma$ and $Q'^{(t)}_\sigma$, which represent the impact of noise.

Let $X$ denote the $2^n$-dimensional vector space for the first ``noiseless'' copy and $Y$ for the second ``noisy'' copy.  To define the dynamics formally, recall the definition of $D$ and $T$ from Eqs.~\eqref{eq:Dmatrix} and \eqref{eq:Tmatrix}, and define the following matrix that acts on four bits. 
\begin{align}\label{eq:Rmatrix}
\begin{split}
   R= &D \otimes D + D \otimes T + T \otimes D  + T \otimes T \left(\ketbra{IS,SI} + \ketbra{SI,IS}\right)\\
        &+ \frac{q^2}{q^2+1} \ketAbraB{II,II}{IS,IS}+\frac{1}{q^2+1} \ketAbraB{SS,SS}{IS,IS} \\
        &+\frac{q^2}{q^2+1}\ketAbraB{II,II}{SI,SI}  +\frac{1}{q^2+1} \ketAbraB{SS,SS}{SI,SI}
    \end{split}\,.
\end{align}
The matrix $R$ is stochastic. It should be understood as a correlated bit flip where, if the first and third bits are equal and the second and fourth bits are equal, they are sent to a state where that is still true. However, its marginal on either the first two bits or the last two bits is precisely $P$ from Eq.~\eqref{eq:Pmatrix}. Refer to the $i$th bit of the first random variable as $X_i$ and the $i$th bit of the second random variable as $Y_i$.
Then define
\begin{equation}
    R_\sigma^{(t)} = \left(\mathcal{I}_X \otimes (Q'^{(t)}_{\sigma} Q_\sigma^{(t)})_Y\right)\left(\mathcal{I}_{XY\setminus\{X_{i_t}X_{j_t},Y_{i_t}Y_{j_t}\}} \otimes R_{\{X_{i_t}X_{j_t},Y_{i_t}Y_{j_t}\}}\right)\,.
\end{equation}
In words, what $R_\sigma^{(t)}$ does is first generate a correlated noiseless transition among the bits involved in the gate $\{X_{i_t}X_{j_t},Y_{i_t}Y_{j_t}\}$ for both the first ``noiseless'' $X$ copy and the second ``noisy'' $Y$ copy, and then apply the noise transitions only to the $Y$ copy. Since the marginal dynamics of the matrix $R$ restricted either to the first two bits or to the last two bits is the matrix $P$, the marginal dynamics of $R_\sigma^{(t)}$ are $P^{(t)}$ on the $X$ copy and $Q'^{(t)}_{\sigma} Q_\sigma^{(t)}P^{(t)}$ on the $Y$ copy. The action of $R^{(t)}_\sigma$ on an example configuration is illustrated in \autoref{fig:XYWfigure}.

An additional property of $R_\sigma^{(t)}$ is that it preserves a certain subspace of the $2^n \times 2^n$ Hilbert space. If we define the projector $    \pi_i = \left(\ketbra{II} + \ketbra{SS} + \ketbra{SI}\right)_{\{X_iY_i\}}$, then the support of $\bigotimes_{i=0}^{n-1} \pi_i$ is not coupled with its orthogonal complement by the matrix $R_\sigma^{(t)}$. Let us refer to this subspace as the \textit{accessible subspace}. This corresponds to the fact that the noise can send $S \rightarrow I$ but not vice versa. 

We define the initial state to be the correlated version of $\ket{\Lambda}$
\begin{equation}
    \ket{\Lambda\Lambda} = \frac{1}{(q+1)^n}\sum_{\vec{\nu}}  q^{n-|\vec{\nu}|}\ket{\vec{\nu}}_X \otimes \ket{\vec{\nu}}_Y\,,
\end{equation}
which lies in the accessible subspace, so evolution by $R_\sigma^{(t)}$ is guaranteed to remain within the accessible subspace for the entire evolution.

In terms of $R_\sigma^{(t)}$ we can rewrite Eq.~\eqref{eq:Zsigma} as
\begin{equation}
    \mathcal{Z}_\sigma = \bra{\mathbf{1},\mathbf{q}}\prod_{t=1}^s R_\sigma^{(t)}\ket{\Lambda\Lambda}\,,
\end{equation}
where $\ket{a,b}$ is shorthand for $\ket{a}_X \otimes \ket{b}_Y$. Inner product with $\bra{\mathbf{1}}$ in the equation above simply marginalizes over the noiseless $X$ copy (since the vector is normalized in the $1$-norm), and in our proofs, we will use this notation often. 

Note also that since the marginal dynamics of the $X$ copy is the noiseless dynamics, we can marginalize over the $Y$ copy and conclude that
\begin{equation}
    \mathcal{Z}_0 = \bra{\mathbf{q},\mathbf{1}}\prod_{t=1}^s R_\sigma^{(t)}\ket{\Lambda\Lambda}
\end{equation}
for any $\sigma$. 

In our proof, we find it convenient to define
\begin{equation}
    \ket{v^{(t)}} = \prod_{t'=1}^t R_\sigma^{(t')}\ket{\Lambda\Lambda}\,,
\end{equation}
which represents the joint probability distribution over the $2^n$ configurations after $t$ gates (and their associated noise channels) have been applied. 
Note that for circuit architectures where the circuit diagram is chosen randomly, such as the complete-graph architecture, $\ket{v^{(t)}}$ is defined as the above expression averaged over all circuit diagrams. 

Finally, let $W$ refer to a third copy of the $2^n$-dimensional Hilbert space and define a mapping from the $i$th bits of $X$ and $Y$ to the $i$th bit of $W$, as follows:
\begin{equation}
    \Delta_i = \ket{S}_{W_i}\bra{SS}_{X_iY_i} + \ket{S}_{W_i}\bra{II}_{X_iY_i}  + \ket{I}_{W_i}\bra{IS}_{X_iY_i}  + \ket{I}_{W_i}\bra{SI}_{X_iY_i} \,.
\end{equation}
It maps a bit pair to $\ket{S}$ if they agree and $\ket{I}$ if they disagree. Let
\begin{equation}
    \Delta = \bigotimes_{i=0}^{n-1} \Delta_i
\end{equation}
be the map from $X \otimes Y$ to $W$. Note that $\Delta \ket{\Lambda\Lambda} = \ket{S^n}$. 

\subsubsection{$I$-destined and $S$-destined probability mass}

\begin{figure}[h]
\definecolor{bdr}{HTML}{AF2D0A}
\definecolor{bdb}{HTML}{142194}
\definecolor{bdp}{HTML}{EAAC05}
\centering
\scalebox{0.8}{
\begin{tikzpicture}[scale=1.0,thick]
\tikzmath{\xsep = 2.2; \ysep = 1.5; \radius = 0.5; \n = 6; \factor=1.76;}
\pgfsnakesegmentaspect=0.3
%
\foreach \i in {1,2,3,4,5}
{
\draw[bdr,very thick] (\i*\xsep,0) circle[radius=\radius];
\draw[bdb,very thick] (\i*\xsep,\ysep) circle[radius=\radius];
\draw[dotted] (\i*\xsep-\radius*\factor,-\radius*\factor)--(\i*\xsep+\radius*\factor,-\radius*\factor)--(\i*\xsep+\radius*\factor,\ysep+\radius*\factor)--(\i*\xsep-\radius*\factor,\ysep+\radius*\factor)--(\i*\xsep-\radius*\factor,-\radius*\factor);
\node at (\i*\xsep,\ysep+\radius*\factor*1.75) {\LARGE{$\i$}};
\draw[bdp,<-] (\i*\xsep + \radius*1.2,\radius*0.4)--(\i*\xsep+\xsep - \radius*1.2,\radius*0.4);
\draw[bdp,<-] (\i*\xsep -\xsep + \radius*1.2,\ysep+\radius*0.4)--(\i*\xsep - \radius*1.2,\ysep+\radius*0.4);
}
\foreach \i in {0,1,2,3,4,5}
{
\draw[bdp,<-] (\i*\xsep + \radius,\ysep-\radius*0.75)--(\i*\xsep+\xsep - \radius,\radius*0.75);
}
\foreach \i in {1,2,3,4}
{
\draw[bdr,<->,thick] (\i*\xsep + \radius*1.25,0) -- (\i*\xsep+\xsep - \radius*1.25,0);
\draw[bdb,<->,thick] (\i*\xsep + \radius*1.25,\ysep) -- (\i*\xsep+\xsep - \radius*1.25,\ysep);
}
\draw[bdr,->,thick] (\n*\xsep-\xsep + \radius*1.25,0) -- (\n*\xsep - \radius*1.25,0);
\draw[bdb,<-,thick] (\radius*1.25,\ysep) -- (\xsep - \radius*1.25,\ysep);
\node at (0,\ysep+\radius*\factor*1.75) {\LARGE{$0$}};
\node at (\n*\xsep,\ysep+\radius*\factor*1.75) {\LARGE{$6$}};
\draw[bdb,very thick] (0,\ysep) circle[radius=\radius];
\node at (0,\ysep) {\LARGE{\textcolor{bdb}{$I^n$}}};
\draw[bdr,very thick] (\xsep*\n,0) circle[radius=\radius];
\node at (\xsep*\n,0) {\LARGE{\textcolor{bdr}{$S^n$}}};
\node at (-\xsep*1.2,0) {\Large\textcolor{bdr}{$S$-destined}};
\node at (-\xsep*1.2,\ysep) {\Large\textcolor{bdb}{$I$-destined}};
\node at (-\xsep*1.55,\ysep+\radius*\factor*1.75) {\Large{Hamming weight}};
\end{tikzpicture}
}
\caption{Schematic of the concept of $I$-destined and $S$-destined probability mass in an $n=6$ example. Each of the $2^n$ configurations corresponds to a Hamming weight between $0$ and $n$, that is, the number of $S$ assignments out of $n$. For a given configuration, the mass can be broken into an $I$-destined and an $S$-destined portion corresponding to the fraction that would end at Hamming weight 0 and Hamming weight $n$, respectively, if an infinite number of noiseless gates were applied. In the diagram, this corresponds to a division of the mass into the blue and red circles within each Hamming weight bucket. For each $x$, the ratio of $I$-destined to $S$-destined mass at Hamming weight $x$ is always precisely $(1-q^{-2n+2x})/(q^{-2n+2x}-q^{-2n})$. A portion of probability mass that is conditioned on being $I$-destined or $S$-destined obeys effective transition dynamics (given by transition matrices $P_I^{(t)}$ and $P_S^{(t)}$, respectively) that preserve which fixed-point the portion of probability mass is destined for. The allowed transitions of these conditional noiseless dynamics are given by blue and red arrows in the diagram. The allowed transitions associated with action of a noise location are given by yellow lines: a portion of $S$-destined mass that experiences a $S \rightarrow I$ flip due to noise can remain $S$-destined, or it can become $I$-destined, but $I$-destined mass can never become $S$-destined. The proof decomposes the $I$-destined mass according to which time step it first became $I$-destined. }
\label{fig:destinedmass}
\end{figure}
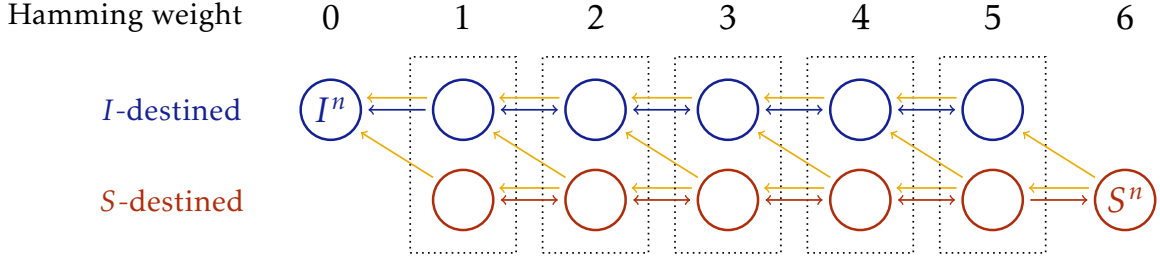

We view $\ket{v^{(t)}}$ as the probability vector for the correlated stochastic process. Suppose starting at timestep $t+1$, we begin running noiseless dynamics on \textit{both} copies, i.e.~we apply $R_0^{(t)}$, and we continue for an infinite number of gates. Then we will get full convergence to the fixed points $\ket{I^n}\otimes \ket{I^n}$, $\ket{S^n}\otimes \ket{S^n}$ and $\ket{S^n} \otimes \ket{I^n}$. The fourth fixed point $\ket{I^n} \otimes \ket{S^n}$ is not in the accessible subspace. We can compute precisely the probability of each of these outcomes. In Ref.~\cite{dalzell2020anticoncentration}, we arrived at an expression for these probabilities by solving a certain recursion relation. Here, we need only the result of that calculation to inform how we define the diagonal matrices $L_I$ and $L_S$:
\begin{align}
    L_I &= \sum_{\vec{\nu}} \frac{1-q^{-2n+2|\vec{\nu}|}}{1-q^{-2n}} \ketbra{\vec{\nu}} \\
    L_S &= \sum_{\vec{\nu}} \frac{q^{-2n+2|\vec{\nu}|}-q^{-2n}}{1-q^{-2n}} \ketbra{\vec{\nu}}\,.
\end{align}
Note that $L_I+L_S$ is the identity matrix $\mathcal{I}$. The coefficient of $\ketbra{\vec{\nu}}$ in $L_I$ gives the probability that a configuration that starts at $\ket{\vec{\nu}}$ ends at the $I^n$ fixed point if it undergoes completely noiseless dynamics, and the coefficient in $L_S$ gives the probability of ending at the $S^n$ fixed point \cite{dalzell2020anticoncentration}. 

Then define
\begin{align}
    L_{II} &= L_I \otimes \mathcal{I} \\
    L_{SS} &= \mathcal{I} \otimes L_S \\
    L_{SI} &= \mathcal{I} \otimes L_I - L_I \otimes \mathcal{I} \,,
\end{align}
which are the analogous matrices for the joint dynamics to end at $\ket{I^n}\otimes \ket{I^n}$, $\ket{S^n}\otimes \ket{S^n}$, and $\ket{S^n} \otimes \ket{I^n}$, respectively. 

Now we may define
\begin{align}
    P^{(t)}_I &= L_I P^{(t)} L_I^{-1} \\
    P^{(t)}_S &= L_S P^{(t)} L_S^{-1}
\end{align}
and
\begin{align}
    R^{(t)}_{II} &= L_{II} R_0^{(t)} L_{II}^{-1} \\
    R^{(t)}_{SS} &= L_{SS} R_0^{(t)} L_{SS}^{-1} \\
    R^{(t)}_{SI} &= L_{SI} R_0^{(t)} L_{SI}^{-1}\,,
\end{align}
where in each case $O^{-1}$ denotes the Moore-Penrose pseudo-inverse of $O$. We interpret these matrices as the transition operators for probability mass that has been conditioned to end up at a certain fixed point. For example, $P^{(t)}_S$ is the transition operator for a single copy conditioned on eventually ending up at the $S^n$ fixed point. Even though the walk is generally biased toward $I$, it will be biased toward $S$ when conditioned on ending at the $S^n$ fixed point. The following lemma asserts that these are indeed stochastic matrices. All lemmas stated here are proved in \autorefapp{app:deferredproofslemmas}.

\begin{lemma}\label{lem:stochasticmatrices}
    The matrices $P^{(t)}_I$, $P^{(t)}_S$, $R^{(t)}_{II}$, $R^{(t)}_{SS}$, $R^{(t)}_{SI}$, restricted to their support, are stochastic matrices.
\end{lemma}

The next lemma asserts that if the $X \otimes Y$ system undergoes dynamics under $R_{SI}^{(t)}$, then the $W$ system undergoes dynamics under $P_I^{(t)}$. This makes sense, since conditioning on $X$ to go to $S^n$ and $Y$ to go to $I^n$ should be equivalent to conditioning the $W$ system to go to $I^n$.  

\begin{lemma}\label{lem:WsystemEvolvesByPI}
Within the accessible subspace, the following holds.
    \begin{equation}
        \Delta R^{(t)}_{SI} = P^{(t)}_I \Delta\,.
    \end{equation}
\end{lemma}

We now introduce some more notation. For any vector $\ket{x}$ on a single copy of the vector space, let
\begin{align}
    \ket{x_I} &= L_I \ket{x} \\
    \ket{x_S} &= L_S \ket{x}\,,
\end{align}
and for any vector $\ket{v}$ on two copies of the vector space, let 
\begin{align}
    \ket{v_{II}} &= L_{II} \ket{v} \\
    \ket{v_{SS}} &= L_{SS} \ket{v} \\
    \ket{v_{SI}} &= L_{SI} \ket{v} \,.
\end{align}
Thus, if $\ket{x}$ represents a probability distribution over the $2^n$ basis states on a single copy of the Hilbert space, then the vector $\ket{x_I}$ is the portion of $\ket{x}$ that is destined to end at the fixed point $I^n$, and $\ket{x_S}$ is the portion destined to end at $S^n$ (if all future gates are noiseless). The division of probability mass into separate $I$ and $S$-destined parts is depicted schematically in \autoref{fig:destinedmass}.

The amount of probability mass for which the noisy copy is destined for the $S^n$ fixed point cannot decay too quickly with the number of noise locations (note that if the noisy copy ends at $S^n$, the noiseless copy must also end at $S^n$). In \autoref{fig:destinedmass}, this is depicted by the fact that the only way to transition from the $S$-destined to an $I$-destined division of probability mass is due to the action of a noise location, which induces a $S\rightarrow I$ transition with probability $\sigma$. 
\begin{lemma}\label{lem:decayOfvSS}
    The $S$-destined probability mass obeys the following inequality, for any $t' \geq t$.
    \begin{equation}
        \braket{\mathbf{1}, \mathbf{1}}{v_{SS}^{(t')}} \geq (1-\sigma)^{2(t'-t)}\braket{\mathbf{1}, \mathbf{1}}{v_{SS}^{(t)}}\,.
    \end{equation}
\end{lemma}
\begin{proof}[Proof idea]
Recall that the inner product with $\bra{\mathbf{1},\mathbf{1}}$ gives the sum of the entries of the vector. We interpret $\ket{v_{SS}^{(t)}}$ as the probability vector of mass destined to reach the $S^n$ fixed point on both copies. Each time a noise location acts, it can affect at most a $\sigma$ fraction of the mass, so even after two noise locations act, at least a $(1-\sigma)^2$ fraction of the mass that was $S$-destined before will still be $S$-destined. 
\end{proof}

\subsubsection{Decomposing the $I$-destined probability mass}

The final piece of machinery we need is an accounting of which error leads to each piece of $I$-destined probability mass. To do this, for each $t \geq 1$ define
\begin{align}
    \ket{v_{SI}^{(t,t)}} ={}& \ket{v_{SI}^{(t)}}-(\mathcal{I} \otimes Q'^{(t)}_{\sigma} Q_\sigma^{(t)})R_{SI}^{(t)}\ket{v_{SI}^{(t-1)}} \\
    ={}& \left(L_{SI} \left(\mathcal{I} \otimes Q'^{(t)}_{\sigma} Q_\sigma^{(t)}\right)-\left(\mathcal{I} \otimes Q'^{(t)}_{\sigma} Q_\sigma^{(t)}\right)L_{SI}\right)R_0^{(t)}\ket{v^{(t-1)}}\,,
\end{align}
and define the evolution rule
\begin{equation}
\ket{v_{SI}^{(t'+1,t)}} = Q'^{(t'+1)}_\sigma Q_\sigma^{(t'+1)}R_{SI}^{(t'+1)}\ket{v_{SI}^{(t',t)}} \,.
\end{equation}
The vector $\ket{v_{SI}^{(t',t)}}$ represents the probability mass that would have gone to the $S^n$ fixed point, but the noise at time step $t$ caused it to be redirected to the $I^n$ fixed point, and we have subsequently evolved it forward to timestep $t'$. 

Importantly, we can verify from the definition that
\begin{align}\label{eq:sumSI}
    \sum_{t=1}^{t'} \ket{v_{SI}^{(t',t)}} = \ket{v_{SI}^{(t')} }\,,
\end{align}
indicating that all of the mass at time step $t'$ is accounted for as having originated at some previous time step $t$. 

\begin{lemma}\label{lem:11vSI}
For all $t$ and $t'\geq t$,
    \begin{equation}
        \braket{\mathbf{1},\mathbf{1}}{v_{SI}^{(t',t)}} \leq (1-(1-\sigma)^2) \braket{\mathbf{1},\mathbf{1}}{v_{SS}^{(t-1)}}\,.
    \end{equation}
\end{lemma}
\begin{proof}[Proof idea]
The vector $\ket{v_{SI}^{(t,t)}}$ represents the mass that satisfies two conditions: (1) it was destined for the $\ket{S^n} \otimes \ket{S^n}$ fixed point at time step $t-1$, and (2) the noise at time step $t$ caused it to be destined for the $\ket{S^n} \otimes \ket{I^n}$ fixed point at time step $t$. At most $\braket{\mathbf{1},\mathbf{1}}{v_{SS}^{(t-1)}}$ mass qualifies under condition (1). Among that mass, each of the two noise location can only impact a $\sigma$ fraction of the mass, so the fraction of mass that can be re-directed is at most $(1-(1-\sigma)^2)$.
\end{proof}

\subsection{Consequences of anti-concentration}

In all of our rigorous proofs, we assume we have a random quantum circuit architecture that is $h$-regularly connected for some constant $h=O(1)$, and has anti-concentration size equal to $s_{AC}$. Recall that this means that $Z_0$ becomes twice its limiting value at $s_{AC}$. When this is the case, we have the following lemmas. All constants are dependent on $q$ and $h$, but not on $n$ or any noise parameters. 

\begin{lemma}\label{lem:anticoncentrationInNoisySection}
     Suppose the random quantum circuit architecture is regularly connected. There exist constants $\chi_1$ and $\chi_2$ such that for all $t \geq s_{AC}$
    \begin{equation}
         \braket{\mathbf{q},\mathbf{1}}{v^{(t)}} \leq \frac{2q^n}{q^n+1} + \eta_t\,,
    \end{equation}
    where
    \begin{equation}
        \eta_t = \chi_2 \exp\left(-\frac{\chi_1}{n}(t-s_{AC})\right)\,.
    \end{equation}
\end{lemma}
\begin{proof}[Proof idea]
The left-hand side is precisely $Z_0$ for a circuit with size $t$. The regularly connected property indicates that for any configuration not at a fixed point, there will be a gate that couples an $I$ with an $S$ roughly once every $O(n)$ gates. When this happens, the difference between $Z_0$ and its infinite-size limit is reduced by a constant factor, leading to the scaling in the lemma. 
\end{proof}

\begin{lemma}\label{lem:ACconvergencetoSn}
     Suppose the random quantum circuit architecture is regularly connected. There exist constants $\chi_3$ and $\chi_4$ such that for all $t$
    \begin{equation}
        \braket{S^n,\mathbf{1}}{v^{(t)}} \geq \frac{1-\eta'_t}{q^n+1}\,,
    \end{equation}
    where
    \begin{equation}
        \eta'_t = \chi_4 \exp\left(-\frac{\chi_3}{n}(t-s_{AC})\right)\,.
    \end{equation}
\end{lemma}
\begin{proof}[Proof idea]
Anti-concentration happens because most of the probability mass makes it to one of the fixed points. This lemma states that after the anti-concentration size, most of the mass destined for the $S^n$ fixed point has already reached it. The fraction that has not yet reached is $\eta_t'$, which decays exponentially with $t/n$. We show that if this were not the case, then the bound in  \autoref{lem:anticoncentrationInNoisySection} could not hold.
\end{proof}

\begin{lemma}\label{lem:ACconvergenceUnderPI}
    Suppose the random quantum circuit architecture is regularly connected. There exist constants $\chi_5$ and $\chi_6$ such that for any non-negative vector $\ket{v}$ that is normalized (i.e.~$\braket{\mathbf{1},\mathbf{1}}{v}=1$), the following holds for any $t_0$ and any $t_1\geq t_0$.
    \begin{align}
        &\bra{\mathbf{q}}\Delta \prod_{t=t_0+1}^{t_1}\left((\mathcal{I} \otimes Q'^{(t)}_{\sigma} Q_\sigma^{(t)}) R_{SI}^{(t)}\right) \ket{v}-1 \nonumber\\
        \leq{}& \left(\bra{\mathbf{q}}\Delta \ket{v}-1\right)  \chi_6 \exp\left(-\frac{\chi_5 (t_1-t_0)}{n}\right) \,.
    \end{align}
\end{lemma}
\begin{proof}[Proof idea]
Recall from \autoref{lem:WsystemEvolvesByPI} that if $\ket{v}$ evolves by $R_{SI}^{(t)}$, then $\Delta \ket{v}$ evolves by $P_{I}^{(t)}$. The transition matrix $P_I^{(t)}$ is the matrix that conditions on sending the vector to the $I^n$ fixed point, so it is even more $I$-biased than the transition matrix $P^{(t)}$. Thus, each time a bit is flipped, the Hamming weight is likely to decrease, and the inner product with $\bra{\mathbf{q}}-\bra{\mathbf{1}}$ will be reduced by a constant factor. This will (usually) happen once every $O(n)$ gates if the architecture is regularly connected. The insertion of the $Q_\sigma^{(t)}$ operators will only make the Hamming weight smaller since they can only flip $S \rightarrow I$. 
\end{proof}

\subsection{Exponential clustering of $S$-destined probability mass}

A key step in our analysis is that the $S$-destined mass stays close to the $S^n$ fixed point, as long as $\sigma = O(1/n)$. In fact, the probability of deviating from the fixed point by $x$ bit flips decays exponentially in $x$. Intuitively, this is because the $S$-destined mass is biased to move upward in Hamming weight, and when $\sigma$ is small enough, this upward pressure will be greater than the downward pressure coming from the noise itself. 

We prove this for the $W$ system, which captures the difference between the (noiseless) $X$ and (noisy) $Y$ systems. We cannot directly analyze the $Y$ system because at time step 0, the statement is definitively not true. It takes $s_{AC}$ gates for the $S$-destined mass in the $Y$ system to initially converge. Meanwhile, the $W$ system begins at the $S^n$ fixed point. This is the main reason we introduced the $W$ system in the first place. 

Define the projector
\begin{equation}
    \Pi_w = \sum_{\vec{\nu}: |\vec{\nu}| = w} \ketbra{\vec{\nu}}\,.
\end{equation}
\begin{lemma}\label{lem:expclustering}
There exist constants $\chi_7$, $\chi_8$, $\chi_9$, and $n_0$ such that as long as $\sigma \leq \chi_7/n$ and $n \geq n_0$, the following holds for any $t$ and any integer $w$ with $1 \leq w < n$. 
\begin{equation}
       \frac{\bra{\mathbf{1}}\Pi_w\Delta \ket{v^{(t)}_{SS}}}{\braket{\mathbf{1},\mathbf{1}}{v^{(t)}_{SS}}}\leq n \sigma \xi_{w}\,,
 \end{equation}
where
\begin{equation}
    \xi_{w}= \chi_9(n-w) q^{-(n-w)}e^{-\chi_8 (n-w)}\,.
\end{equation}
\end{lemma}
\begin{proof}[Proof idea]
The $S$-destined portion of the mass within the $W$ system starts at the $S^n$ fixed point. When noise acts at time step $t$, some of the mass moves to Hamming weight $n-1$ but continues to be $S$-destined, and some of it is ``redirected'' to become $I$-destined, which is captured in the $\ket{v_{SI}^{(t,t)}}$ vector. The total amount of redirected mass cannot be too large, as we see in \autoref{lem:11vSI}. Moreover, the redirected mass must steadily move downward in Hamming weight (after all, it is $I$-destined), which we quantify with \autoref{lem:ACconvergenceUnderPI}. This is important because for each value of the Hamming weight $w$, the amount of $S$-destined mass divided by the amount of $I$-destined mass at that Hamming weight is precisely $\frac{q^{-2n+2w} - q^{-2n}}{1-q^{-2n+2w}} \approx q^{-2(n-w)}$, so as the $I$-destined mass moves down in Hamming weight, the $S$-destined mass that corresponds to it decreases exponentially. After accounting for each bit of $I$-destined mass by summing over all $\ket{v_{SI}^{(t',t)}}$, we can prove the lemma. 
\end{proof}

\subsection{Relating $\mathcal{Z}_\sigma$ to the amount of $S$-destined probability mass}

The following lemma states that keeping track of the amount of $S$-destined mass is sufficient to get good upper and lower bounds on the quantity $\mathcal{Z}_\sigma$. 

\begin{lemma}\label{lem:Zsigmaupperlowerbound}
The following lower bound always holds
    \begin{equation}
       \mathcal{Z}_\sigma -1  \geq \left(q^n-1\right)\braket{\mathbf{1},\mathbf{1}}{v_{SS}^{(s)}}
    \end{equation}
Moreover, there exist constants $\chi_{10}$, $\chi_{11}$, $\chi_{12}$, $\chi_{13}$, and $n_0$ such that as long as $\sigma \leq \chi_{13}/n$ and $n\geq n_0$, the following upper bound holds.
\begin{equation}
            \mathcal{Z}_\sigma - 1 \leq 
        \left(q^n-1\right)\braket{\mathbf{1},\mathbf{1}}{v_{SS}^{(s)}}\exp\left(1+\chi_{10} n\sigma+\chi_{12}  e^{-\frac{\chi_{11}}{n}(s-s_{AC})+4s\sigma}\right)
\end{equation}

\end{lemma}
\begin{proof}[Proof idea]
For each $w$, we know the ratio of the $I$-destined and $S$-destined mass at Hamming weight $w$: for each portion of $S$-destined probability mass, there is roughly $q^{2(n-w)}$ $I$-destined probability mass. This decreases with $w$ like $q^{-2w}$. The contribution of mass at Hamming weight $w$ to $\mathcal{Z}_\sigma$ increases, but at the slower rate of $q^{w}$. Thus, for a fixed amount of $S$-destined mass, $\mathcal{Z}_\sigma$ is minimized when all of it is at the $S^n$ fixed point, leading to our lower bound. On the other hand, we know that the $S$-destined mass is exponentially clustered near the $S^n$ fixed point (\autoref{lem:expclustering}), so this lower bound cannot be too loose, which we leverage into an upper bound. 
\end{proof}

\subsection{Bounding the $S$-destined mass}

Now, all that remains is to compute the amount of $S$-destined mass. Here we show upper and lower bounds on this quantity for layered architectures and for the complete-graph architecture. 

\begin{lemma}\label{lem:boundSdestinedmassLAYERS}
Suppose the random quantum circuit architecture is regularly connected and layered. Let $d_{AC}$ be its anti-concentration depth. Then, for any $d$,
\begin{equation}
    \braket{\mathbf{1},\mathbf{1}}{v_{SS}^{(dn/2)}} \geq \frac{\left(1-\frac{1-(1-\sigma(1-q^{-2}))^n}{1-q^{-2n}}\right)^d }{q^n+1} \,.
\end{equation}
Moreover, there exist constants $a_0$, $a_1$, $a_2$, $a_3$, and $n_0$ such that, as long as $\sigma \leq a_3/n$ and $n \geq n_0$, 
\begin{align}
    \braket{\mathbf{1}, \mathbf{1}}{v_{SS}^{(dn/2)}} \leq \frac{\left(1-\frac{1-(1-\sigma(1-q^{-2}))^n}{1-q^{-2n}}\right)^d }{q^n+1}e^{a_0\sigma^2 dn+ a_1\sigma n d_{AC} +a_2 n\sigma \log(1/(n\sigma))}\,,
\end{align}
where $d_{AC}$ is the anti-concentration depth. 
\end{lemma}

\begin{lemma}\label{lem:boundSdestinedmassCG}
Suppose the random quantum circuit architecture is the complete-graph architecture. Let $s_{AC}$ be its anti-concentration size. Then, for any $s$,
\begin{equation}
    \braket{\mathbf{1},\mathbf{1}}{v_{SS}^{(s)}} \geq \frac{\left(1-\frac{1-(1-\sigma(1-q^{-2}))^2}{1-q^{-2n}}\right)^s }{q^n+1} \,.
\end{equation}
Moreover, there exist constants $b_0$, $b_1$, $b_2$, $b_3$, and $n_0$ such that, as long as $\sigma \leq b_3/n$ and $n \geq n_0$,
\begin{align}
    \braket{\mathbf{1}, \mathbf{1}}{v_{SS}^{(s)}} \leq \frac{\left(1-\frac{1-(1-\sigma(1-q^{-2}))^2}{1-q^{-2n}}\right)^s }{q^n+1}e^{b_0\sigma^2 s + b_1\sigma s_{AC} +b_2 n\sigma \log(1/(n\sigma))}
\end{align}
\end{lemma}

\begin{proof}[Proof idea for \autoref{lem:boundSdestinedmassLAYERS} and \autoref{lem:boundSdestinedmassCG}]
When a portion of $S$-destined mass is at the $S^n$ fixed point, and noise acts to move it to Hamming weight $n-1$, we have a good understanding of what fraction remains $S$-destined. Specifically, there is a $\frac{q^{-2}-q^{-2n}}{1-q^{-2n}}$ chance that it re-equilibrates to $S^n$. We also know the chance that it will make the transition in the first place; the transition from $S \rightarrow I$ happens with probability precisely $\sigma$. This scenario gives the maximum amount of lost $S$-destined mass, and gives rise to our lower bound. However, if the portion of $S$-destined mass is not at the $S^n$ fixed point, then this is complicated in two ways. First, the probability of re-equilibrating back to $S^n$ is a slightly different expression, and, more importantly, the noise will not cause a transition as often, as there is a chance it acts on a bit that is already 0. If the configuration has Hamming weight $w$ and the noise acts on a random bit, the chance of a transition is $\frac{n-w}{n}\sigma$ so a smaller amount of $S$-destined mass is lost at each step. Luckily, we know that the $S$-destined mass is exponentially clustered near $w=n$ (\autoref{lem:expclustering}), so the corrections are small, which gives rise to the upper bound.

We utilize the layered architecture property to be able to say that \textit{every} qudit is acted upon by noise after each layer, and thus, from the perspective of the amount of $S$-destined mass, all that matters is the Hamming weight of the configuration prior to the noise. The same is true for the complete-graph case because the gates are chosen randomly and each qudit is equally likely to participate. However, we do not believe this property is necessary for our result to be true.  
\end{proof}

\subsection{Deferred proofs of lemmas}\label{app:deferredproofslemmas}

\subsubsection{Proof of \autoref{lem:stochasticmatrices}}

\begin{proof}
    We demonstrate this for $P^{(t)}_I$ and leave the others to be verified in a similar fashion. First of all, since $P^{(t)}$ is a stochastic matrix, its matrix elements are non-negative. Since $L_I$ and $L_I^{-1}$ are diagonal matrices with non-negative entries, $P^{(t)}_I = L_I P^{(t)} L_I^{-1}$ also has non-negative matrix elements. The support of $P_I$ is the entire vector space except for the span of $\ket{S^n}$. Consider another basis state $\ket{\vec{\nu}}$. Since gate $t$ acts on qudits $\{i_t,j_t\}$, if $\nu_{i_t} = \nu_{j_t}$ then it is a $+1$ eigenvector of $\ket{P^{(t)}}$ and
    \begin{align}
        \bra{\mathbf{1}}P^{(t)}_I\ket{\vec{\nu}} &= \sum_{\vec{\mu}} \bra{\vec{\mu}} L_I P^{(t)} L_I^{-1} \ket{\vec{\nu}} \\
        &= \sum_{\vec{\mu}} \frac{1-q^{-2n+2|\vec{\mu}|}}{1-q^{-2n+2|\vec{\nu}|}} \bra{\vec{\mu}} P^{(t)} \ket{\vec{\nu}} \\
        &=\sum_{\vec{\mu}} \frac{1-q^{-2n+2|\vec{\mu}|}}{1-q^{-2n+2|\vec{\nu}|}} \braket{\vec{\mu}}{\vec{\nu}}  = 1\,.
    \end{align}
    If $\nu_{i_t} \neq \nu_{j_t}$, then $P^{(t)}$ sends $\ket{\vec{\nu}}$ to a basis state with Hamming weight reduced by 1 with probability $q^2/(q^2+1)$, and to Hamming weight increased by 1 with probability $1/(q^2+1)$, so
    \begin{align}
        \bra{\mathbf{1}}P^{(t)}_I\ket{\vec{\nu}} 
        &= \sum_{\vec{\mu}} \frac{1-q^{-2n+2|\vec{\mu}|}}{1-q^{-2n+2|\vec{\nu}|}} \bra{\vec{\mu}} P^{(t)} \ket{\vec{\nu}} \\
        &=\left(\frac{q^2}{q^2+1}\frac{1-q^{-2n+2|\vec{\nu}|-2}}{1-q^{-2n+2|\vec{\nu}|}}+\frac{1}{q^2+1}\frac{1-q^{-2n+2|\vec{\nu}|+2}}{1-q^{-2n+2|\vec{\nu}|}} \right) = 1\,.
    \end{align}
    This demonstrates $P_I^{(t)}$ is a stochastic matrix when restricted to its support. 
\end{proof}

\subsubsection{Proof of \autoref{lem:WsystemEvolvesByPI}}
\begin{proof}
    We consider the action of both sides of the equation on an input state $\ket{\vec{\nu},\vec{\mu}}$. Let $a$ and $b$ be the number of $1$ entries in $\vec{\nu}$ and $\vec{\mu}$, excluding the positions $\{i_t,j_t\}$, respectively, and let $c$ be the number of entries on which $\vec{\nu}$ and $\vec{\mu}$ agree. Since we are restricting to the accessible subspace, we have $c=n-2-a+b$. Since $\Delta$ is a tensor product across all bits $i \in \{0,\ldots,n-1\}$, and both $P_I^{(t)}$ and $R_{SI}^{(t)}$ modify only bits $i_t$ and $j_t$, it is sufficient to consider the transitions among just bits $i_t$ and $j_t$. First, define
\begin{align}
    c_0 &= \frac{1-q^{-2n+2c}}{1-q^{-2n+2c+2}}\frac{q^2}{q^2+1}  \\
    c_1 &= \frac{1-q^{-2n+2c+4}}{1-q^{-2n+2c+2}}\frac{1}{q^2+1} \,.
\end{align}
    Let the four bits below be ordered $X_{i_t}X_{j_t}\,,\,Y_{i_t}Y_{j_t}$. The right-hand side has the following effect, where the first arrow is application of $\Delta$ and the second is application of $P_I^{(t)}$. 
    \begin{align}
        \ket{SS,SS} &\rightarrow \ket{SS} \rightarrow \ket{SS}\nonumber\\
        \ket{SS,SI} &\rightarrow \ket{SI} \rightarrow c_0\ket{II} + c_1\ket{SS} \nonumber\\
        \ket{SS,IS} &\rightarrow \ket{IS} \rightarrow c_0\ket{II} + c_1\ket{SS} \nonumber\\
        \ket{SS,II} &\rightarrow \ket{II} \rightarrow \ket{II}\nonumber\\
        \ket{SI,SI} &\rightarrow \ket{SS} \rightarrow \ket{SS}\nonumber\\
        \ket{SI,II} &\rightarrow \ket{IS} \rightarrow c_0\ket{II} + c_1\ket{SS} \nonumber\\
        \ket{IS,IS} &\rightarrow \ket{SS} \rightarrow \ket{SS} \nonumber\\
        \ket{IS,II} &\rightarrow \ket{SI} \rightarrow c_0\ket{II} + c_1\ket{SS} \nonumber\\ 
        \ket{II,II} & \rightarrow \ket{SS} \rightarrow \ket{SS}\,. \nonumber
    \end{align}

    Now, we can do the same for the left-hand side. For example, consider the input state $\ket{SS,SI}$. Action by $R_{SI}^{(t)}$ sends it to
    \begin{align}
        \ket{SS,SI} \rightarrow{}& \frac{q^{-2n+2a+4}-q^{-2n+2b}}{q^{-2n+2a+4} - q^{-2n+2b+2}}\frac{q^2}{q^2+1}\ket{SS,II} + \frac{q^{-2n+2a+4}-q^{-2n+2b+4}}{q^{-2n+2a+4} - q^{-2n+2b+2}}\frac{1}{q^2+1}\ket{SS,SS} \\
        ={}& c_0 \ket{SS,II} + c_1 \ket{SS,SS}\,,
    \end{align}
    where the last line follows by recalling the relation $c=n-2-a+b$. Action by $\Delta$ then yields the state $c_0 \ket{II} + c_1 \ket{SS}$. We can now list this calculation for each input state, where the first arrow is action by $R_{SI}^{(t)}$ and the second by $\Delta$. 
    \begin{align}
        \ket{SS,SS} \rightarrow{}& \ket{SS,SS} \rightarrow \ket{SS}\\
        \ket{SS,SI} \rightarrow{}& c_0\ket{SS,II} + c_1\ket{SS,SS} \rightarrow c_0\ket{II} + c_1 \ket{SS}  \\
        \ket{SS,IS} \rightarrow{}& c_0\ket{SS,II} + c_1\ket{SS,SS} \rightarrow c_0\ket{II} + c_1 \ket{SS}  \\
        \ket{SS,II} \rightarrow{}& \ket{SS,II} \rightarrow \ket{II}\\
        \ket{SI,SI} \rightarrow{}&  \frac{q^{-2n+2a}-q^{-2n+2b}}{q^{-2n+2a+2}-q^{-2n+2b+2}}\frac{q^2}{q^2+1}\ket{II,II} +\frac{q^{-2n+2a+4}-q^{-2n+2b+4}}{q^{-2n+2a+2}-q^{-2n+2b+2}}\frac{1}{q^2+1}\ket{SS,SS} \\
        \rightarrow{}& \ket{SS}\\
        \ket{SI,II} \rightarrow{}& c_1\ket{II,II} + c_0 \ket{SS,II} \rightarrow c_0\ket{II} + c_1 \ket{SS} \\
        \ket{IS,IS} \rightarrow{}& \frac{q^{-2n+2a}-q^{-2n+2b}}{q^{-2n+2a+2}-q^{-2n+2b+2}}\frac{q^2}{q^2+1}\ket{II,II}+\frac{q^{-2n+2a+4}-q^{-2n+2b+4}}{q^{-2n+2a+2}-q^{-2n+2b+2}}\frac{1}{q^2+1}\ket{SS,SS} \\
        \rightarrow{}& \ket{SS}\\
        \ket{IS,II} \rightarrow{}& c_1\ket{II,II} + c_0 \ket{SS,II} \rightarrow c_0\ket{II} + c_1 \ket{SS} \\\\ 
        \ket{II,II} \rightarrow{}& \ket{II,II} \rightarrow \ket{SS} \,,
    \end{align}
    which verifies that the left-hand and right-hand sides are equal. 
\end{proof}

\subsubsection{Proof of \autoref{lem:decayOfvSS}}
\begin{proof}
    \begin{align}
        &\braket{\mathbf{1}, \mathbf{1}}{v_{SS}^{(t)}} =\bra{\mathbf{1}, \mathbf{1}}L_{SS}R_\sigma^{(t)}\ket{v^{(t-1)}} =  \bra{\mathbf{1}, \mathbf{1}}L_{SS}R^{(t)}_\sigma L_{SS}^{-1}\ket{v_{SS}^{(t-1)}}\\
        ={}&\sum_{\substack{\vec{\mu}\\\vec{\nu} \neq I^n}} \bra{\mathbf{1}, \mathbf{1}}L_{SS} \ketbra{\mathbf{1},\vec{\mu}}R^{(t)}_\sigma  \ketbra{\mathbf{1},\vec{\nu}} L_{SS}^{-1}\ket{v_{SS}^{(t-1)}} \\
        ={}& \sum_{\substack{\vec{\mu}\\\vec{\nu} \neq I^n}} \frac{q^{-2n+2|\vec{\mu}|}-q^{-2n}}{q^{-2n+2|\vec{\nu}|}-q^{-2n}} \bra{\mathbf{1},\vec{\mu}}R^{(t)}_\sigma  \ket{\mathbf{1},\vec{\nu}}\braket{\mathbf{1},\vec{\nu}}{v_{SS}^{(t-1)}}\\
        ={}& \sum_{\substack{\vec{\mu}\\\vec{\nu} \neq I^n}} \frac{q^{-2n+2|\vec{\mu}|}-q^{-2n}}{q^{-2n+2|\vec{\nu}|}-q^{-2n}} \bra{\vec{\mu}}Q'^{(t)}_\sigma Q^{(t)}_\sigma P^{(t)}  \ket{\vec{\nu}}\braket{\mathbf{1},\vec{\nu}}{v_{SS}^{(t-1)}}\\
        ={}& \sum_{\substack{\vec{\mu}\\\vec{\nu}, \vec{\zeta} \neq I^n}} E_{\vec{\mu}\vec{\zeta}} G_{\vec{\zeta}\vec{\nu}} \braket{\mathbf{1},\vec{\nu}}{v_{SS}^{(t-1)}}
    \end{align}
    where
    \begin{align}
        E_{\vec{\mu}\vec{\zeta}} &= 
         \frac{q^{-2n+2|\vec{\mu}|}-q^{-2n}}{q^{-2n+2|\vec{\zeta}|}-q^{-2n}} \bra{\vec{\mu}}Q'^{(t)}_\sigma Q^{(t)}_\sigma \ket{\vec{\zeta}}\\
        G_{\vec{\zeta}\vec{\nu}} &= \frac{q^{-2n+2|\vec{\zeta}|}-q^{-2n}}{q^{-2n+2|\vec{\nu}|}-q^{-2n}}\bra{\vec{\zeta}}P^{(t)}  \ket{\vec{\nu}} = \bra{\vec{\zeta}}P_S^{(t)}  \ket{\vec{\nu}}
    \end{align}
    However, note that $E_{\vec{\zeta}\vec{\zeta}} \geq (1-\sigma)^2$ (with equality when $\zeta_{i_t} = \zeta_{j_t} = 1$), and all $E_{\vec{\mu}\vec{\zeta}}$ are non-negative. Moreover, note that
    \begin{equation}
        \sum_{\vec{\zeta}} G_{\vec{\zeta}\vec{\nu}} = 1\,,
    \end{equation}
    owing to the fact that $P_S^{(t)}$ is stochastic. 
    Thus $\braket{\mathbf{1}, \mathbf{1}}{v_{SS}^{(t)}} \geq (1-\sigma)^2 \braket{\mathbf{1}, \mathbf{1}}{v_{SS}^{(t-1)}}$, and by recursion, the statement holds. 
\end{proof}

\subsubsection{Proof of \autoref{lem:11vSI}}

\begin{proof}
    Recall that $L_{SI} = \mathcal{I} \otimes L_I - L_I \otimes \mathcal{I}$, but the second term commutes with $\mathcal{I} \otimes Q'^{(t)}_\sigma Q_\sigma^{(t)}$, thus we may ignore it in the following calculation.
    \begin{align}
        \braket{\mathbf{1},\mathbf{1}}{v_{SI}^{(t,t)}}  &= \sum_{\vec{\mu},\vec{\nu}} \bra{\vec{\mu}} L_IQ'^{(t)}_{\sigma} Q_\sigma^{(t)}-Q'^{(t)}_{\sigma} Q_\sigma^{(t)}L_I\ket{\vec{\nu}}\bra{\mathbf{1},\vec{\nu}}R_0^{(t)}\ket{v^{(t-1)}} \\
        &=  \sum_{\vec{\mu},\vec{\nu}}\frac{q^{-2n+2|\vec{\nu}|}-q^{-2n+2|\vec{\mu}|}}{1-q^{-2n}} \bra{\vec{\mu}} Q'^{(t)}_{\sigma} Q_\sigma^{(t)}\ket{\vec{\nu}}\bra{\mathbf{1},\vec{\nu}}R_0^{(t)}\ket{v^{(t-1)}} 
    \end{align}
    If $\vec{\mu} = \vec{\nu}$ the factor gives 0. For each $\vec{\nu}$ there are at most three possible $\vec{\mu} \neq \vec{\nu}$ for which the matrix element $\bra{\vec{\mu}} Q'^{(t)}_{\sigma} Q_\sigma^{(t)}\ket{\vec{\nu}} \neq 0$, corresponding to a single error on either qudit or an error on both at once. In those cases, the matrix element is $\sigma(1-\sigma)$ (for single error) or $\sigma^2$ (for double error). The double error is only possible if $|\vec{\nu}| \geq 2$, but note that we may assume $|\vec{\nu}| \neq 1$ since action by $R_0^{(t)}$ will leave the two bits it acts on equal, and cannot lead to a configuration with Hamming weight 1. We have
    \begin{align}
        &\sum_{\vec{\mu}}\frac{q^{-2n+2|\vec{\nu}|}-q^{-2n+2|\vec{\mu}|}}{1-q^{-2n}} \bra{\vec{\mu}} Q'^{(t)}_{\sigma} Q_\sigma^{(t)}\ket{\vec{\nu}} \nonumber \\
        \leq{}& 2\sigma(1-\sigma)\frac{q^{-2n+2|\vec{\nu}|}-q^{-2n+2|\vec{\nu}|-2}}{1-q^{-2n}}+\sigma^2\frac{q^{-2n+2|\vec{\nu}|}-q^{-2n+2|\vec{\nu}|-4}}{1-q^{-2n}} \label{eq:lemvSIttintermediate}\\
        ={}& \left(\frac{q^{-2n+2|\vec{\nu}|}-q^{-2n}}{1-q^{-2n}}\right)\frac{2\sigma(1-\sigma)(1-q^{-2}) + \sigma^2(1-q^{-4})}{1-q^{-2|\vec{\nu}|}} \\
        \leq{}& \left(\frac{q^{-2n+2|\vec{\nu}|}-q^{-2n}}{1-q^{-2n}}\right)\left(2\sigma-\sigma^2\right)\,.
    \end{align}
    This lets us say
    \begin{align}
         \braket{\mathbf{1},\mathbf{1}}{v_{SI}^{(t,t)}} &\leq \sum_{\vec{\nu}}\left(\frac{q^{-2n+2|\vec{\nu}|}-q^{-2n}}{1-q^{-2n}}\right)\left(2\sigma-\sigma^2\right)\bra{\mathbf{1},\vec{\nu}}R_0^{(t)}\ket{v^{(t-1)}}  \\
         &=\sum_{\vec{\nu}}\left(2\sigma-\sigma^2\right)\bra{\mathbf{1},\vec{\nu}}L_{SS}R_0^{(t)}\ket{v^{(t-1)}} \\
         &=\sum_{\vec{\nu}}\left(2\sigma-\sigma^2\right)\bra{\mathbf{1},\vec{\nu}}R_{SS}^{(t)}L_{SS}\ket{v^{(t-1)}} \\
         &=\left(2\sigma-\sigma^2\right)\sum_{\vec{\nu}}\bra{\mathbf{1},\vec{\nu}}R_{SS}^{(t)}\ket{v_{SS}^{(t-1)}} \\
         &= (1-(1-\sigma)^2) \braket{\mathbf{1},\mathbf{1}}{v_{SS}^{(t-1)}}\,,
    \end{align}
    where the last equality follows because $R_{SS}$ is stochastic. 
    
    The fact that this is also true for $\ket{v^{(t',t)}}$ with $t' >t$ follows from the fact that $\ket{v^{(t',t)}}$ is related to $\ket{v^{(t,t)}}$ by a sequence of stochastic matrices, which preserves the left-hand side of the lemma statement. 
\end{proof}

\subsubsection{Proof of \autoref{lem:anticoncentrationInNoisySection}}
\begin{proof}
    This proof is similar to the proof of the general upper bound on the collision probability in Ref.~\cite{dalzell2020anticoncentration}. Define $Z^{(t')} = \braket{\mathbf{q},\mathbf{1}}{v^{(t')}}$. If the anti-concentration size is $s_{AC}$, this means that 
    \begin{equation}
        Z^{(s_{AC})} \leq 2q^n Z_H = \frac{4q^n}{q^n+1}\,.
    \end{equation}
    where $Z_H=2/(q^n+1)$ is the limiting value of the collision probability studied in Ref.~\cite{dalzell2020anticoncentration}.
    Note that $Z^{(t')}$ is monotonically non-increasing with $t'$ (i.e.,~collision probability only decreases as more gates are applied). Recall that for architectures where the circuit diagram is random, $\ket{v^{(t')}}$ represents an average over choice of circuit diagram. The $h$-regularly connected property says that, no matter what the circuit diagram has looked like up to time step $t'$, given any partition of the qudits into two parts, there is at least a $1/2$ probability that the next $hn$ gates in the circuit diagram will include at least one gate that couples qudits from opposite parts. Conditioned on coupling the two parts, the portion of the collision probability associated with configurations not already at a fixed point will decrease by a factor $2q/(q^2+1)$, as was seen in the general upper bound on the collision probability in Ref.~\cite{dalzell2020anticoncentration}.
    Thus for all $t'$, 
    \begin{align}
        Z^{(t'+rn)}-\frac{2q^n}{q^n+1} &\leq \left(\frac{1}{2} + \frac{1}{2}\frac{2q}{q^2+1}\right) \left(Z^{(t')}-\frac{2q^n}{q^n+1}\right) \\
        &=\frac{(q+1)^2}{2(q^2+1)}\left( Z^{(t')} -\frac{2q^n}{q^n+1}\right)\,.
    \end{align}
    Applying the above recursively, we have
    \begin{equation}
        Z^{(s_{AC}+zhn)} -\frac{2q^n}{q^n+1} \leq \left(\frac{(q+1)^2}{2(q^2+1)}\right)^z \frac{2q^n}{q^n+1} \leq 2\left(\frac{(q+1)^2}{2(q^2+1)}\right)^z\,.
    \end{equation}
    Now we ensure something similar holds for every value of $t$ and not just $t=s_{AC} +zhn$ for integers $z$.  Let $t_0$ be the maximum integer for which $t_0 \leq t$, and $t_0=s_{AC}+z_0hn$ for some integer $z_0$. So $t-t_0 \leq hn$ and $z_0 \geq (t-s_{AC})/(hn)-1$. Moreover, by monotonicity, we have $Z^{(t)} \leq Z^{(t_0)}$. Together, this implies
    \begin{align}
        Z^{(t)} &\leq \frac{2q^n}{q^n+1} + 2\left(\frac{(q+1)^2}{2(q^2+1)}\right)^{z_0} ={} \frac{2q^n}{q^n+1} + 2\left(\frac{(q+1)^2}{2(q^2+1)}\right)^{\frac{t-s_{AC}}{hn}-1} \\
        &= \frac{2q^n}{q^n+1} + \chi_2 e^{-\chi_1 (t-s_{AC})/n}\,,
    \end{align}
    where $\chi_2 = 4(q^2+1)/(q+1)^2$ and $\chi_1 = \frac{1}{h}\log(2(q^2+1)/(q+1)^2)$.
\end{proof}

\subsubsection{Proof of \autoref{lem:ACconvergencetoSn}}
\begin{proof}
We have 
\allowdisplaybreaks
\begin{align}
\frac{\braket{\mathbf{q},\mathbf{1}}{v^{(t)}}-1}{q^n-1} 
={}& \sum_{\vec{\nu}}\frac{q^{|\vec{\nu}|}-1}{q^n-1} \braket{\vec{\nu},\mathbf{1}}{v^{(t)}} \\
={}& \braket{S^n,\mathbf{1}}{v^{(t)}} +\sum_{\vec{\nu} \neq I^n,S^n} \frac{q^{|\vec{\nu}|}-1}{q^n-1} \braket{\vec{\nu},\mathbf{1}}{v^{(t)}} \\
={}& \braket{S^n,\mathbf{1}}{v^{(t)}} +\sum_{\vec{\nu} \neq I^n,S^n} \frac{q^{|\vec{\nu}|}-1}{q^n-1} \bra{\vec{\nu},\mathbf{1}}(L_S^{-1}L_S \otimes \mathcal{I})\ket{v^{(t)}} \\
={}& \braket{S^n,\mathbf{1}}{v^{(t)}} +\sum_{\vec{\nu} \neq I^n,S^n}\frac{\left(1-q^{-2n}\right)\left(q^{|\vec{\nu}|}-1\right)}{(q^{-2n+2|\vec{\nu}|} - q^{-2n})(q^n-1)} \bra{\vec{\nu},\mathbf{1}}L_S \otimes \mathcal{I}\ket{v^{(t)}} \\
\geq{}&  \braket{S^n,\mathbf{1}}{v^{(t)}} +\frac{\left(1-q^{-2n}\right)\left(q^{n-1}-1\right)}{(q^{-2} - q^{-2n})(q^n-1)}\sum_{\vec{\nu} \neq I^n,S^n} \bra{\vec{\nu},\mathbf{1}}L_S \otimes \mathcal{I}\ket{v^{(t)}} \\
={}& \braket{S^n,\mathbf{1}}{v^{(t)}} +\frac{q\left(1+q^{-n}\right)}{1 + q^{-n+1}}\sum_{\vec{\nu} \neq I^n,S^n} \bra{\vec{\nu},\mathbf{1}}L_S \otimes \mathcal{I}\ket{v^{(t)}} \\
%
%
={}&-\left(\frac{q\left(1+q^{-n}\right)}{1 + q^{-n+1}}-1\right)\braket{S^n,\mathbf{1}}{v^{(t)}} +\frac{q\left(1+q^{-n}\right)}{1 + q^{-n+1}}\sum_{\vec{\nu} \neq I^n} \bra{\vec{\nu},\mathbf{1}}L_S \otimes \mathcal{I}\ket{v^{(t)}} \\
={}&-\frac{q-1}{1+q^{-n+1}}\braket{S^n,\mathbf{1}}{v^{(t)}} +\frac{q\left(1+q^{-n}\right)}{1 + q^{-n+1}}\bra{\mathbf{1},\mathbf{1}}L_S \otimes \mathcal{I}\ket{v^{(t)}} \\
={}& -\frac{q-1}{1+q^{-n+1}}\braket{S^n,\mathbf{1}}{v^{(t)}} +\frac{q\left(1+q^{-n}\right)}{1 + q^{-n+1}}\frac{1}{q^n+1}\,,
\end{align}
\allowdisplaybreaks[0]
where the last line follows because the total amount of $S$-destined mass for the noiseless copy is exactly $1/(q^n+1)$.
From \autoref{lem:anticoncentrationInNoisySection}, we have
\begin{align}
     \frac{\braket{\mathbf{q},\mathbf{1}}{v^{(t)}}-1}{q^n-1} &\leq \frac{1}{q^n+1} + \frac{\eta_t}{q^n-1}\,.
\end{align}
Combining the above, we have
\begin{align}
    \braket{S^n,\mathbf{1}}{v^{(t)}}\frac{q-1}{1+q^{-n+1}} \geq \frac{1}{q^n+1}\left(\frac{q\left(1+q^{-n}\right)}{1 + q^{-n+1}}-1\right) - \frac{\eta_t}{q^n-1}\,,
\end{align}
and hence
\begin{align}
    \braket{S^n,\mathbf{1}}{v^{(t)}}\geq \frac{1-\eta'_t}{q^n+1}\,,
\end{align}
where 
\begin{equation}
    \eta'_t = \eta_t\frac{(q^n+1)(1+q^{-n+1})}{(q-1)(q^n-1)} \leq 6\eta_t = 6\chi_2e^{-\frac{\chi_1}{n}(t-s_{AC})}\,.
\end{equation}
The inequality above is true for all $n \geq 1$ and $q \geq 2$. We choose $\chi_4 = 6\chi_2$ and $\chi_3 = \chi_1$, and the lemma is proved. 
\end{proof}

\subsubsection{Proof of \autoref{lem:ACconvergenceUnderPI}}

\begin{proof}
The gate at time step $t$ acts on bits $i_{t}$ and $j_t$. Suppose for some configuration $\vec{\nu}$ these bits disagree, i.e.~$\nu_{i_t} \neq \nu_{j_t}$. 
Consider a state $\ket{\vec{\eta},\vec{\eta}'}$ for which $\Delta \ket{\vec{\eta},\vec{\eta}'} = \ket{\vec{\nu}}$. Then consider the quantity
\begin{align}
\bra{\mathbf{q}}\Delta R_{SI}^{(t)}\ket{\vec{\eta},\vec{\eta}'}-1 &=  \bra{\mathbf{q}}P_I^{(t)}\Delta\ket{\vec{\eta},\vec{\eta}'}-1 = \bra{\mathbf{q}}P_I^{(t)}\ket{\vec{\nu}}-1\\
    &= \sum_{\vec{\mu}} (q^{|\vec{\mu}|}-1)\bra{\vec{\mu}}  L_I P^{(t)} L_I^{-1}\ket{\vec{\nu}} \\
    &=\sum_{\vec{\mu}}\frac{(q^{|\vec{\mu}|}-1)(1-q^{-2n+2|\vec{\mu}|})}{1-q^{-2n+2|\vec{\nu}|}} \bra{\vec{\mu}}P^{(t)}\ket{\vec{\nu}}\,.
\end{align}
The action of $P^{(t)}$ on $\ket{\vec{\nu}}$ will force a bit flip, so there are only two possible $\vec{\mu}$ that lead to a non-zero contribution, one for which $|\vec{\mu}|=|\vec{\nu}|+1$ and one for which $|\vec{\mu}|=|\vec{\nu}|-1$. The matrix element (probability) of the former is $1/(q^2+1)$ and the matrix element for the latter is $q^2/(q^2+1)$. Thus, we have
\begin{align}
    \bra{\mathbf{q}}P_I^{(t)}\ket{\vec{\nu}}-1 &= \frac{q^2(q^{|\vec{\nu}|-1}-1)(1-q^{-2n+2|\vec{\nu}|-2})}{(q^2+1)(1-q^{-2n+2|\vec{\nu}|})} + \frac{(q^{|\vec{\nu}|+1}-1)(1-q^{-2n+2|\vec{\nu}|+2})}{(q^2+1)(1-q^{-2n+2|\vec{\nu}|})} \\
    &=\frac{2q}{q^2+1}\frac{q^{|\vec{\nu}|}-\frac{q+q^{-1}}{2} - q^{-2n+2|\vec{\nu}|}\left(q^{|\vec{\nu}|}\frac{q^2+q^{-2}}{2}-\frac{q+q^{-1}}{2}\right)}{1-q^{-2n+2|\vec{\nu}|}}\\
    &\leq \frac{2q}{q^2+1}(q^{|\vec{\nu}|}-1) = \frac{2q}{q^2+1} \left(\braket{\mathbf{q}}{\vec{\nu}}-1\right)\,.
\end{align}
The above is true for all $\vec{\nu}$, and demonstrates that each time disagreeing bits are coupled, the total contribution under inner product with $(\bra{\mathbf{q}}-\bra{\mathbf{1}})\Delta$ decreases by a constant factor. 

Now consider the sequence $\prod_{t=t_0+1}^{t_1}\left(\mathcal{I} \otimes Q'^{(t)}_{\sigma} Q_\sigma^{(t)}\right) R_{SI}^{(t)}$ acting on $\ket{\vec{\eta},\vec{\eta}'}$.  Since the architecture is $h$-regularly connected, for any $t$ there is at least a 1/2 chance that there will be some pair $(i_{t'},j_{t'})$ with $t < t' \leq t+hn$ for which $\nu_{i_{t'}} \neq \nu_{j_{t'}}$ (assuming $\vec{\nu}$ is not a fixed point). The first time this happens, it will lead to a decrease in inner product with $(\bra{\mathbf{q}}-\bra{\mathbf{1}})\Delta$ by the factor $2q/(q^2+1)$. The only way this would not happen is if one of the bits $\nu_{i_{t'}}$ or $\nu_{j_{t'}}$ was flipped already by action by one of the operators $Q^{(t'')}$. However, since the $Q^{(t)}_\sigma$ operators act only on the noisy $Y$ copy, they can only flip a bit of $\vec{\eta}'$ from a 1 to a 0, which would also induce a bit flip in $\vec{\nu}$ from a 1 to a 0. In this case, the Hamming weight decreases by 1 and the inner product with $(\bra{\mathbf{q}}-\bra{\mathbf{1}})\Delta$ would decrease by a factor of $\frac{q^{|\vec{\nu}|-1}-1}{q^{|\vec{\nu}|}-1}$ which is less than $2q/(q^2+1)$.

Thus, if $z_0$ is the largest integer such that $t_0 + z_0 hn \leq t_1$, then 
\begin{align}
    \bra{\mathbf{q}}\Delta \prod_{t=t_0+1}^{t_1}\left((\mathcal{I} \otimes Q'^{(t)}_{\sigma} Q_\sigma^{(t)}) R_{SI}^{(t)}\right) \ket{v}-1 
    \leq{}& \left(\frac{1}{2} + \frac{1}{2}\frac{2q}{q^2+1}\right)^{z_0}\left(\bra{\mathbf{q}}\Delta \ket{v}-1\right) \\
    \leq{}& \left(\frac{1}{2} + \frac{1}{2}\frac{2q}{q^2+1}\right)^{\frac{t_1-t_0}{hn}-1}\left(\bra{\mathbf{q}}\Delta \ket{v}-1\right) \\
    ={}& \chi_6\exp\left(-\frac{\chi_5 (t_1-t_0)}{n}\right) \left(\bra{\mathbf{q}}\Delta \ket{v}-1\right)
\end{align}
for appropriate choice of $\chi_5$ and $\chi_6$.
\end{proof}

\subsubsection{Proof of \autoref{lem:expclustering}}

\begin{proof}
    When probability mass is redirected from $S$-destined at time step $t-1$ to $I$-destined at time step $t'$, it may begin with Hamming weight as large as $n-1$. But since it is $I$-destined, it will quickly move down in Hamming weight. We wish to quantify this phenomenon. First of all, 
    \begin{align}
        \bra{\mathbf{1}}\Pi_{ w}\Delta \ket{v_{SI}^{(t,t')}} &= \sum_{\vec{\mu}:|\vec{\mu}| = w} \bra{\vec{\mu}}\Delta \ket{v_{SI}^{(t,t')}} = \frac{\sum_{\vec{\mu}:|\vec{\mu}| = w}(q^{|\vec{\mu}|}-1) \bra{\vec{\mu}}\Delta \ket{v_{SI}^{(t,t')}}}{q^w-1} \\
        &\leq \frac{\bra{\mathbf{q}}\Delta \ket{v_{SI}^{(t,t')}}-\braket{\mathbf{1}}{v_{SI}^{(t,t')}}}{q^w-1}\,.
    \end{align}
    Now, note that $\ket{v_{SI}^{(t,t')}} = \prod_{t''=t'+1}^{t}\left((\mathcal{I} \otimes Q'^{(t'')}_\sigma Q_\sigma^{(t'')}) R_{SI}^{(t'')}\right) \ket{v_{SI}^{(t',t')}}$, so we can invoke \autoref{lem:ACconvergenceUnderPI}.
    \begin{align}
        \bra{\mathbf{1}}\Pi_{w}\Delta \ket{v_{SI}^{(t,t')}} &\leq \frac{\bra{\mathbf{q}}\Delta \ket{v_{SI}^{(t',t')}}-\braket{\mathbf{1}}{v_{SI}^{(t,t')}}}{q^w-1}\chi_6 \exp\left(-\frac{\chi_5 (t-t')}{n}\right) \\
        &\leq \frac{q^{n}-1}{q^w-1}\braket{\mathbf{1},\mathbf{1}}{v_{SI}^{(t',t')}}\chi_6 \exp\left(-\frac{\chi_5 (t-t')}{n}\right)\,,
    \end{align}
    where the second line follows because $q^n$ is the maximum entry in $\bra{\mathbf{q}}$, and the quantity $\braket{\mathbf{1}}{v_{SI}^{(t,t')}}$ does not change as $t$ increases (it evolves by stochastic transformations). 
    
    We now invoke \autoref{lem:11vSI} (in the first line) and \autoref{lem:decayOfvSS} (in the second line) to say
    \begin{align}
        &\bra{\mathbf{1}}\Pi_{w}\Delta \ket{v_{SI}^{(t,t')}}\leq{} \frac{q^{n}-1}{q^w-1}(2\sigma-\sigma^2)\braket{\mathbf{1},\mathbf{1}}{v_{SS}^{(t'-1)}}\chi_6 e^{-\frac{\chi_5 (t-t')}{n}} \\
        \leq{}& \frac{q^{n}-1}{q^w-1}(2\sigma-\sigma^2)(1-\sigma)^{-2(t-t'+1)}\braket{\mathbf{1},\mathbf{1}}{v_{SS}^{(t)}}\chi_6 e^{-\frac{\chi_5 (t-t')}{n}} \\
        \leq{}& \sigma(4\chi_6 q^{n-w})  \exp\left(-\frac{\chi_5 (t-t')}{n}+ 2(t-t'+1)\log\left(\frac{1}{1-\sigma}\right)  \right)\braket{\mathbf{1},\mathbf{1}}{v_{SS}^{(t)}}\,, \label{eq:boundFilterExpDecay}
    \end{align}
    where the extra factor of $2$ comes from a very crude bound $(q^n-1)/(q^w-1) \leq 2q^{n-w}$.
    As long as $\chi_5/n$ is greater than $2\log(1/(1-\sigma))$, the above is exponentially decaying in $t$. This will be the case whenever $\sigma \leq 1-\exp(-\chi_5/2n))$. There is an $n_0$ and $\chi_7$ such that $\sigma \leq \chi_7/n$ whenever $n \geq n_0$ is a weaker condition. 
    Alternatively, we could make a simpler bound by invoking \autoref{lem:11vSI} and \autoref{lem:decayOfvSS}, but not \autoref{lem:ACconvergenceUnderPI}.
    \begin{align}
        \bra{\mathbf{1}}\Pi_{w}\Delta \ket{v_{SI}^{(t,t')}} &\leq \braket{\mathbf{1},\mathbf{1}} {v_{SI}^{(t,t')}} \leq 2\sigma \braket{\mathbf{1},\mathbf{1}}{v_{SS}^{(t'-1)}} \\
        &\leq 2\sigma (1-\sigma)^{-2(t-t'+1)} \braket{\mathbf{1},\mathbf{1}}{v_{SS}^{(t)}} \label{eq:boundFilterNaive}
    \end{align}
    Both Eq.~\eqref{eq:boundFilterExpDecay} and Eq.~\eqref{eq:boundFilterNaive} will be useful. 
    
    Now, we connect $\ket{v^{(t)}_{SS}}$ to $\ket{v_{SI}^{(t,t')}}$. First we note
    \begin{align}
        \bra{\mathbf{1}}\Pi_{w}\Delta \ket{v^{(t)}_{SS}}
        ={}& \bra{\mathbf{1}}\Pi_{w} \Delta L_{SS}\ket{v^{(t)}}  
        ={}\sum_{\vec{\mu},\vec{\nu}} \bra{\mathbf{1}} \Pi_{w} \Delta  \ketbra{\vec{\mu},\vec{\nu}}L_{SS}\ket{v^{(t)}} \\
        ={}& \sum_{\substack{\vec{\mu},\vec{\nu} \\|\vec{\mu}| = |\vec{\nu}| +n-w }}\bra{\vec{\mu},\vec{\nu}}L_{SS}\ket{v^{(t)}} ={} \sum_{\substack{\vec{\mu},\vec{\nu} \\|\vec{\mu}| = |\vec{\nu}|+n-w }}\frac{q^{-2n+2|\vec{\nu}|}-q^{-2n}}{1-q^{-2n}}\braket{\vec{\mu},\vec{\nu}}{v^{(t)}} \\
        ={}& \sum_{\substack{\vec{\mu},\vec{\nu} \\|\vec{\mu}| = |\vec{\nu}|+n-w }}\frac{q^{2|\vec{\nu}|}-1}{q^{2|\vec{\mu}|}-q^{2|\vec{\nu}|}}\frac{q^{-2n+2|\vec{\mu}|}-q^{-2n+2|\vec{\nu}|}}{1-q^{-2n}}\braket{\vec{\mu},\vec{\nu}}{v^{(t)}} \\
        ={}& \sum_{\substack{\vec{\mu},\vec{\nu} \\|\vec{\mu}| = |\vec{\nu}|+n-w }}q^{-2(n-w)}\frac{1-q^{-2|\vec{\nu}|}}{1-q^{-2(n-w)}}\bra{\vec{\mu},\vec{\nu}}L_{SI}\ket{v^{(t)}} \\
        \leq {}&\frac{q^{-2(n-w)}}{1-q^{-2}}\sum_{\substack{\vec{\mu},\vec{\nu} \\|\vec{\mu}| =  |\vec{\nu}|+n-w }}\braket{\vec{\mu},\vec{\nu}}{v_{SI}^{(t)}} = \frac{q^{-2(n-w)}}{1-q^{-2}}\bra{\mathbf{1}}\Pi_w \Delta \ket{v_{SI}^{(t)}}\,.
        \end{align}
    This allows us to use Eq.~\eqref{eq:sumSI} and assert
    \begin{equation}\label{eq:1piwdeltaIntermediate}
        \bra{\mathbf{1}}\Pi_w\Delta \ket{v^{(t)}_{SS}} = \frac{q^{-2(n-w)}}{1-q^{-2}}\sum_{t' = 1}^{t}\bra{\mathbf{1}}\Pi_w \Delta \ket{v_{SI}^{(t,t')}}\,.
    \end{equation}
    
    Let $t_w = t-\lceil n(n-w)\log(q)/\chi_5 \rceil$. For $t' > t_w$, we will bound $\ket{v_{SI}^{(t,t')}}$ with Eq.~\eqref{eq:boundFilterNaive}, and for $t' \leq t_w$, we will use Eq.~\eqref{eq:boundFilterExpDecay}. Let us examine these sums separately. For the $t' > t_w$ portion, we make the substitution $a=t'-t_w-1$, and we have 
    \begin{align}
\sum_{t'=t_w+1}^{t}\bra{\mathbf{1}}\Pi_{w}\Delta \ket{v_{SI}^{(t,t')}} \leq{}& \sum_{t'=t_w+1}^{t} 2\sigma (1-\sigma)^{-2(t-t'+1)} \braket{\mathbf{1},\mathbf{1}}{v_{SS}^{(t)}} \\
={}& \braket{\mathbf{1},\mathbf{1}}{v_{SS}^{(t)}}2\sigma(1-\sigma)^{-2(t-t_w)}\sum_{a=0}^{t-t_w-1}  (1-\sigma)^{2a} \\
={}&\braket{\mathbf{1},\mathbf{1}}{v_{SS}^{(t)}}2\sigma(1-\sigma)^{-2(t-t_w)}\frac{1-(1-\sigma)^{2(t-t_w-1)}}{2\sigma-\sigma^2} \\
\leq{}&\braket{\mathbf{1},\mathbf{1}}{v_{SS}^{(t)}}(1-\sigma)^{-2(t-t_w)}\left(4\sigma(t-t_w)\right) \\
\leq{}&\braket{\mathbf{1},\mathbf{1}}{v_{SS}^{(t)}}(1-\sigma)^{-2\lceil n(n-w)\log(q)/\chi_5 \rceil}\left(4\sigma\lceil n(n-w)\log(q)/\chi_5 \rceil\right) \\
\leq{}& \braket{\mathbf{1},\mathbf{1}}{v_{SS}^{(t)}}q^{-2n(n-w)\log(1-\sigma)/\chi_5}\chi_5'n\sigma(n-w)\, \label{eq:sumabovetwresult}
\end{align}
for some constant $\chi_5'$ slightly larger than $4\log(q)/\chi_5$ to account for dropping the ceiling in the last line. Note that in the third-to-last line, the extra factor of 2 comes from the bound $2\sigma /(2\sigma - \sigma^2) \leq 2$. 

 For the $t \leq t_w$ portion, we use the substitution $a = t_w-t'$ and find (assuming $\chi_5/n \geq 2\log(1/(1-\sigma))$)
\begin{align}
    \sum_{t'=1}^{t_w}\bra{\mathbf{1}}\Pi_{w}\Delta \ket{v_{SI}^{(t,t')}} \leq{}& \sum_{t'=1}^{t_w} \sigma(4\chi_6 q^{n-w})  e^{-\frac{\chi_5 (t-t')}{n}+ 2(t-t'+1)\log\left(\frac{1}{1-\sigma}\right)  }\braket{\mathbf{1},\mathbf{1}}{v_{SS}^{(t)}} \\
    ={}& \braket{\mathbf{1},\mathbf{1}}{v_{SS}^{(t)}}\sigma(4\chi_6 q^{n-w})\sum_{a=0}^{t_w-1}e^{-\frac{\chi_5 (t-t_w+a)}{n}+ 2(t-t_w+a+1)\log\left(\frac{1}{1-\sigma}\right)  } \\
    \leq{}& \braket{\mathbf{1},\mathbf{1}}{v_{SS}^{(t)}}\sigma(4\chi_6 q^{n-w})\sum_{a=0}^{\infty}e^{-\frac{\chi_5 (t-t_w+a)}{n}+ 2(t-t_w+a+1)\log\left(\frac{1}{1-\sigma}\right)  } \\
    ={}& \braket{\mathbf{1},\mathbf{1}}{v_{SS}^{(t)}}\sigma(4\chi_6 q^{n-w})\frac{\exp\left(-\lceil n(n-w)\log(q)/\chi_5 \rceil (\frac{\chi_5}{n}+2\log(1-\sigma))\right)}{ 
    (1-e^{-\chi_5/n-2\log(1-\sigma)})(1-\sigma)^2} \\
    \leq{}& \braket{\mathbf{1},\mathbf{1}}{v_{SS}^{(t)}}\sigma(4\chi_6)\frac{\exp\left(-2\lceil n(n-w)\log(q)/\chi_5 \rceil \log(1-\sigma)\right)}{ 
    (1-e^{-\chi_5/n-2\log(1-\sigma)})(1-\sigma)^2} \\
    \leq{}& \braket{\mathbf{1},\mathbf{1}}{v_{SS}^{(t)}}\sigma\chi_6'q^{-2n(n-w)\log(1-\sigma)/\chi_5} \label{eq:sumbelowtwresult}
\end{align}
for some constant $\chi_6'$. Plugging the bounds on the two parts of the sum into Eq.~\eqref{eq:1piwdeltaIntermediate}, we find
\begin{align}
    \frac{\bra{\mathbf{1}}\Pi_w\Delta \ket{v^{(t)}_{SS}}}{\braket{\mathbf{1},\mathbf{1}}{v_{SS}^{(t)}}}&\leq \frac{q^{-2(n-w)}}{1-q^{-2}}n\sigma q^{-2n(n-w)\log(1-\sigma)/\chi_5}\left(\chi_5'(n-w)+\frac{\chi_6'}{n}\right)\\
    &\leq \chi_9 q^{-2(n-w)}n\sigma (n-w) q^{c'(n-w)} \\
    &= \chi_9n\sigma(n-w) q^{-(n-w)}q^{-(1-c')(n-w)}
\end{align}
for some constants $\chi_9$ and $c'$ which is less than 1 whenever $\sigma \leq \chi_7/n$ and $n \geq n_0$ hold. Thus we may define $\chi_8=(1-c')\log(q)$ and the lemma is proved. 
\end{proof}

\subsubsection{Proof of \autoref{lem:Zsigmaupperlowerbound}}

\begin{proof}
Recall that $\mathcal{Z}_\sigma = \braket{\mathbf{1},\mathbf{q}}{v^{(s)}}$. 
and that $\ket{v^{(t)}_{SS}} = L_{SS} \ket{v^{(t)}}$. The matrix $L_{SS}^{-1}$ is defined to be the Moore-Penrose pseudo-inverse of $L_{SS}$, and note that the null space of $L_{SS}$ is the space spanned by $\ket{\vec{\nu},I^n}$ for all $\vec{\nu}$. The projector onto this subspace is $\ketbra{\mathbf{1},I^n}$. Thus,
\begin{align}
    \ket{v^{(s)}} &= \mathcal{I} \ket{v^{(s)}} = (\ketbra{\mathbf{1},I^n} + L_{SS}^{-1}L_{SS}) \ket{v^{(s)}}\\
    &= \ket{\mathbf{1},I^n}\braket{\mathbf{1},I^n}{v^{(s)}}+ L_{SS}^{-1}\ket{v_{SS}^{(s)}}  \,.
\end{align}
The lower bound is shown as follows:
\begin{align}
    \mathcal{Z}_\sigma-1 &= \sum_{\vec{\nu}}\left(q^{|\vec{\nu}|}-1\right)\braket{\mathbf{1},\vec{\nu}}{v^{(s)}} =\sum_{\vec{\nu} \neq I^n}\left(q^{|\vec{\nu}|}-1\right)\braket{\mathbf{1},\vec{\nu}}{v^{(s)}}\\
    &=\sum_{\vec{\nu} \neq I^n}
    \left(q^{|\vec{\nu}|}-1\right)\bra{\mathbf{1},\vec{\nu}}\left(\ket{\mathbf{1},I^n}\braket{\mathbf{1},I^n}{v^{(s)}}+ L_{SS}^{-1}\ket{v_{SS}^{(s)}}  \right)\\
    &=\sum_{\vec{\nu} \neq I^n}
    \left(q^{|\vec{\nu}|}-1\right)\bra{\mathbf{1},\vec{\nu}} L_{SS}^{-1}\ket{v_{SS}^{(s)}}\\
    &=\sum_{\vec{\nu} \neq I^n}
    \left(q^{|\vec{\nu}|}-1\right)
    \bra{\vec{\nu}}\frac{1-q^{-2n}}{q^{-2n+2|\vec{\nu}|}-q^{-2n}}\ket{v_S^{(s)}}\\
    &=\sum_{\vec{\nu} \neq I^n}\left(q^n-1\right)
    \left(\frac{1+q^n}{1+q^{|\vec{\nu}|}}\right)\braket{\mathbf{1},\vec{\nu}}{v_{SS}^{(s)}}\\
    &\geq \sum_{\vec{\nu} \neq I^n}\left(q^n-1\right)
    \braket{\mathbf{1},\vec{\nu}}{v_{SS}^{(s)}}\\
    &= (q^n-1) \braket{\mathbf{1},\mathbf{1}}{v_{SS}^{(s)}}\,.
\end{align}
Now, we will show the upper bound.

\begin{align}
    \mathcal{Z}_\sigma-1 &= \sum_{\vec{\nu}}\left(q^{|\vec{\nu}|}-1\right)\braket{\mathbf{1},\vec{\nu}}{v^{(s)}} \\
    &= \sum_{\vec{\nu}}\left(q^{|\vec{\nu}|}-1\right)
    \left(\bra{S^n,\vec{\nu}} + \sum_{\vec{\mu} \neq S^n}\bra{\vec{\mu},\vec{\nu}}\right)
    \ket{v^{(s)}}\\
    &= \sum_{\vec{\nu}}\left(\left(q^{|\vec{\nu}|}-1\right)\braket{S^n,\vec{\nu}}{v^{(s)}}+ 
    \sum_{\vec{\mu} \neq S^n}\left(q^{|\vec{\nu}|}-1\right)\braket{\vec{\mu},\vec{\nu}}{v^{(s)}}\right)\\
    &\leq \sum_{\vec{\nu}}\left(\left(q^{|\vec{\nu}|}-1\right)\braket{S^n,\vec{\nu}}{v^{(s)}}+ 
    \sum_{\vec{\mu} \neq S^n}\left(q^{|\vec{\mu}|}-1\right)\braket{\vec{\mu},\vec{\nu}}{v^{(s)}}\right)\\
    &= \sum_{\vec{\nu}}\left(\left(q^{|\vec{\nu}|}-1\right)\braket{S^n,\vec{\nu}}{v^{(s)}}\right)+ Z_0-1 - (q^n-1)\braket{S^n,\mathbf{1}}{v^{(s)}}\,,
\end{align}
where we have used $Z_0 = \sum_{\vec{\nu}}\sum_{\vec{\mu}}q^{|\vec{\mu}|}\braket{\vec{\mu},\vec{\nu}}{v^{(s)}}$. Now we invoke \autoref{lem:ACconvergencetoSn}, to say
\begin{align}
    \mathcal{Z}_\sigma-1 &\leq
    \sum_{\vec{\nu}}\left(q^{|\vec{\nu}|}-1\right)\braket{S^n,\vec{\nu}}{v^{(s)}}
    + Z_0-1 - \frac{q^n-1}{q^n+1}(1-\eta'_s) \\
    &= \sum_{\vec{\nu}}
    \left(q^{|\vec{\nu}|}-1\right)\braket{S^n,\vec{\nu}}{v^{(s)}}
    + \left(Z_0 - \frac{2q^n}{q^n+1}\right) + \eta'_s\left(\frac{q^n-1}{q^n+1}\right) \\
    &\leq \sum_{\vec{\nu}}\left(q^{|\vec{\nu}|}-1\right)\braket{S^n,\vec{\nu}}{v^{(s)}}+ \left(Z_0 - \frac{2q^n}{q^n+1}\right) + \eta'_s \,.
\end{align}
Now we invoke \autoref{lem:anticoncentrationInNoisySection} to bound $Z_0 - 2q^{n}/(q^n+1)$ in the first step below, and continue on. Denote $\eta_s'' = \eta_s + \eta'_s$.
\allowdisplaybreaks
\begin{align}
    \mathcal{Z}_\sigma-1
    \leq{}& \sum_{\vec{\nu}}\left(q^{|\vec{\nu}|}-1\right)\braket{S^n,\vec{\nu}}{v^{(s)}}+ \eta'_s + \eta_s \\
    ={}& \sum_{\vec{\nu}}\left(q^{|\vec{\nu}|}-1\right)\bra{S^n,\vec{\nu}}L_{SS}^{-1}L_{SS}\ket{v^{(s)}}+ \eta_s'' \\
    ={}& \sum_{\vec{\nu}}\left(\frac{(q^{|\vec{\nu}|}-1)(1-q^{-2n})}{q^{-2n+2|\vec{\nu}|}-q^{-2n}}\right)\braket{S^n,\vec{\nu}}{v_{SS}^{(s)}}+ \eta_s''  \\
    ={}&\sum_{\vec{\nu}}(q^n-1)\left(\frac{q^{n}+1}{q^{|\vec{\nu}|}+1}\right)\braket{S^n,\vec{\nu}}{v_{SS}^{(s)}}+ \eta_s'' \\
    \leq{}& \eta_s'' +(q^n-1)\sum_{\vec{\nu}}q^{n-|\vec{\nu}|}\braket{S^n,\vec{\nu}}{v_{SS}^{(s)}}\\
    ={}&\eta_s'' +(q^n-1)\sum_{\vec{\nu}}q^{n-|\vec{\nu}|}\bra{\vec{\nu}}\Delta\ket{v_{SS}^{(s)}}\\
     ={}&\eta_s'' +(q^n-1)\braket{S^n,S^n}{v_{SS}^{(s)}}+(q^n-1)\sum_{w=1}^{n-1}q^{n-w}\bra{\mathbf{1}}\Pi_w\Delta\ket{v_{SS}^{(s)}}\\
     \leq{}& \eta_s'' +(q^n-1)\braket{S^n,S^n}{v_{SS}^{(s)}}+(q^n-1)\sum_{w=1}^{n-1}q^{n-w}n\sigma \xi_w\braket{\mathbf{1},\mathbf{1}}{v_{SS}^{(s)}}\\
     %
     %
     %
     \leq{}& \eta_s'' +(q^n-1)\braket{\mathbf{1},\mathbf{1}}{v_{SS}^{(s)}}\left(1+\chi_9 n \sigma \sum_{w=1}^{n-1}(n-w) e^{-\chi_8(n-w)}\right)\,,
\end{align}
\allowdisplaybreaks[0]
where in the second-to-last line we have invoked \autoref{lem:expclustering}, which requires $\sigma \leq \chi_7/n$ and $n \geq n_0$ (leading to our requirements in this lemma that $\sigma \leq \chi_{13}/n$ and $n \geq n_0$). Now, we make the choice of $\chi_{10}=\chi_9\sum_{w=1}^{n-1}(n-w)e^{-\chi_8 (n-w)}\leq \chi_9\sum_{w=1}^{\infty}w e^{-\chi_8 w} = O(1)$, which yields the following. (In line 2, we invoke \autoref{lem:decayOfvSS}.)
\begin{align}
    \mathcal{Z}_\sigma-1 &\leq (q^n-1)\braket{\mathbf{1},\mathbf{1}}{v_{SS}^{(s)}}\left(1+\chi_{10} n\sigma+\frac{\eta''_s}{(q^n-1)\braket{\mathbf{1},\mathbf{1}}{v_{SS}^{(s)}}}\right)  \\
    &\leq (q^n-1)\braket{\mathbf{1},\mathbf{1}}{v_{SS}^{(s)}}\left(1+\chi_{10} n\sigma+\frac{q^n+1}{q^n-1}\eta''_s(1-\sigma)^{-2s}\right) \\
    &\leq (q^n-1)\braket{\mathbf{1},\mathbf{1}}{v_{SS}^{(s)}}\left(1+\chi_{10} n\sigma+3\eta''_s(1-\sigma)^{-2s}\right) \\
    &\leq (q^n-1)\braket{\mathbf{1},\mathbf{1}}{v_{SS}^{(s)}}\left(1+\chi_{10} n\sigma+\chi_{12} e^{-\frac{\chi_{11}}{n}(s-s_{AC})+4s\sigma}\right)\\
    &\leq (q^n-1)\braket{\mathbf{1},\mathbf{1}}{v_{SS}^{(s)}}\exp\left(1+\chi_{10} n\sigma+\chi_{12}  e^{-\frac{\chi_{11}}{n}(s-s_{AC})+4s\sigma}\right)\,,
\end{align}
where the third-to-last line is true for all $q\geq 2$ and $n \geq 1$, and the second-to-last line plugs in the equations for $\eta_s$ and $\eta'_s$, chooses constants $\chi_{11}$ and $\chi_{12}$ appropriately, and asserts $(1-\sigma)^{2s} \leq e^{-4\sigma s}$, which is true whenever $\sigma \leq 0.79$, so it is certainly true under the assumption $\sigma \leq \chi_7/n$ for sufficiently large $n$. 
\end{proof}

\subsubsection{Proof of \autoref{lem:boundSdestinedmassLAYERS}}

\begin{proof}
Recall that $\braket{\mathbf{1},\mathbf{1}}{v_{SS}^{(0)}} = 1/(q^n+1)$. Let $t_0 = dn/2$. 
\allowdisplaybreaks
    \begin{align}
        \braket{\mathbf{1}, \mathbf{1}}{v_{SS}^{(t_0+n/2)}}
        ={}&\bra{\mathbf{1}, \mathbf{1}}L_{SS}\prod_{t=t_0+1}^{t_0+n/2}R_\sigma^{(t)}\ket{v^{(t_0)}}\\
        ={}& \bra{\mathbf{1}, \mathbf{1}}L_{SS}\prod_{t=t_0+1}^{t_0+n/2}R^{(t)}_\sigma L_{SS}^{-1}\ket{v_{SS}^{(t_0)}}\\
        ={}& \bra{\mathbf{1}, \mathbf{1}}L_{SS}\prod_{t=t_0+1}^{t_0+n/2}(\mathcal{I} \otimes Q'^{(t)}_{\sigma} Q_\sigma^{(t)})\prod_{t=t_0+1}^{t_0+n/2}(\mathcal{I} \otimes R_0^{(t)}) L_{SS}^{-1}\ket{v_{SS}^{(t_0)}}\\
        ={}&  \bra{\mathbf{1}, \mathbf{1}}L_{SS}\prod_{t=t_0+1}^{t_0+n/2}(\mathcal{I} \otimes Q'^{(t)}_{\sigma} Q_\sigma^{(t)})L_{SS}^{-1}\prod_{t=t_0+1}^{t_0+n/2} R_{SS}^{(t)}\ket{v_{SS}^{(t_0)}}\\
        ={}&  \bra{\mathbf{1}, \mathbf{1}}(\mathcal{I} \otimes L_S)\prod_{t=t_0+1}^{t_0+n/2}(\mathcal{I} \otimes Q'^{(t)}_{\sigma} Q_\sigma^{(t)})(\mathcal{I} \otimes L_{S}^{-1})\prod_{t=t_0+1}^{t_0+n/2} R_{SS}^{(t)}\ket{v_{SS}^{(t_0)}}\\
        ={}&\bra{\mathbf{1}} L_S\prod_{t=t_0+1}^{t_0+n/2}( Q'^{(t)}_{\sigma} Q_\sigma^{(t)}) L_{S}^{-1}(\bra{\mathbf{1}}\otimes \mathcal{I})\prod_{t=t_0+1}^{t_0+n/2}R_{SS}^{(t)}\ket{v_{SS}^{(t_0)}}\\
        ={}&\sum_{\vec{\nu}\neq I^n}\bra{\mathbf{1}} L_S\prod_{t=t_0+1}^{t_0+n/2}( Q'^{(t)}_{\sigma} Q_\sigma^{(t)}) L_{S}^{-1}\ket{\vec{\nu}}\bra{\mathbf{1},\vec{\nu}}\prod_{t=t_0+1}^{t_0+n/2}R_{SS}^{(t)}\ket{v_{SS}^{(t_0)}}\,. \label{eq:SSprobIntermediate}
        \end{align}
        \allowdisplaybreaks[0]
    We now examine the quantity
    \begin{align}
        \bra{\mathbf{1}} L_S\prod_{t=t_0+1}^{t_0+n/2}( Q'^{(t)}_{\sigma} Q_\sigma^{(t)}) L_{S}^{-1} \ket{\vec{\nu}} ={}&\sum_{\vec{\mu}} \bra{\vec{\mu}}L_S\prod_{t=t_0+1}^{t_0+n/2}( Q'^{(t)}_{\sigma} Q_\sigma^{(t)}) L_{S}^{-1} \ket{\vec{\nu}} \\
        ={}& \sum_{\vec{\mu}} 
        \frac{q^{-2n+2|\vec{\mu}|}-q^{-2n}}{q^{-2n+2|\vec{\nu}|}-q^{-2n}}\bra{\vec{\mu}}
        \prod_{t=t_0+1}^{t_0+n/2}( Q'^{(t)}_{\sigma} Q_\sigma^{(t)})  \ket{\vec{\nu}}\,.
    \end{align}
    Note that, because of the layered property, all $n$ qudits are acted upon by one of the $Q_\sigma^{(t)}$ or $Q'^{(t)}_{\sigma} $. This can cause some $S$ bits to flip to $I$ bits (with probability $\sigma$). For a configuration $\vec{\mu}$ to have non-zero contribution in the above sum, it must have $\mu_i \leq \nu_i$ for all $i$ (under the ordering $I < S$), a condition we denote by $\vec{\mu} \leq \vec{\nu}$, and in this case we have
    \begin{equation}
        \bra{\vec{\mu}}
        \prod_{t=t_0+1}^{t_0+n/2}( Q'^{(t)}_{\sigma} Q_\sigma^{(t)})  \ket{\vec{\nu}} = (1-\sigma)^{|\vec{\mu}|}\sigma^{|\vec{\nu}| - |\vec{\mu}|}\,.
    \end{equation}
    Note also the following sum formula, which holds for any real number $z$.
    \begin{equation}
        \sum_{\vec{\mu}\leq \vec{\nu}}q^{z|\vec{\mu}|} (1-\sigma)^{|\vec{\mu}|}\sigma^{|\vec{\nu}| - |\vec{\mu}|} = \sum_{x=0}^{|\vec{\nu}|} \binom{|\vec{\nu}|}{x}q^{zx}(1-\sigma)^{x}\sigma^{|\vec{\nu}| - x} = (\sigma+q^z(1-\sigma))^{|\vec{\nu}|}\,.
    \end{equation}
    We find
    \begin{align}
        \bra{\mathbf{1}} L_S\prod_{t=t_0+1}^{t_0+n/2}( Q'^{(t)}_{\sigma} Q_\sigma^{(t)}) L_{S}^{-1} \ket{\vec{\nu}}
        ={}& \frac{1}{q^{2|\vec{\nu}|}-1}\sum_{\vec{\mu}\leq \vec{\nu}} 
        (q^{2|\vec{\mu}|}-1)(1-\sigma)^{|\vec{\mu}|}\sigma^{|\vec{\nu}| - |\vec{\mu}|} \\
        ={}&\frac{(\sigma + q^2(1-\sigma))^{|\vec{\nu}|}-1}{q^{2|\vec{\nu}|}-1} = \frac{(1-\sigma')^{|\vec{\nu}|}-q^{-2|\vec{\nu}|}}{1-q^{-2|\vec{\nu}|}} \,,
    \end{align}
    where $\sigma' = \sigma(1-q^{-2})$. 
     Denote this final expression by
\begin{equation}
    E_w = \frac{(1-\sigma')^{w}-q^{-2w}}{1-q^{-2w}}\,,
\end{equation}
which allows us to rewrite Eq.~\eqref{eq:SSprobIntermediate} as
\begin{equation}
    \braket{\mathbf{1}, \mathbf{1}}{v_{SS}^{(t_0+n/2)}} = \sum_{\vec{\nu}\neq I^n}E_{|\vec{\nu}|}\bra{\mathbf{1},\vec{\nu}}\prod_{t=t_0+1}^{t_0+n/2}R_{SS}^{(t)}\ket{v_{SS}^{(t_0)}}\label{eq:SSprobIntermediate2}\,.
\end{equation}
    
    Now we claim that, for any $|\vec{\nu}| \neq 0$,
    \begin{equation}
        E_n \leq E_{|\vec{\nu}|} \,.
    \end{equation}
    We can prove the statement above by noting that it holds for $|\vec{\nu}| = n$ and observing that the derivative with respect to $|\vec{\nu}|$ is always negative (in this verification, note that $(1-\sigma') \geq 1/q$ holds for all $ \sigma \leq 1$). 
    
    Collecting these observations, we have
    \begin{align}
    \braket{\mathbf{1}, \mathbf{1}}{v_{SS}^{(t_0+n/2)}} &\geq \sum_{\vec{\nu} \neq I^n}E_n \bra{\mathbf{1},\vec{\nu}}\prod_{t=t_0+1}^{t_0+n/2}R_{SS}^{(t)}\ket{v_{SS}^{(t_0)}}\\
    &=E_n \bra{\mathbf{1},\mathbf{1}}\prod_{t=t_0+1}^{t_0+n/2}R_{SS}^{(t)}\ket{v_{SS}^{(t_0)}}\\
    &=E_n\braket{\mathbf{1},\mathbf{1}}{v_{SS}^{(t_0)}}\,.
    \end{align}
Hence, the lower bound in the lemma statement follows by recursively applying the above conclusion for increasing $d$. 

To show the upper bound, we return to Eq.~\eqref{eq:SSprobIntermediate2}.
Note that $E_w \leq 1$. We can restate what we know and divide the mass into whether or not the noiseless copy has reached the $S^n$ fixed point, and if it has, what value $w$ for the Hamming weight the noisy copy ends up at. 
    \begin{align}
        \braket{\mathbf{1}, \mathbf{1}}{v_{SS}^{(t_0+n/2)}} 
        &= A_{not} + \sum_{w=1}^n E_w A_w\,, \label{eq:intermediateAwDecomp}
    \end{align}
where 
\begin{align}
    A_{not} &= \sum_{\vec{\nu}\neq I^n,\vec{\mu} \neq S^n}E_{|\vec{\nu}|}\bra{\mathbf{1},\vec{\nu}}\prod_{t=t_0+1}^{t_0+n/2}R_{SS}^{(t)}\left(\ketbra{\vec{\mu}}\otimes \mathcal{I}\right)\ket{v_{SS}^{(t_0)}} \\
    A_{w} &= \sum_{\vec{\nu}:|\vec{\nu}|=w}\bra{\mathbf{1},\vec{\nu}}\prod_{t=t_0+1}^{t_0+n/2}R_{SS}^{(t)}\left(\ketbra{S^n}\otimes \mathcal{I}\right)\ket{v_{SS}^{(t_0)}}\,.
\end{align}
Since $E_{|\vec{\nu}|} \leq 1$, we may directly apply \autoref{lem:ACconvergencetoSn} and bound $A_{not} \leq \eta_{t_0}'/(q^n+1)$. 

To bound $A_w$, we will need to use \autoref{lem:expclustering}. 
Applying the layer of $R_{SS}^{(t)}$ from $t=t_0+1$ to $t=t_0+n/2$ can at most double the number of $I$-assigned bits, since each qudit participates in at most one gate. So, in order to land at a configuration with Hamming weight $w$ at time step $t_0+n/2$, the configuration at time step $t_0$ must have Hamming weight at most $\lfloor \frac{n+w}{2}\rfloor$. In other words,
\begin{align}
    A_w \leq \sum_{w'=1}^{\lfloor \frac{n+w}{2}\rfloor} \sum_{\vec{\mu}:|\vec{\mu}| = w'}\braket{S^n,\vec{\mu}}{v_{SS}^{(t_0)}}\,.
\end{align}
When $w < n$, the right-hand side of the above is then bounded with \autoref{lem:expclustering}, which requires $\sigma \leq \chi_7/n$ and $n \geq n_0$ (and thus the upper bound portion of lemma inherits these requirements).
\begin{align}
    %
    A_w &\leq \sum_{w'=1}^{\lfloor \frac{n+w}{2}\rfloor} \sum_{\vec{\mu}:|\vec{\mu}| = w'}\bra{\vec{\mu}}\Delta \ket{v_{SS}^{(t_0)}} = \sum_{w'=1}^{\lfloor \frac{n+w}{2}\rfloor} \bra{\mathbf{1}}\Pi_{w'}\Delta \ket{v_{SS}^{(t_0)}} \\
    &\leq \sum_{w'=1}^{\lfloor \frac{n+w}{2}\rfloor} n \sigma \chi_9 (n-w')q^{-(n-w')}e^{-\chi_8(n-w')}\braket{\mathbf{1},\mathbf{1}}{v_{SS}^{(t_0)}} \\
    &\leq  n\sigma\chi_9 \braket{\mathbf{1},\mathbf{1}}{v_{SS}^{(t_0)}} \sum_{a=\lceil \frac{n-w}{2}\rceil}^\infty a e^{-a(\chi_8+\log(q))}\,,
\end{align}
where we have used the substitution $a=n-w'$.
For any $c$, there is a constant $c''$ such that $\sum_{a=a_0}^{\infty} ae^{-ca}$ is bounded by $c''e^{-ca_0}$. Thus, there is a constant $c''$ such that 
\begin{align}
    A_w &\leq \braket{\mathbf{1},\mathbf{1}}{v_{SS}^{(t_0)}} n\sigma\chi_9c''e^{-(\chi_8+\log(q))\lceil \frac{n-w}{2}\rceil} \leq  \braket{\mathbf{1},\mathbf{1}}{v_{SS}^{(t_0)}}n\sigma\chi_9c''e^{-(\chi_8+\log(q)) \frac{n-w}{2}} \\
    &= \braket{\mathbf{1},\mathbf{1}}{v_{SS}^{(t_0)}}f n \sigma e^{-f'(n-w)}\,,
\end{align}
with the definitions $f = \chi_9c'' = O(1)$ and $f' = \chi_8+\log(q) = O(1)$. 
Note also that by construction $\sum_{w=1}^n A_w \leq \braket{\mathbf{1},\mathbf{1}}{v_{SS}^{(t_0)}}$. Thus,
\begin{equation}
    \sum_{w=1}^n E_w A_w = \sum_{w=1}^n (E_n+E_w-E_n) A_w  \leq  E_n\braket{\mathbf{1},\mathbf{1}}{v_{SS}^{(t_0)}} + \sum_{w=1}^{n-1} (E_w-E_n)A_w \,,
\end{equation}
which we can insert into Eq.~\eqref{eq:intermediateAwDecomp}, along with the bounds on $A_w$, giving
\begin{align}
     &\braket{\mathbf{1}, \mathbf{1}}{v_{SS}^{(t_0+n/2)}} \leq \braket{\mathbf{1},\mathbf{1}}{v_{SS}^{(t_0)}}\left(E_n+n\sigma\sum_{w=1}^{n-1}(E_{w}-E_n)fe^{-f' (n-w)}\right)+\frac{\eta_{t_0}'}{q^n+1}\label{eq:eqInPars}
\end{align}
We also have
\begin{equation}
    \frac{E_w}{E_n} = \frac{1-q^{-2n}}{1-q^{-2w}}\frac{(1-\sigma')^{w}-q^{-2w}}{(1-\sigma')^{n}-q^{-2n}} \leq (1-\sigma')^{-(n-w)}\,,
\end{equation}
which can be verified by observing that the quantity 
\begin{equation}
    \frac{E_w}{E_n}(1-\sigma')^{n-w}=\frac{(1-q^{-2n})(1-(q\sqrt{1-\sigma'})^{-2w})}{(1-q^{-2w})(1-(q\sqrt{1-\sigma'})^{-2n})}
\end{equation} 
achieves its maximum with respect to $\sigma'$ when $\sigma'=0$, where it equals 1. The quantity in parentheses in Eq.~\eqref{eq:eqInPars} is now at most
\begin{align}
     \left(E_n+n\sigma E_n\sum_{w=1}^{n-1}fe^{-f'(n-w)}(e^{-\log(1-\sigma') (n-w)}-1)\right)
     \leq{}&\left(E_n+n\sigma E_n\sum_{w=1}^{n-1}fe^{-f'(n-w)}\tau\sigma(n-w)\right)\\
     \leq {}& E_n\left(1+f''n\sigma^2 \right)\,,
\end{align}
where in the first line, we bound $e^{-x\log(1-\sigma)}-1$ by $\tau \sigma x$ for some constant $\tau$, which holds for $x$ sufficiently small, as is the case when $\sigma \leq O(1/n)$ with $n$ sufficiently large; in the second line, we choose the appropriate constant $f''$ as a bound for the sum $f\tau\sum_{a=1}^{n-1} ae^{-a}$. 
This gives us the recursion relation
\begin{equation}\label{eq:layeredUBrecursionrelationA}
\braket{\mathbf{1}, \mathbf{1}}{v_{SS}^{(t_0+n/2)}} \leq \braket{\mathbf{1},\mathbf{1}}{v_{SS}^{(t_0)}}E_n(1+f''n\sigma^2)+\frac{\eta_{t_0}'}{q^n+1}\,.
\end{equation}

For the first few layers, before anti-concentration has been reached and $\eta_{t_0}'$ has become small, we will just use the simpler naive bound $\braket{\mathbf{1}, \mathbf{1}}{v_{SS}^{(t_0+n/2)}} \leq \braket{\mathbf{1},\mathbf{1}}{v_{SS}^{(t_0)}}$. Define the anti-concentration depth as $d_{AC} = 2s_{AC}/n$. Then we have
\begin{align}
    \frac{\eta_{dn/2}'}{q^n+1} &\leq \frac{\chi_4}{q^n+1} e^{-\chi_3(d-d_{AC})/2}\leq \frac{\chi_4'}{q^n+1} E_n e^{-\chi_3(d-d_{AC})/2-n\log(1-\sigma)} \\
    &\leq \chi_4' E_n \braket{\mathbf{1},\mathbf{1}}{v_{SS}^{(dn/2)}} e^{-\chi_3(d-d_{AC})/2-n\log(1-\sigma)-dn\log(1-\sigma)} \\
    &\leq E_n \braket{\mathbf{1},\mathbf{1}}{v_{SS}^{(dn/2)}} n \sigma e^{-\chi_3'(d-d^*)}\,,
\end{align}
where in line 1,  we refer back to the definition of $E_n$ and choose $\chi_4'$ slightly larger than $\chi_4$, in line 2, we use \autoref{lem:decayOfvSS}, and in line 3 we choose 
\begin{equation}\label{eq:dstar}
    d^* = d_{AC}\chi_3/2\chi'_3 + f''' + \log(1/n\sigma)/\chi_3'
\end{equation} 
for some constant $f'''$ that is $O(1)$ whenever $-n\log(1-\sigma)$ is $O(1)$. Note that this also requires $n\log(1-\sigma) \leq \chi_3$. We can choose the constant $a_3$ such that the condition $\sigma \leq a_3/n$ implies these requirements hold.  Note we also must choose a weaker exponential decay constant $\chi_3'$.  Thus our recursion relation is
\begin{equation}\label{eq:layeredUBrecursionrelation}
\braket{\mathbf{1}, \mathbf{1}}{v_{SS}^{(t_0+n/2)}} \leq \braket{\mathbf{1},\mathbf{1}}{v_{SS}^{(t_0)}}E_n(1+f''n\sigma^2 + n\sigma e^{-\chi_3'(d-d^*)})\,.
\end{equation}
Iterating this equation starting at $d=d^*$, we get
\begin{align}
    \braket{\mathbf{1}, \mathbf{1}}{v_{SS}^{(dn/2)}} &\leq \frac{E_n^{d-d^*}}{q^n+1}\prod_{d'=d^*+1}^{d}(1+f''n\sigma^2 + n\sigma e^{-\chi_3'(d'-d^*)}) \\
    &\leq \frac{E_n^{d-d^*}}{q^n+1}\exp\left(\sum_{d'=d^*+1}^{d}(f''n\sigma^2 + n\sigma e^{-\chi_3'(d'-d^*)})\right) \\
    &\leq \frac{E_n^{d-d^*}}{q^n+1}\exp\left((d-d^*)(f''n\sigma^2) + n\sigma \chi_3''\right)
\end{align}
for some choice of $\chi_3''$ (the exponentially decaying sum is bounded). Now, we note from the definition of $E_n$ that as long as $\sigma \leq O(1/n)$, there is a constant $g$ (slightly larger than 1) such that $E_n \geq \exp(-g n\sigma')$, allowing us to say
\begin{align}
    \braket{\mathbf{1}, \mathbf{1}}{v_{SS}^{(dn/2)}} \leq \frac{E_n^d}{q^n+1}\exp\left(f''n\sigma^2d+ gn\sigma'd^* +  n\sigma\chi_3''\right)\,,
\end{align}
which, recalling the definition of $d^*$ in Eq.~\eqref{eq:dstar}, implies the lemma statement for appropriate choices of $a_0$, $a_1$, and $a_2$. 
\end{proof}

\subsubsection{Proof of \autoref{lem:boundSdestinedmassCG}}

\begin{proof}
In the layered case (proof of \autoref{lem:boundSdestinedmassLAYERS}), we considered the action of all $n/2$ gates in a layer at once. For complete-graph, we can treat each gate individually. Following the layered derivation to Eq.~\eqref{eq:SSprobIntermediate}, for complete-graph we have
\begin{align}
        \braket{\mathbf{1}, \mathbf{1}}{v_{SS}^{(t)}} \nonumber
        ={}&\sum_{\vec{\nu}\neq I^n}\bra{\mathbf{1}} L_S Q'^{(t)}_{\sigma} Q_\sigma^{(t)} L_{S}^{-1}\ket{\vec{\nu}}\bra{\mathbf{1},\vec{\nu}}R_{SS}^{(t)}\ket{v_{SS}^{(t-1)}} \,. \label{eq:SSprobIntermediateCG}
\end{align}
Here the $t$th gate acts on two qudits $i_t$ and $j_t$, but in forming $\ket{v_{SS}^{(t)}}$ from $\ket{v_{SS}^{(t-1)}}$, we take the average over all possible choices of $\{i_t, j_t\}$, as the complete-graph architecture chooses the pair of qudits to act on uniformly at random. After action by $R_{SS}^{(t)}$ the values assigned at position $i_t$ and $j_t$ must be set equal. If they are assigned $S$, then errors can send the new configuration to one of four possible configurations, corresponding to errors on none, one, or both qudits. If they are assigned $I$ then no errors are possible. If we assume $\nu_{i_t} = \nu_{j_t} = S$, then zero errors occurs with probability $(1-\sigma)^2$, one error with probability $2\sigma(1-\sigma)$, and two errors with probability $\sigma^2$. Thus, we have
\begin{align}
    \bra{\mathbf{1}} L_S Q'^{(t)}_{\sigma} Q_\sigma^{(t)} L_{S}^{-1}\ket{\vec{\nu}} 
    ={}& (1-\sigma)^2+ 2\sigma(1-\sigma) \frac{q^{-2n+2|\vec{\nu}|-2}-q^{-2n}}{q^{-2n+2|\vec{\nu}|}-q^{-2n}} +\sigma^2 \frac{q^{-2n+2|\vec{\nu}|-4}-q^{-2n}}{q^{-2n+2|\vec{\nu}|}-q^{-2n}} \\
    ={}&\frac{(1-\sigma')^2 - q^{-2|\vec{\nu}|}}{1 - q^{-2|\vec{\nu}|}}\,,
\end{align}
where $\sigma' = \sigma(1-q^{-2})$. Define the final expression as
\begin{equation}
     J_{w} = \frac{(1-\sigma')^2 - q^{-2w}}{1 - q^{-2w}}\,.
\end{equation}
The quantity $J_w$ is monotonically increasing in $w$ and satisfies $J_w \leq J_n$ for all $w$. 
Meanwhile, if $\nu_{i_t} = \nu_{j_t} = I$, then $\bra{\mathbf{1}} L_S Q'^{(t)}_{\sigma} Q_\sigma^{(t)} L_{S}^{-1}\ket{\vec{\nu}} = 1$. 

Recall the marginal dynamics of $R^{(t)}_{SS}$ on the noisy
copy are simply $P_S^{(t)}$. Suppose the noisy copy starts at a configuration $\ket{\vec{\eta}}$. If $|\vec{\eta}| = w$, then let $\phi_{SS,w}$ be the probability that the qudits $i_t$ and $j_t$ are both assigned $S$, $\phi_{IS,w}$ be the probability one is assigned $S$ and one is assigned $I$, and $\phi_{II,w}$ be the probability both are assigned $I$. 
\begin{align}
    \phi_{SS,w} &= \frac{w(w-1)}{n(n-1)} \\
    \phi_{IS,w} &= \frac{2w(n-w)}{n(n-1)} \\
    \phi_{II,w} &= \frac{(n-w)(n-w-1)}{n(n-1)}\,.
\end{align}
Note that $\phi_{SS,w} + \phi_{IS,w} + \phi_{II,w} = 1$. 
In the case where one is $I$ and one is $S$, the $I$ is flipped to $S$ by $P_S^{(t)}$ with probability $P_{\uparrow,w}$ and the $S$ is flipped to $I$ with probability $P_{\downarrow,w}$, where
\begin{align}
    P_{\uparrow,w} &= \frac{1}{q^2+1}\frac{q^{-2n + 2w + 2}-q^{-2n}}{q^{-2n + 2w}-q^{-2n}} \\
    P_{\downarrow,w} &= 1-P_{\uparrow,w} \,,
\end{align}
which increases or decreases the Hamming weight of $w$ by 1. 
Note the following equalities and inequalities:
\begin{align}
    P_{\downarrow,w} &= \frac{1}{q^2+1}\frac{1-q^{-2w+2}}{1-q^{-2w}} \geq \frac{1}{q^2+1}-q^{-2w} \\
    1-J_w &= \frac{1-(1-\sigma')^2}{1-q^{-2w}} =\frac{2\sigma'-\sigma'^2}{1-q^{-2w}}  \\
J_n-J_w &= \frac{(1-(1-\sigma')^2)(q^{-2w}-q^{-2n})}{(1-q^{-2n})(1-q^{-2w})} \leq \frac{q^{-2w}(2\sigma'-\sigma'^2)}{1-q^{-2w}} \\
\phi_{II,w} + \phi_{IS,w}P_{\downarrow,w} &\geq 
\begin{cases}
\frac{n-w}{n-1}\left(\frac{1}{q^2+1}-q^{-2w}\right)  \geq \frac{n-w}{n-1}\frac{1}{q^2+1} - q^{-2w}& \text{if } w \geq n/2 \\
\frac{1}{4}& \text{if } w < n/2
\end{cases}\,,
\end{align}
where the last inequality follows because, when $w \geq n/2$, $\phi_{IS,w} \geq \frac{n-w}{n-1}$, and when $w < n/2$, $\phi_{II,w} \geq \frac{1}{4}$.

We may now define $G_w$ by the following equation, where $|\vec{\eta}|=w$,
\begin{align}
    G_w ={}& \sum_{\vec{\nu}\neq I^n}\bra{\mathbf{1}} L_S Q'^{(t)}_{\sigma} Q_\sigma^{(t)} L_{S}^{-1}\ket{\vec{\nu}}\bra{\vec{\nu}}P_{S}^{(t)}\ket{\vec{\eta}} \\
    ={}& \phi_{SS,w}J_w + \phi_{IS,w}(P_{\uparrow,w}J_{w+1} + P_{\downarrow,w}) + \phi_{II,w} \,.
\end{align}
We want to lower bound this quantity. If $n=2$, then $G_1=G_2=J_2$. If $n > 2$, we have 
\begin{align}
  G_w \geq{}& \phi_{SS,w}J_w + \phi_{IS,w}(P_{\uparrow,w}J_{w} + P_{\downarrow,w}) + \phi_{II,w} \\
  ={}& J_n + (1-J_w)(\phi_{II,w} + P_{\downarrow,w}\phi_{IS,w}) - (J_n-J_w) \\
    %
    \geq{}& J_n + \frac{2\sigma'-\sigma'^2}{1-q^{-2w}}\begin{cases}
    \frac{n-w}{n-1}\frac{1}{q^2+1}- 2q^{-2w} & \text{if } w \geq n/2 \\ 
    \frac{1}{4}- q^{-2w} & \text{if } w < n/2 
    \end{cases} \,.
\end{align}
By inspection of the final equation, we see that $G_w \geq J_n$ for every combination $n>2$, $w\geq 1$ (since $q>2$) except when $w = n$, but for $w=n$, $G_w=J_n$ by definition, so $G_w \geq J_n$ also holds.

This immediately gives us
\begin{align}
    \braket{\mathbf{1}, \mathbf{1}}{v_{SS}^{(t)}}
={}&\sum_{w=1}^n G_w\sum_{\vec{\eta}:|\vec{\eta}|=w}\braket{\mathbf{1},\vec{\eta}}{v_{SS}^{(t-1)}}  \geq J_n \sum_{w=1}^n \sum_{\vec{\eta}:|\vec{\eta}|=w}\braket{\mathbf{1},\vec{\eta}}{v_{SS}^{(t-1)}} = J_n\braket{\mathbf{1},\mathbf{1}}{v_{SS}^{(t-1)}}\,,
\end{align}
which proves the lower bound by recursion on increasing $t$ and the fact that $\braket{\mathbf{1},\mathbf{1}}{v_{SS}^{(0)}}=1/(q^n+1)$.

To show the upper bound, we first observe
\begin{align}
    G_w \leq J_n + (1-J_n)(\phi_{II,w} + P_{\downarrow,w}\phi_{IS,w})\,.
\end{align}
We have the inequalities
\begin{align}
    1-J_n &= \frac{2\sigma'-\sigma'^2}{1-q^{-2n}} \leq 2\sigma \\
    \phi_{II,w} + P_{\downarrow,w}\phi_{IS,w} &\leq \phi_{II,w} + \frac{1}{2}\phi_{IS,w} = \frac{n-w}{n}\,.
\end{align}
Moreover, there exists a constant $b$ such that $J_n \geq 1/b$ as long as $n \geq 2$ and $\sigma \leq 0.5$. 
and thus
\begin{align}
    G_w \leq J_n(1 + 2b\sigma \frac{n-w}{n})\,.
\end{align}

Similar to the proof of \autoref{lem:boundSdestinedmassLAYERS}, we can split the initial weight into parts for which the noiseless copy has reached the $S^n$ fixed point, and a part that has not. 
\begin{align}
        \braket{\mathbf{1}, \mathbf{1}}{v_{SS}^{(t)}} 
        &= A_{not} + \sum_{w=1}^n G_w A_w\,,
    \end{align}
where 
\begin{align}
    A_{not} &= \sum_{\vec{\eta},\vec{\mu} \neq S^n}G_{|\vec{\eta}|}\braket{\vec{\mu},\vec{\eta}}{v_{SS}^{(t-1)}} \\
    A_{w} &= \sum_{\vec{\eta}:|\vec{\eta}|=w}\braket{S^n,\vec{\eta}}{v_{SS}^{(t-1)}}\,.
\end{align}
Since $G_{|\vec{\eta}|} \leq 1$ by definition, we may directly apply \autoref{lem:ACconvergencetoSn} and bound $A_{not} \leq \eta_{{t-1}}'/(q^n+1)$. 

When $w<n$, we also have
\begin{align}
    A_w &\leq \sum_{\vec{\eta}:|\vec{\eta}|=w}\bra{\vec{\eta}}\Delta\ket{v_{SS}^{(t-1)}} \leq n\sigma (n-w) q^{-(n-w)} \chi_9 e^{-\chi_8(n-w)} \braket{\mathbf{1},\mathbf{1}}{v^{(t-1)}_{SS}}
\end{align}
by \autoref{lem:expclustering}. This requires $\sigma \leq \chi_7/n$ and $n \geq n_0$, so the upper bound inherits these requirements. Meanwhile by definition $\sum_{w=1}^n A_w \leq \braket{\mathbf{1},\mathbf{1}}{v_{SS}^{(t-1)}}$.  

Thus we have
\begin{align}
    \sum_{w=1}^n G_w A_w  & = G_n\sum_{w=1}^n A_w  + \sum_{w=1}^n(G_w-G_n)A_w\\
    &\leq \braket{\mathbf{1},\mathbf{1}}{v_{SS}^{(t-1)}}\left( G_n + \sum_{w=1}^{n-1} (G_w-G_n) n\sigma (n-w) q^{-(n-w)} \chi_9 e^{-\chi_8(n-w)} \right)\\
    &\leq \braket{\mathbf{1},\mathbf{1}}{v_{SS}^{(t-1)}}J_n\left(1 + \sum_{w=1}^{n-1} 2b\sigma \frac{n-w}{n} n\sigma (n-w) q^{-(n-w)} \chi_9 e^{-\chi_8(n-w)} \right)\\
    &\leq \braket{\mathbf{1},\mathbf{1}}{v_{SS}^{(t-1)}}J_n(1 + f\sigma^2)
\end{align}
for some constant $f$, since $\sum_{a=1}^\infty a^2 e^{-ca}$ is bounded by a constant. 

This gives us the recursion relation
\begin{equation}\label{eq:CGrecursionrelation}
    \braket{\mathbf{1},\mathbf{1}}{v_{SS}^{(t)}} \leq \braket{\mathbf{1},\mathbf{1}}{v_{SS}^{(t-1)}}J_n(1 + f\sigma^2) + \frac{\eta'_{t-1}}{q^n+1}\,.
\end{equation}
However, for the first roughly $s_{AC}$ gates, we will use the naive recursion relation $\braket{\mathbf{1},\mathbf{1}}{v_{SS}^{(t)}} \leq \braket{\mathbf{1},\mathbf{1}}{v_{SS}^{(t-1)}}$. We will begin to use Eq.~\eqref{eq:CGrecursionrelation} once $\eta_{t-1}'$ is small. We have
\begin{align}
    \frac{\eta_{t-1}'}{q^n+1} &\leq \frac{\chi_4}{q^n+1} e^{-\chi_3(t-1-s_{AC})/n}\leq \frac{\chi_4}{q^n+1} J_n e^{-\chi_3(t-1-s_{AC})/n-2\log(1-\sigma)} \\
    &\leq \chi_4 J_n \braket{\mathbf{1},\mathbf{1}}{v_{SS}^{(t-1)}} e^{-\chi_3(t-1-s_{AC})/n-2\log(1-\sigma)-2(t-1)\log(1-\sigma)} \\
    &\leq J_n n\sigma \braket{\mathbf{1},\mathbf{1}}{v_{SS}^{(t-1)}} e^{-\chi_3'(t-s^*)/n}\,,
\end{align}
where in the first line we used the fact that $J_n \geq (1-\sigma)^2$, in the second line we invoked \autoref{lem:decayOfvSS}, and in the third line we have defined
\begin{equation}\label{eq:sstar}
    s^* = s_{AC} + n\log(1/n\sigma)/\chi_3' + f'' + n \log(\chi_4)/\chi_3'
\end{equation}
for an appropriate constant $f''$ and a weaker exponential decay coefficient $\chi_3'$. This requires $-2\log(1-\sigma) < \chi_3/n$, which will hold as long as $\sigma \leq b_3/n$ for a properly chosen constant $b_3$. 
This gives us
\begin{equation}
        \braket{\mathbf{1},\mathbf{1}}{v_{SS}^{(t)}} \leq \braket{\mathbf{1},\mathbf{1}}{v_{SS}^{(t-1)}}J_n(1 + f\sigma^2 + n\sigma e^{-\chi_3'(s-s^*)/n})\,.
\end{equation}
Iterating this equation starting at $t=s^*$, and recalling that $\braket{\mathbf{1},\mathbf{1}}{v_{SS}^{(s^*)}} \leq 1/(q^n+1)$, 
\begin{align}
    \braket{\mathbf{1},\mathbf{1}}{v_{SS}^{(t)}} &\leq \frac{J_n^{t-s^*}}{q^n+1}\prod_{t'=s^*+1}^t\left(1 + f\sigma^2 + n\sigma e^{-\chi_3'(t'-s^*)/n}\right) \\
    &\leq \frac{J_n^{t-s^*}}{q^n+1}\exp\left(\sum_{t'=s^*+1}^t\left( f\sigma^2 + n\sigma e^{-\chi_3'(t'-s^*)/n}\right) \right) \\
    &\leq \frac{J_n^{t-s^*}}{q^n+1}e^{ft\sigma^2 + \chi_3'' n \sigma}
\end{align}
for some choice of $\chi_3'' = O(1)$ (the exponentially decaying sum is bounded). Now, we note that $J_n \geq \exp(-g\sigma')$ for a constant $g$ slightly larger than 2 (when $\sigma$ is beneath some constant), allowing us to say
\begin{equation}
    \braket{\mathbf{1},\mathbf{1}}{v_{SS}^{(t)}}  \leq \frac{J_n^{t}}{q^n+1}e^{ft\sigma^2 + \chi_3'' n \sigma + g\sigma's^*}\,,
\end{equation}
which, recalling the definition of $s^*$ in Eq.~\eqref{eq:sstar}, implies the lemma statement for appropriate choices of $b_0$, $b_1$, and $b_2$. Note that the $O(n\sigma)$ term can be collected with the $O(s_{AC}\sigma)$ term since $s_{AC} \geq \Omega(n \log(n))$. 
\end{proof}

\section{Complexity theory of the white-noise sampling problem}\label{app:complexitytheorywhitenoise}

Recent experiments on superconducting qubit devices \cite{Arute2019GoogleQuantumSupremacy, USTC2021StrongQCompAdv,USTC2021Zhuchongzhi2.1} have claimed that the output distribution $\pnoisy$ sampled by their device would be intractable to sample on a classical computer. This claim is motivated by progress in complexity theory on showing that sampling the outputs of quantum computations is hard, but ultimately these claims must rely on conjecture.

The argument that quantum computations should be hard to simulate classically begins with the observation that an efficient classical algorithm for sampling $\pideal$ exactly with probability 1 over choice of $U$ (i.e.~in the worst case) would lead to a contradiction of the widely believed conjecture that the polynomial hierarchy (\PH) does not collapse \cite{Bremner2011IQPCollapse}. The main problem with this result in practice is that noisy quantum devices cannot sample exactly from $\pideal$. It has been conjectured that the task of \textit{approximately} sampling $\pideal$ with high probability over circuit instance cannot be efficiently classically performed, assuming the \PH{} does not collapse. Here ``approximate'' means that the sampled distribution $\pnoisy$ is close to $\pideal$ in total variation distance. Henceforth we refer to this task as \textit{approximate Random Circuit Sampling (RCS)}.

In the following, when we say a task is \PH-hard, we mean that there is a level of the polynomial hierarchy for which granting access to an oracle that performs the task would imply that that level contains the entire \PH. Thus a polynomial time algorithm for the task would imply that the \PH{} is contained within one of its levels and collapses. 
\begin{conjecture}[Approximate RCS is \PH-hard]\label{con:apprRCS}
    There exists a choice of $\varepsilon =O(1)$ and $\delta \geq 1/\poly(n)$ such that the task of sampling from a distribution $\pnoisy$ for which $\frac{1}{2}\lVert \pnoisy-\pideal \rVert_1 \leq \varepsilon$ for at least a $1-\delta$ fraction of random quantum circuit instances is \PH-hard.
\end{conjecture}
This conjecture mirrors similar conjectures for random linear optical networks and random ``instantaneous'' quantum (IQP) circuits in Refs.~\cite{Aaronson2011BosonSampling,Bremner2016AverageCaseIQP}. There is weak evidence for these conjectures in the form of worst-to-average case reductions for \textit{computing} the entries of $\pideal$ with very small error tolerance \cite{Aaronson2011BosonSampling,Bouland2019RCSComplexity,Movassagh2019QSandRQC,Bouland2021NoiseQuantumSupremacy,Kondo2021ImprovedRobustness,Dalzell2020HowManyQubits}, but these results are multiple steps away from proving \autoref{con:apprRCS} because they concern computing probabilities (strong simulation) as opposed to sampling (weak simulation), and furthermore they cannot tolerate errors of size $O(1)$ in total variation distance. 

However, another issue with applying the conjecture in practice is that actual devices are unlikely to be able to sample from a distribution with such small total variation distance from ideal, as doing so requires error rates to be exceedingly small. Sampling from a distribution $\pnoisy$ that is close in total variation distance to $\pwhitenoise$ (for some non-negligible choice of $F$) is potentially much more tractable in the near term; indeed, the experiments from Refs.~\cite{Arute2019GoogleQuantumSupremacy,USTC2021StrongQCompAdv, USTC2021Zhuchongzhi2.1} claim to have performed this task (although note that their random circuits were not Haar random, but rather chosen from some other discrete random ensemble). We refer to this task as \textit{white-noise RCS}. 
\begin{conjecture}[White-noise RCS is \PH-hard]\label{con:whitenoiseRCS}
    There exists a choice of $\varepsilon =O(1)$ and $\delta \geq 1/\poly(n)$ such that whenever the fidelity $F$ satisfies $F \geq 1/\poly(n)$, the task of sampling from a distribution $\pnoisy$ for which $\frac{1}{2}\lVert \pwhitenoise-\pnoisy\rVert_1 \leq \varepsilon F$ for at least a $1-\delta$ fraction of random quantum circuit instances is \PH-hard.
\end{conjecture}
Note that exact worst-case white-noise sampling is \PH-hard (as long as $F$ is at least inverse polynomial). A version of this statement, which further claims that the exact worst-case white-noise task can be at most a factor of $F$ easier for classical computers than the exact worst-case noiseless task, appears in the Supplementary Material of Ref.~\cite{Arute2019GoogleQuantumSupremacy}. However, allowing error of size $\varepsilon F$ was not explicitly considered. Here we show that this is not an issue, and that approximate white-noise RCS and approximate RCS are essentially equivalent in this context, up to a linear factor in $F$, whenever the underlying random quantum circuits have the anti-concentration property.

\begin{theorem}\label{thm:hardnessreduction}
    Consider a random quantum circuit architecture that has the anti-concentration property. That is, there is a constant $z$ such that $\EV_U[\sum_x\pideal(x)^2] \leq z q^{-n}$. Define an oracle $\mathcal{O}$ as follows. On input $(U,b)$, where $U$ is a description of a $n$-qudit circuit with $\poly(n)$ gates drawn randomly from the architecture, and $b$ is a string of $\poly(n)$ uniformly random bits, $\mathcal{O}$ produces an output $x$ from a distribution $\pnoisy$ for which $\frac{1}{2}\lVert \pnoisy - \pwhitenoise \rVert_1 \leq \varepsilon F$ holds for a certain (known) constant $F$ on at least $1-\delta$ fraction of random circuit instances $U$. 
    
    Then, given access to $\mathcal{O}$ and an \NP{} oracle, there is an algorithm with runtime $F^{-1}\poly(n)$ that produces samples from a distribution $p$ for which $\frac{1}{2}\lVert p-\pideal \rVert_1 \leq \varepsilon'$ on at least $1-\delta'$ fraction of circuit instances, with
    \begin{align}
        \varepsilon' &= 4\varepsilon + 1/\poly(n) \\
        \delta' &= \delta + 1/\poly(n)
    \end{align}
\end{theorem}

\begin{corollary}\label{cor:equivalentConjectures}
    For a random quantum circuit architecture with the anti-concentration property, \autoref{con:apprRCS} is true if and only if \autoref{con:whitenoiseRCS} is true. 
\end{corollary}
\begin{proof}[Proof of \autoref{cor:equivalentConjectures}]
 It is straightforward to show that \autoref{con:whitenoiseRCS} implies \autoref{con:apprRCS} simply by reduction from the white-noise RCS task to the approximate RCS task: suppose one could efficiently classically produce samples from a distribution $\pnoisy$ for which $\frac{1}{2}\lVert \pnoisy-\pideal \rVert_1 \leq \varepsilon$. Then, for any choice of $F$, one can design another algorithm that samples from a distribution $\pnoisy'$ by producing a uniformly random output with probability $1-F$ and an output drawn from $\pnoisy$ with probability $F$. Then we have $\frac{1}{2}\lVert \pnoisy'-\pwhitenoise \rVert_1 \leq \varepsilon F$. Thus, whenever approximate RCS can be performed efficiently, white-noise RCS can also be performed efficiently with the same $(\varepsilon, \delta)$ parameters, and if the latter is \PH-hard then the former is also \PH-hard. 
 
 The fact that \autoref{con:apprRCS} implies \autoref{con:whitenoiseRCS} is a direct implication of \autoref{thm:hardnessreduction}.  Given a target $(\varepsilon',\delta')$ pair for which approximate RCS is hard, we can choose $\varepsilon = O(1)$ and $\delta \geq 1/\poly(n)$ such that if a white-noise sampler exists with those parameters, there is also an approximate sampler with parameters $(\varepsilon',\delta')$ that runs in $\poly(n)$ time and requires access to an \NP{} oracle. However, since \NP{} lies within the \PH, this would still imply a collapse of the \PH{} to one of its levels.
\end{proof}

The part of the proof of \autoref{cor:equivalentConjectures} that shows \autoref{con:whitenoiseRCS} implies \autoref{con:apprRCS} also illustrates why a linear factor of $F$ is optimal. To simulate a white-noise output, one need only produce an output from $\pideal$ an $F$ fraction of the time, so producing $T$ samples requires only $FT$ queries to a sampler for $\pideal$. If sampling from $\pideal$ is a hard classical task, sampling from $\pwhitenoise$ is thus at least a factor of $F$ easier. \autoref{thm:hardnessreduction} shows that, in a sense, it is also \textit{at most} a factor of $F$ easier. 

This observation essentially puts the low-fidelity and high-fidelity noise regimes on the same theoretical footing when it comes to hardness of sampling, as long as the fidelity is at least inverse polynomial in $n$. One might object that $F \geq 1/\poly(n)$ is unrealistic in an asymptotic sense, and in many cases, this may be true. However, one way to achieve $F \geq 1/\poly(n)$ is to run circuits with Pauli error rate $\epsilon = \Theta(1/n)$ and circuit size $s = \Theta(n\log(n))$, which, conveniently, is precisely the size required to achieve the anti-concentration property, as shown in Ref.~\cite{dalzell2020anticoncentration}. Moreover, when the fidelity is inverse exponential in $n$ (but larger than $2^{-n}$), there is still a sense in which the low-fidelity regime can be at most a factor of $F$ easier for a classical computer than the high-noise regime.

\begin{proof}[Proof of \autoref{thm:hardnessreduction}]
    
    The idea behind our reduction is to combine approximate rejection sampling with the ability to efficiently estimate $\pnoisy(x)$ up to  $1/\text{poly}(n)$ relative error for any fixed instance $U$ using an $\NP$ oracle (Stockmeyer's approximate counting algorithm \cite{Stockmeyer1983ApproximateCounting}). 
    To be precise, for any $\nu$, any $\mu$, and any $x$, there is a randomized algorithm (with access to $\NP$ oracle) that produces a number, denoted $p'$ such that with probability at least $1-\mu$, 
    \begin{equation}\label{eq:apprestimate}
        |\pnoisy(x) - p'| \leq 2\nu \pnoisy(x), 
    \end{equation}
    and the algorithm runs in time $\nu^{-1} *\poly(n, \log(1/\mu))$. For the linear dependence on $\nu^{-1}$, see the Supplementary Material of Ref.~\cite{Arute2019GoogleQuantumSupremacy} or the lecture notes in Ref.~\cite{Trevisan2002ComputationalComplexityLectureNotes}.  For a fixed $\nu$ and $\mu$, we may take $\mu' = q^{-n}\mu$ and note that $\log(1/\mu') = \poly(n) + \log(1/\mu)$. Now fix a set of random bits $\omega$ to feed into the randomized algorithm above. If we feed the same bits $\omega$ for every choice of $x$ with parameters $\nu$ and $\mu'$, then we have a fixed set of outputs $\pnoisy'(x)$ for each possible $x$, and by the union bound, these values satisfy
    \begin{equation}
        |\pnoisy(x) - \pnoisy'(x)| \leq 2\nu \pnoisy(x)
    \end{equation}
    for every $x$ simultaneously with probability at least $1-\mu$ over the choice of $\omega$. 
    On any particular $x$, the algorithm still runs in time $\nu^{-1}\poly(n, \log(1/\mu))$.
    When this is the case,
    \begin{equation}
        \frac{1}{2}\lVert \pnoisy(x) - \pnoisy'(x)\rVert_1 \leq \nu \,.
    \end{equation}
    Also, let
    \begin{equation}
        \overline{\pideal}(x) = 
         \frac{\pnoisy(x)-(1-F)q^{-n}}{F}
    \end{equation}
    and
    \begin{equation}
        \overline{\pideal}'(x) = \begin{cases}
         \frac{\pnoisy'(x)-(1-F)q^{-n}}{F} & \text{if } \pnoisy'(x) > (1-F)q^{-n}\\
         0 & \text{otherwise}
        \end{cases} \,,
    \end{equation}
    so that, as long as the instance $U$ is among the $1-\delta$ fraction for which $\mathcal{O}$ succeeds, the following hold:
    \begin{align}
        \frac{1}{2}\lVert \overline{\pideal} - \pideal \rVert_1 &\leq \varepsilon \\
        \frac{1}{2}\lVert \overline{\pideal} - \overline{\pideal}' \rVert_1 &\leq \nu/F \label{eq:distpidealpidealprime} \,,
    \end{align}
    and by the triangle inequality
    \begin{equation}
        \frac{1}{2}\lVert \pideal - \overline{\pideal}' \rVert_1 \leq \nu/F + \varepsilon \,.
    \end{equation}
    Note that in general the function $\overline{\pideal}'$ as defined does not describe a probability distribution since it is not necessarily normalized. 
    
    Now let $k > 1$ and consider the following approximate rejection sampling algorithm, similar to that in the Supplementary Information of Ref.~\cite{Neville2017ClassicalBosonSampling}. 
    \begin{enumerate}
        \item Choose a set of random bits $\omega$, which implicitly determines a function $\pnoisy'$.
        \item Choose an $x$ uniformly at random, and use the estimation algorithm with bits $\omega$ to produce $\pnoisy'(x)$, from which $\overline{\pideal}'(x)$ can be determined.
        \item Generate a random real number $0 \leq \eta \leq 1$
        \item If $ \overline{\pideal}'(x) \leq 2k q^{-n}$ and if $\eta \leq \overline{\pideal}'(x)q^n/(2k)$, output $x$ (accept); otherwise, return to step 2 (reject). 
    \end{enumerate}. 
    
    Following the observations in Ref.~\cite{Neville2017ClassicalBosonSampling}, we first analyze the output distribution, denoted by $p_\omega$, of the above algorithm for a certain choice of $\omega$ in step 1. We see that  $p_\omega$ is precisely the distribution $\overline{\pideal}'$ conditioned on $x \in W$ where $W$ is the set of $x$ for which $ \overline{\pideal}'(x) \leq 2k q^{-n}$. Define \begin{align}
        \mathcal{M} &= \sum_x \overline{\pideal}'(x) \\
        \mathcal{N} &= \sum_{x\in W} \overline{\pideal}'(x) \,.
    \end{align}
    Then,
    \begin{equation}
        p_\omega(x) =
        \begin{cases} 
        \mathcal{N}^{-1}\overline{\pideal}'(x) & \text{if } x \in W \\
        0 & \text{otherwise}
        \end{cases} \,.
    \end{equation}
    Hence,
    \begin{align}
        \frac{1}{2} \lVert p_\omega - \overline{\pideal}'(x) \rVert_1 &= \frac{1}{2}\sum_{x \in W} |\mathcal{N}^{-1}\overline{\pideal}'(x)-\overline{\pideal}'(x)| + \frac{1}{2}\sum_{x\not\in W} \overline{\pideal}'(x)\\
        &= \frac{1}{2}|1-\mathcal{N}| + \frac{1}{2}(\mathcal{M}-\mathcal{N})\\\
        &\leq \frac{1}{2}|1-\mathcal{M}| + (\mathcal{M}-\mathcal{N}) \,.
    \end{align}
    Note that $|1-\mathcal{M}| \leq 2\nu/F$ is an implication of Eq.~\eqref{eq:distpidealpidealprime}. Also note that the values of $\overline{\pideal}$ sum to 1 (although some can in principle be negative). To handle the quantity $\mathcal{M}-\mathcal{N} = \sum_{x \not \in W} \overline{\pideal}'(x)$, we invoke \autoref{lem:TVDthreshT}, with $p_1 = \overline{\pideal}'$, $p_2 = \pideal$ and $T = 2kq^{-n}$. It shows that 
    \begin{equation}
        \mathcal{M}-\mathcal{N} \leq 4\varepsilon + 4\nu/F + \sum_{x: \pideal(x) > kq^{-n}} \pideal(x) \,,
    \end{equation} 
    and thus
    \begin{equation}
        \frac{1}{2} \lVert p_\omega - \overline{\pideal}'(x) \rVert_1 \leq 5\nu/F + 4\varepsilon + \sum_{x: \pideal(x) > kq^{-n}} \pideal(x) \,.
    \end{equation}
    This is progress because the right-hand side only has dependence on the ideal distribution $\pideal$, and not the approximate distribution output by the estimator. 
    
    Now, recall that we assume that $\EV_U[\sum_x \pideal(x)^2] \leq z q^{-n}$. By Markov's inequality, for any $z'$, $\sum_x \pideal(x)^2 \leq z' q^{-n}$ for at least $1-z/z'$ fraction of instances $U$. Suppose we have such an instance. Then 
    \begin{align}
        \sum_{x: \pideal(x) > kq^{-n}} \pideal(x) &= \sum_{x: \pideal(x) > kq^{-n}} \frac{\pideal(x)^2}{\pideal(x)} \leq  \sum_{x: \pideal(x) > kq^{-n}} \frac{\pideal(x)^2}{kq^{-n}} \leq z'/k \,.
    \end{align}
    
    We conclude that the algorithm produces outputs from a distribution $p_\omega$ for which
    \begin{equation}
        \frac{1}{2}\lVert p_\omega - \pideal \rVert_1 \leq 5\nu/F + 4\varepsilon + z'/k
    \end{equation}
    (with probability at least $1-\mu$ over its internal randomness) and succeeds on at least $1-\delta'$ fraction of circuit instances, where 
    \begin{equation}
        \delta' = \delta+ z/z' \,.
    \end{equation}
    The $\delta'$ fraction of failed instances arise either because the underlying white-noise sampler also fails on those instances or because the output distribution is not sufficiently anti-concentrated. Either way, whether an instance is among this $\delta'$ fraction is independent of the  choice of $\omega$. Thus, we may note that in the $\mu$ chance that the total variation distance bound is not satisfied for the random choice of $\omega$, it will be at most its maximal value of 1, and thus, for any of the $1-\delta'$ successful instances, the overall total variation distance of the sampler is at most $\varepsilon'$, where
    \begin{align}
        \varepsilon' &= 5\nu/F + 4\varepsilon + z'/k + \mu \,.
    \end{align}
    
    Now, we analyze the algorithm's runtime. Each random choice of $x$ and subsequent calculation of $\overline{\pideal}'(x)$ takes at most $\nu^{-1}\poly(n, \log(1/\mu))$ time, but sometimes this step must be repeated.  Each time the algorithm returns to step 2, it will end up accepting on step 4 with probability $\mathcal{N}/2k$. By the above analysis, 
    \begin{align}
        |\mathcal{N}-1| \leq |\mathcal{M}-1| + (\mathcal{M}-\mathcal{N}) \leq 4\varepsilon + 6\nu/F + z'/k \,.
    \end{align}
    Thus, as long $4\varepsilon + 6\nu/F + z'/k \leq 1/2$, then the acceptance probability will be at least $1/4k$, and the expected number of repetitions required to produce an output is at most $4k$. 
    
    Recall that $z = O(1)$. Then we may choose $z' = \poly(n)$ sufficiently large, $k = \poly(n)$ even larger, $\nu^{-1} = F^{-1}*\poly(n)$ sufficiently large, and $\mu^{-1} = \poly(n)$ sufficiently large that the algorithm runs in expected\footnote{To make the runtime bounded, we could impose a cap on the number of times the algorithm returns to step 2 of $4k\cdot \text{polylog}(n)$ which, if hit, results in a uniformly random output. This would increase the total variation distance $\varepsilon'$ by only $1/\poly(n)$ and can thus be ignored.} time $F^{-1}\poly(n)$ and solves the approximate RCS task with parameters $\varepsilon' = 4\varepsilon + 1/\poly(n)$ and $\delta' = \delta + 1/\poly(n)$. It is likely the factor of 4 could be optimized.  
\end{proof}

\begin{lemma}\label{lem:TVDthreshT}
    Suppose $p_1$ and $p_2$ are two real functions on $[q]^n$ for which 
    \begin{equation}
        \frac{1}{2}\lVert p_1-p_2 \rVert_1 \leq \varepsilon \,.
    \end{equation}
    Let $\mathbf{1}(\cdot)$ be the indicator function. Then for any threshold $T>0$, we have
    \begin{align}
        \sum_x p_1(x) \mathbf{1}(p_1(x) > T) \leq 4\varepsilon + \sum_x p_2(x) \mathbf{1}(p_2(x) > T/2) \,.
    \end{align}
\end{lemma}
\begin{proof}
Let $A_1$ be the subset of $[q]^n$ for which $p_1(x) > T$, $A_2$ be the subset for which $p_2(x) > T$, and $A_3$ be the subset for which $p_2(x) > T/2$. For a subset $X$ let $\overline{X}$ denote its complement.
    \begin{align}
        \sum_x p_1(x) \mathbf{1}(p_1(x) > T)={}& \sum_{x \in A_1} p_1(x)
         ={} \sum_{x \in A_1} (p_1(x)-p_2(x))  + \sum_{x \in A_1} p_2(x)\\
         \leq{}& 2\varepsilon + \sum_{x \in A_1} p_2(x)\\ 
         ={}& 2\varepsilon +  \sum_{x \in A_1 \cap \overline{A}_3} p_2(x)+\sum_{x \in A_1 \cap A_3} p_2(x) \\
         \leq{}& 2\varepsilon+ \sum_{x \in A_1 \cap \overline{A}_3} p_2(x)+ \sum_{x \in A_3} p_2(x) \\
         \leq{}& 2\varepsilon + (T/2)|A_1 \cap \overline{A}_3|+ \sum_{x \in A_3} p_2(x) \\
         \leq{}& 2\varepsilon + (T/2)\frac{2\varepsilon}{T/2}+ \sum_{x \in A_3} p_2(x) \label{eq:1234}\\
         ={}& 4\varepsilon + \sum_x p_2(x) \mathbf{1}(p_2(x) > T/2)\,,
    \end{align}
    where the second-to-last line follows because any element of $A_1 \cap \overline{A}_3$ must contribute at least $T/2$ toward the $2\varepsilon$ total allowed deviation between the two functions. 
\end{proof}

\bibliographystyle{utphys}
\bibliography{references}

\end{document}